\documentclass[a4paper,11pt]{article}
\usepackage{stmaryrd}
\usepackage{amsfonts}
\usepackage{bbm}
\usepackage{float} 
\usepackage{float}
\usepackage{amscd}
\usepackage{mathrsfs}
\usepackage{latexsym,amssymb,amsmath,amscd,amscd,amsthm,amsxtra,xypic}
\usepackage[dvips]{graphicx}
\usepackage[utf8]{inputenc}
\usepackage[T1]{fontenc}
\usepackage{lmodern}
\usepackage{amssymb}
\usepackage[all]{xy}
\usepackage{nicefrac,mathtools,enumitem}
\usepackage{microtype}
\usepackage{lipsum}
\usepackage{mathtools}
\usepackage{xfrac}

\usepackage{geometry}
\usepackage{hyperref}

\usepackage{times}
\usepackage{amsthm}
\usepackage{amsmath}
\usepackage{amsfonts}
\usepackage{amssymb}
\usepackage{mathdots}
\usepackage{color}
\usepackage{mathrsfs}
\usepackage{stmaryrd}
\SetSymbolFont{stmry}{bold}{U}{stmry}{m}{n}

\usepackage{bm}
\usepackage{tikz}
\usetikzlibrary{calc,math}

\usepackage{tikz-cd}
\usepackage{tikz}
\usetikzlibrary{matrix,arrows,decorations.pathmorphing}
\usepackage[all]{xy}
\usepackage{graphicx}
\usepackage{subfigure}

\usepackage[toc,page]{appendix}
\usepackage[usestackEOL]{stackengine}

\usepackage{theoremref}
\newtheorem{thm}{Theorem}[section]
\newtheorem{lem}[thm]{Lemma}
\newtheorem{cor}[thm]{Corollary}
\newtheorem{pro}[thm]{Proposition}

\theoremstyle{definition}
\newtheorem{ex}[thm]{Example}
\newtheorem{rmk}[thm]{Remark}
\newtheorem{defi}[thm]{Definition}
\newtheorem*{assumption*}{Assumption}

\newcommand{\be }{\begin{equation}}
\newcommand{\ee }{\end{equation}}

\newcommand{\mathi}{\mathrm{i}}

\newcommand{\var}{\varepsilon}

\newcommand{\h}{\mathfrak h}

\newcommand{\br}[1]{   [ \cdot,    \cdot  ]   }

\newcommand {\Sect}{{\rm Sect}}

\newcommand{\D}{\mathbb{D}}

\newcommand {\IC}{\mathbb{C}}

\newcommand{\C}{\mathbb C}
\newcommand{\End}{\operatorname{End}}

\newcommand{\ad}{\operatorname{ad}}

\newcommand {\Uph}{\mathfrak{U}_{p,\hbar}}  
\newcommand {\wUph}{\widehat{\mathfrak{U}}_{p,\hbar}}

\newcommand{\I}{{\mathrm{i}}}
\usepackage{pgfplots}
\pgfplotsset{compat=1.18}
\title{Quantum Stokes matrices and quantum Riemann-Hilbert-Birkhoff maps}
\author{Xiaomeng Xu}
\date{}

\newcommand{\Addresses}{{
  \bigskip
  \footnotesize
\noindent \textsc{
School of Mathematical Sciences \& Beijing International Center
for Mathematical Research, Peking University, Beijing 100871, China}\par\nopagebreak
  \textit{E-mail address}: \texttt{xxu@bicmr.pku.edu.cn}
}}
\begin{document}

\maketitle
    
\begin{abstract}
In this paper, we define quantum Stokes matrices of meromorphic linear systems of ordinary differential equations with noncommutative coefficients and a pole of order $p+1$. We prove that these quantum Stokes
matrices satisfy natural quantum exchange relations. These relations allow us to interpret the quantum Stokes matrices as an associative algebra homomorphism,
which may be viewed as a quantization of the
Riemann-Hilbert-Birkhoff map, regarded as a Poisson map, for
meromorphic connections.
\end{abstract}

\section{Introduction}

The Stokes phenomenon describes the fact that a solution of a meromorphic
linear ordinary differential equation may have different asymptotic expansions
as $z$ approaches an irregular singular point from different sectors. We refer to \cite{Balser, LR, Wa} for detailed introductions
to the summability theory and Stokes phenomenon. The jumps between these sectorial asymptotic
expansions are measured by Stokes matrices.

Over the past decades, Stokes matrices for meromorphic
linear systems with poles of order two have played important roles in many areas of mathematics and
physics. However, they are highly transcendental objects and are
therefore difficult to study directly. In \cite{Xu1, Xu, TX} we relate the Stokes phenomenon, WKB approximation
and isomonodromy deformations to algebraic structures governed by the Yang–Baxter equation, using which we study various analysis problems. The
purpose of this paper is to extend this relation to higher order poles.
In particular, we
derive quantum exchange relations hidden behind the Stokes phenomenon for
poles of arbitrary order, with the hope that the representation theory of the associated quantum algebras can be used to understand various asymptotic analysis problems for general meromorphic linear systems. As a byproduct, we obtain a quantization of the
Riemann-Hilbert-Birkhoff (RHB) map for
meromorphic connections with a pole of order $p+1$. Here the classical RHB map is regarded as a Poisson morphism between two algebraic Poisson varieties, the de Rham and Betti moduli spaces. 

\subsection*{Main theorem: the exchange relations for quantum Stokes matrices}
Let $p>1$ be an integer, and $u_1,\ldots,u_n$ be $n$ distinct complex numbers. 
Let $\Uph$ be the associative
$\mathbb C[\hbar]$-algebra generated by
\[
  \hbar e_{ij},\qquad \hbar e_{ij}^{(a)},
  \qquad 1\le i,j\le n,\quad 1\le a\le p-1,
\]
subject to the relations for all $i,j,k,l=1,\ldots,n$ and $a,b=1,\ldots,p-1$
\begin{align}\label{algebra1}
[ \hbar e_{ij}^{(a)},\hbar e_{kl}^{(b)}]&=\left\{\begin{array}{lr}  \hbar (\delta_{jk} \hbar e_{il}^{(a+b)}-\delta_{li}  \hbar e_{kj}^{(a+b)}),  & \text{ if } a+b\le p-1 \\
\hbar\delta_{jk}\delta_{il}(u_i-u_j), & \text{ if } a+b=p,
 \\
0, & \text{ if } a+b\ge p+1.
             \end{array}
\right.
\\ \label{algebra2}
[\hbar e_{ij},\hbar e_{kl}]&=\hbar(\delta_{jk}\hbar e_{il}-\delta_{li} \hbar e_{kj}), \quad \text{and} \quad [\hbar e_{ij}^{(a)},\hbar e_{kl}]=0.
\end{align}
We introduce a completion $\wUph$ of this algebra in which algebra valued holomorphic functions can be defined. For more details, see Section \ref{beginsec1}.

Let $\h_{\rm reg}$ denote the set of $n\times n$ diagonal matrices with distinct eigenvalues. Given $u\in\h_{\rm reg}$, let us consider the linear differential equation for a holomorphic function $F(z)\in {\wUph} {\otimes} \End(\C^n)$
\begin{equation}\label{hqeq}
\frac{dF}{dz} = \left(\frac{u}{z^{p+1}}+\frac{\hbar T_{[p]}}{z^p}+\cdots+\frac{\hbar T_{[2]}}{z^2}-\frac{\hbar T}{z}\right)\cdot F,\end{equation}
where $u={\rm diag}(u_1,..,u_n)$ is regarded as an $n\times n$ matrix with scalar entries, and $T, T_{[m]}$ are $n\times n$ matrices with entries 
\begin{align*}
(T_{[m]})_{ij}&= e_{ij}^{(m-1)}, \ \ \  \text{for} \ 1\le i,j\le n, \ \ 2\le m\le p,\\
(T)_{ij}&= e_{ij}, \ \ \ \ \  \text{for} \ 1\le i,j\le n.
\end{align*}
Thus $u, \hbar T$ and $\hbar T_{[m]}$ are elements in $\Uph\otimes \End(\C^n)$ and the product $\cdot$ in \eqref{hqeq} is the multiplication in ${\wUph} {\otimes} \End(\C^n)$. 

Equation \eqref{hqeq} has a formal solution $\widehat{F}$ at
$z=0$ (see Theorem \ref{oneformal}). Standard resummation theory (see e.g., \cite{Balser, LR, Wa}) gives sectorial regions around $z=0$, on each of which there is a unique (therefore canonical) holomorphic solution with the prescribed formal asymptotic expansion. The transition matrices between these solutions are encoded by $2p$ (quantum) Stokes matrices $S_{1}, ..., S_{2p}\in\wUph\otimes{\rm End}(\mathbb{C}^n)$. They depend on the choice of an initial anti-Stokes direction, and after a simultaneous permutation determined by the direction and the ordering
of $u_1,\ldots,u_n$, are unipotent triangular matrices. See Section \ref{q-Stokes} for more details. 

\begin{thm}\label{thmRLL}
For $p>1$ and any $u\in\h_{\rm reg}$, choose the initial anti-Stokes ray and the ordering of $u_1,\ldots,u_n$ so that
$S_{2k+1}(u)$ and $S_{2k}(u)$ are lower and upper triangular, respectively. With this convention, the Stokes matrices satisfy the commutation relations (as identities in $\wUph\otimes  {\rm End}(\mathbb{C}^n)\otimes {\rm End}(\mathbb{C}^n)$)
\begin{align*}
e^{\pi\mathi \hbar\delta P}S^{(1)}_i(u)e^{-\pi\mathi \hbar\delta P}S_{i+k}^{(2)}(u)&=S^{(2)}_{i+k}(u)e^{\pi\mathi \hbar\delta P}S^{(1)}_i(u)e^{-\pi\mathi \hbar\delta P}, \ \ \text{for any $i$ and } 1<k<2p-1,\\
R^{12}S^{(1)}_i(u)e^{-\pi\mathi \hbar\delta P}S^{(2)}_i(u) &=S^{(2)}_i(u)e^{-\pi\mathi \hbar\delta P}S^{(1)}_i(u) R^{12}, \ \ \text{for $i$ odd} \\ 
R^{12}S^{(2)}_i(u)e^{-\pi\mathi \hbar\delta P}S^{(1)}_i(u) &=S^{(1)}_i(u)e^{-\pi\mathi \hbar\delta P}S^{(2)}_i(u)R^{12}, \ \ \text{for $i$ even} \\
S^{(2)}_{i-1}(u)e^{-\pi\mathi \hbar\delta P} S^{(1)}_{i}(u)e^{\pi\mathi \hbar\delta P} &=e^{-\pi\mathi \hbar\delta P}S^{(1)}_{i}(u) R^{12} S^{(2)}_{i-1}(u), \ \ \text{for $i$ odd},\\
S^{(1)}_{i-1}(u)e^{-\pi\mathi \hbar\delta P}S^{(2)}_{i}(u)e^{\pi\mathi \hbar\delta P} &=e^{-\pi\mathi \hbar\delta P}S^{(2)}_{i}(u) R^{12} S^{(1)}_{i-1}(u), \ \ \text{for $i$ even},\\
e^{2\pi \mathi D_1^{(1)}}S_i^{(2)}(u)e^{-2\pi \mathi D_1^{(1)}}
 &=
 e^{-2\pi \mathi \hbar\delta P}
 S_i^{(2)}(u)
 e^{2\pi \mathi \hbar\delta P}.
\end{align*}
Here the matrix $R\in {\rm End}(\IC^n)\otimes {\rm End}(\IC^n)\llbracket\hbar\rrbracket$ is defined by (with $E_{ij}$ the matrix units)
\begin{equation}\label{sRmatrix}
    R=\sum_{i\ne j, i,j=1}^n E_{ii}\otimes E_{jj}+e^{\pi\mathi \hbar}\sum_{i=1}^n E_{ii}\otimes E_{ii}+(e^{{\pi\mathi \hbar}}-e^{-{\pi\mathi \hbar}})\sum_{1\le j<i\le n}E_{ij}\otimes E_{ji},
\end{equation} 
and $\delta P=\sum_{k=1}^n 1\otimes E_{kk}\otimes E_{kk}$. We take the convention $S_{i+2p}(u):=e^{-2\pi\I D_1}S_{i}(u) e^{2\pi\I D_1}$, where $D_1$ is the residue term in the diagonal irregular type of the formal solution in Theorem \ref{oneformal}, and
\[S^{(1)}_{k}(u):= \sum_{i,j}S_{k}(u)_{ij}\otimes E_{ij}\otimes 1, \ S^{(2)}_{k}(u):= \sum_{i,j}S_{k}(u)_{ij}\otimes 1\otimes E_{ij}, \text{ and } \ R^{12}:=1\otimes R\] 
as elements in $\wUph\otimes  {\rm End}(\mathbb{C}^n)\otimes {\rm End}(\mathbb{C}^n)$.
\end{thm}
The tensor product identities above encode quadratic $q$-exchange relations among the matrix entries of the
Stokes matrices in $\wUph$, where $q=e^{\pi\mathi\hbar}$. For general $u_1,\ldots,u_n$, the Stokes matrices $S_i(u)$ may not be triangular. After the corresponding
simultaneous permutation, the same relations in Theorem \ref{thmRLL} still
hold. 
\begin{rmk}
The matrix $R$ is the standard $R$-matrix for the quantum group $U_\hbar(\frak{gl}_n)$, see \cite{Jimbo2}\cite{FRT}. For $p=1$, the same construction gives the FRT realization of the quantum group $U_\hbar(\mathfrak{gl}_n)$. See \cite{Xu}, and \cite{LWX} for an analog in the super setting.
\end{rmk}

To prove Theorem \ref{thmRLL}, in Section \ref{SPDE} we study the Stokes phenomenon of the linear PDE system \eqref{kKZ1}-\eqref{kKZ2} with variables $z_1,z_2$. This system is a Knizhnik-Zamolodchikov (KZ) type equation with poles of order $p+1$. In Section \ref{SPDE}, we construct the formal solution $\widehat{Y}(z_1,z_2)$ of the system, and its factorization at $z_2/z_1=0, 1$ and $\infty$. In Section \ref{sec34}, we study the resummation of the formal solution. The resulting holomorphic solutions depend on the summation directions of $z_1,z_2$ and regions. 
In Corollary \ref{mainpro}, we obtain the factorization of these holomorphic solutions, via the factorization of the underlying formal solution, at $z_2/z_1=0,1$ and $\infty$. Since the local singular behavior of the holomorphic solutions are encoded by the second factors in the factorizations, the connection formula between these solutions are given by the Stokes matrices of the second factors, which in turn can be expressed by the standard $R$-matrix and quantum Stokes matrices $S_i(u)$. Thus, we derive the explicit connection formula or Stokes matrices for the solutions of the system \eqref{kKZ1}-\eqref{kKZ2}. In Section \ref{proofThm}, by comparing the connection formula along homotopy paths in the space $z_1,z_2\in (\mathbb{C}\times \mathbb{C})\setminus\{z_1=0,z_2=0,z_1=z_2\}$, we finally prove the identities in Theorem \ref{thmRLL}. 
Although a general analytic framework for the asymptotic
analysis and Stokes phenomenon of integrable meromorphic connections
in several variables is available, see \cite{Sab}, there are still relatively few explicit studies of concrete PDE systems with higher order poles, including their factorization properties and connection formulas between different asymptotic zones. The analysis developed here is therefore also of independent interest.

\subsubsection*{Motivation and Quantum Riemann-Hilbert-Birkhoff maps}

One motivation for studying the algebraic structures underlying the (quantum) Stokes matrices is to solve the WKB approximation problem of meromorphic linear differential equations. For instance, in the case $p = 1$, the algebraic structures allow us to show that the WKB approximation is governed by the crystal structures of quantum groups, see \cite{Xu1}. Another motivation is to  clarify the Poisson geometric nature of the Riemann–Hilbert–Birkhoff (RHB) maps, as explained in the following.

The Poisson geometric nature of the RHB maps was developed in a series of works \cite{Boalch1, Boalch2, Boalch4}, see also Section \ref{QRHBmap} for a brief review. We interpret Theorem \ref{thmRLL} as a quantization of the classical RHB maps. 

Consider the differential equations for a function $f(z)\in {\rm GL}(n,\mathbb{C})$
\begin{equation}\label{heq}
\frac{df}{dz} = \left(\frac{u}{z^{p+1}}+\frac{A_p}{z^p}+\cdots+\frac{A_2}{z^2}-\frac{B}{z}\right)\cdot f,\end{equation}
where $u={\rm diag}(u_1,\ldots,u_n)\in \h_{\rm reg}$, and $A_i, B\in\frak{gl}_n$ for $i=2,...,p$. 
For fixed $u\in \h_{\rm reg}$, the moduli space of differential equations \eqref{heq} can be identified with the product space 
\begin{equation}
{\mathcal M}_{dR}^{(p)}(u)=\{(A_p,...,A_2,B)\in \frak{gl}_n\times \cdots \times \frak{gl}_n\}.
\end{equation}
The
classical equation \eqref{heq} admits a formal solution
\[
\widehat f(z)
=
\left(1+h_1z+h_2z^2+\cdots\right)\cdot
\exp\int
\left(
\frac{u}{z^{p+1}}
+\frac{d_p}{z^p}
+\cdots
+\frac{d_1}{z}
\right)dz,
\]
where $d_p,\ldots,d_1$ are diagonal matrices. The entries of the matrices
$h_m$ and $d_r$ are polynomial functions on ${\mathcal M}_{dR}^{(p)}(u)$.
After the resummation (order-$p$ Borel-Laplace transform), the classical formal solution gives sectorial holomorphic
solutions and hence classical Stokes matrices $(s_1,\ldots,s_{2p})$. These Stokes matrices are holomorphic matrix functions of the coefficients of \eqref{heq}. Then
the classical Riemann-Hilbert-Birkhoff map 
\[\nu(u): (A_p,...,A_2,B)\in {\mathcal M}_{dR}^{(p)}(u)\rightarrow (s_1,...,s_{2p},d_1)\in\mathcal{M}_B^{(p)}\cong (U_-\times U_+)^{p}\times \h\]
associating the classical Stokes matrices to the differential equation \eqref{heq}, is an analytic Poisson map. See Theorem \ref{kthm}. Here $U_\pm$ are the space of unipotent upper and lower triangular matrices.

On the
quantum side, 
the quantum equation \eqref{hqeq} has an analogous formal solution (see Theorem \ref{oneformal})
\[
\widehat F(z)
=
\left(1+H_1z+H_2z^2+\cdots\right)
\cdot \exp\int
\left(
\frac{u}{z^{p+1}}
+\frac{D_p}{z^p}
+\cdots
+\frac{D_1}{z}
\right)dz,
\]
where $D_p,\ldots,D_1$ are diagonal matrices with entries in $\Uph$, and
$H_m\in \Uph\otimes\End(\C^n)$. The same resummation procedure produces sectorial holomorphic solutions and hence quantum Stokes matrices with entries in the completed algebra $\wUph$.

The semiclassical limits of $\Uph$ and $\wUph$ induce the Poisson brackets on the polynomial coordinate ring of ${\mathcal M}_{dR}^{(p)}(u)$ and on its completed local ring $\widehat{\mathcal{O}}_{{\mathcal M}_{dR}^{(p)}(u)}$ at $0\in {\mathcal M}_{dR}^{(p)}$. Taking the semiclassical limit at the level of formal solutions gives 
\[{\rm scl}: H_k\in\Uph\otimes \End(\C^n)\rightarrow h_k\in {{\mathcal O}}_{{\mathcal M}_{dR}^{(p)}}{\otimes} \End(\C^n)\]
and at the level of holomorphic solutions and Stokes matrices it gives
\[{\rm scl}: S_i\in\wUph\otimes \End(\C^n)\rightarrow s_i\in {\widehat{\mathcal O}}_{{\mathcal M}_{dR}^{(p)}}\widehat{\otimes} \End(\C^n).\]

In Section \ref{PoissonStokes}, we define an associative algebra
$U_{\hbar}^{(p)}(R)$, whose generators are assembled into triangular matrices $L_1,...,L_{2p}$ and a diagonal matrix $K$, subject to the relations in Theorem \ref{thmRLL}. Its semiclassical limit induces the Poisson bracket on the completed local coordinate ring 
$\widehat{\mathcal O}_{\mathcal M_B^{(p)}}$, at $(s_1,\ldots,s_{2p},d_1)=(1,\ldots,1,0)$. As a consequence of Theorem \ref{thmRLL}, in Section \ref{qRHBmap} we show that
\begin{thm}
    The quantum Stokes matrices define
an algebra homomorphism
\[
  \nu_\hbar(u):U_{\hbar}^{(p)}(R)\longrightarrow \wUph~;~ (L_k)_{ij}\mapsto S_k(u)_{ij}, \ K_a\mapsto e^{2\pi\mathi (D_1)_{aa}}, \qquad k=1,\ldots,2p, \quad a=1,\ldots,n,
\]
which maps the matrix generators $L_k, K$ to $S_k, e^{2\pi\mathi D_1}$. Its reduction modulo $\hbar$ coincides with the pullback induced by the classical analytic RHB map on the completed local rings. That is we have a commutative diagram
\[
\begin{CD}
U_{\hbar}^{(p)}(R)
@>{\nu_\hbar(u)}>>
\wUph
\\
@VV{\rm scl}V
@VV{\rm scl}V
\\
\widehat{\mathcal O}_{\mathcal M_B^{(p)}}
@>{\nu(u)^*}>>
 \widehat{\mathcal O}_{{\mathcal M}_{dR}^{(p)}(u)}.
\end{CD}
\]
\end{thm}

Since the semiclassical limit of the associative algebra homomorphism $\nu_\hbar(u)$ yields a Poisson map, it interprets the Poisson geometric nature of the classical RHB map. See Section \ref{PoissonRHB}. In this sense, the algebra homomorphism $\nu_\hbar(u)$, or equivalently the quantum Stokes matrices, is a quantization of the highly transcendental RHB map as a Poisson map.

\subsubsection*{Irregular  Knizhnik-Zamolodchikov equations}
The irregular KZ equation was first introduced in \cite{Re}, and was given a
representation theoretic interpretation in \cite{FR}. Its one variable specialization has the form
\begin{equation}\label{hqeqKZ}
\frac{dF}{dz} = \left(\frac{\hbar E_{[p+1]}}{z^{p+1}}+\frac{\hbar E_{[p]}}{z^p}+\cdots+\frac{\hbar E_{[2]}}{z^2}+\frac{\hbar E_{[1]}}{z}\right)\cdot F,
\end{equation}
where $E_{[m]}$ are the matrix valued generators of the truncated current (Takiff) algebra $
  \frak{gl}_n\llbracket t\rrbracket/
  t^{p+1} \frak{gl}_n\llbracket t\rrbracket.$

The equation \eqref{hqeq} (and the linear PDE system \eqref{kKZ1}-\eqref{kKZ2} with two variables) studied in this paper is related to the irregular KZ equations, but its coefficient algebra comes from the modified loop algebra \cite{RSTS}, see Remark \ref{modiloop}. In other words, the leading irregular term of the equation is fixed to be $u\in\h_{\rm reg}$, and $\Uph$ is a quantization of the Poisson structure on the space ${\mathcal M}_{dR}^{(p)}(u)$ of meromorphic connections with that leading term $u\in\h_{\rm reg}$. In particular, the decoupled relation $[\hbar e_{ij}^{(a)},\hbar e_{kl}]=0$ in \eqref{algebra2} is the quantum counterpart of the decoupled
Poisson bracket on ${\mathcal M}_{dR}^{(p)}(u)$ (arising from Lemma 2.4 in \cite{Boalch2}, or Proposition 14(2) in \cite{Boalch4}).
In some sense, our equation
may be viewed as a regular semisimple normal form model for irregular KZ
equations. 

After choosing a representation in which the leading operator $\hbar E_{[p+1]}$ admits a spectral
decomposition, and specializing $\hbar$ to a nonzero complex value, one may apply the Hukuhara-Levelt-Turrittin reduction to equation 
\eqref{hqeqKZ} to get formal irregular type and then study the Stokes matrices. We leave the study of the Stokes phenomenon in that setting and its relation with the regular semisimple normal form used here to future work.
\vspace{3mm}

In the second order pole case, the Stokes phenomenon and quantum Stokes matrices of the irregular KZ equation were studied in \cite{TL, TLX}. As far as we know, the Stokes phenomenon of higher order irregular KZ
equations has not yet been systematically studied. The present work is an attempt in this direction. Several natural questions remain. One is to construct a quantum isomonodromy system associated with the differential equation \eqref{hqeq}. In the quantum case, the deformation (multitime) parameters are the diagonal entries of $D_r$, $r=p,\cdots,2$ in the irregular type, which are central in $\Uph$ and therefore may be specialized to complex parameters. It is then interesting to ask for the equation governing the variation of the quantum Stokes matrices with respect to these parameters. Another question is to study the WKB approximation of the quantum Stokes matrices and its relation with canonical/crystal basis, generalizing the second order pole results in \cite{Xu1}. 

\section{Formal solutions, canonical holomorphic solutions and quantum Stokes matrices}\label{beginsec1}
Section \ref{beginsec} introduces the coefficient algebra of the differential equation \eqref{hqeq}. Section \ref{formalsolution} constructs the formal solutions. In the end, Section \ref{q-Stokes} defines canonical holomorphic solutions and quantum Stokes matrices of \eqref{hqeq}.

\subsection{Coefficient algebras}\label{beginsec}
In this subsection we fix the coefficient algebra and its completion. 

The coefficient algebra $\Uph$ is the associative
$\mathbb C[\hbar]$-algebra generated by $\hbar e_{ij},\hbar e_{ij}^{(a)}$, $1\le i,j\le n, 1\le a\le p-1,$
subject to the relations \eqref{algebra1} and \eqref{algebra2}. The elements $e_{ij}$ and $e_{ij}^{(a)}$ may be interpreted in the localization $\Uph[\hbar^{-1}]$. However, every coefficient of the differential equations considered in this paper belongs to $\Uph$, and all algebraic arguments in the following sections use the $\hbar$ scaled generators.

The assignments
\[
 \epsilon(\hbar)=\epsilon(\hbar e_{ij})
 =\epsilon(\hbar e_{ij}^{(a)})=0
\]
define an augmentation $\epsilon:\Uph\longrightarrow\mathbb C$. Let
\[
 \mathfrak m:=\ker\epsilon
 =
 (\hbar,\hbar e_{ij},\hbar e_{ij}^{(a)})\subset\Uph
\]
be its augmentation ideal.  Since \(\Uph\) is generated as a
\(\mathbb C\)-algebra by the finite set
$\{\hbar,\hbar e_{ij},\hbar e_{ij}^{(a)}\}\subset\mathfrak m,$ the quotient $\Uph/\frak m^N$ is spanned by the
images of the finitely many words of length less than $N$, and is
therefore finite dimensional over $\mathbb C$.

We define the $\frak m$-adic completion of $\Uph$
\[\wUph:=\varprojlim_N\Uph/\mathfrak m^N.
\]
Then the canonical map 
\[
 \Uph\longrightarrow\wUph
\]
is injective. It is because the defining relations \eqref{algebra1} and \eqref{algebra2} are homogeneous for the positive weights
\[
 {\rm wt}(\hbar)={\rm wt}(\hbar e_{ij})=p,
 \qquad
 {\rm wt}(\hbar e_{ij}^{(a)})=p-a.
\]
If $W_{r}$ denotes the space of the homogeneous elements of weight $r$, then
\[
\bigoplus_{r\geq pN}W_r\subseteq\mathfrak m^N\subseteq \bigoplus_{r\geq N}W_r.
\]
Therefore, we have
$\bigcap_{N\geq1}\frak m^N=0.$ The PBW property
and its semiclassical limit will be
given in Section \ref{PoissonStokes}. 

Let $\mathbb{D}\subset\mathbb C$ be a domain. A function
\[
f:\mathbb{D}\longrightarrow {\wUph}
\]
is called holomorphic if  for every $N\ge 1$, the induced map to the finite dimensional quotient 
\[
f_N:\mathbb{D}\longrightarrow \Uph/\frak m^N
\]
is holomorphic in the usual sense. 
The same convention is used in the paper for functions with values in
\[
\wUph\otimes\End(\C^n),
\qquad
\wUph\otimes
\End(\C^n)\otimes\End(\C^n).
\]

Throughout the paper, we denote by $X_{ij}$ the matrix coefficient of an element $X\in \wUph\otimes\End(\C^n)$ in the matrix basis $E_{ij}$, i.e., we write $X=\sum_{i,j}X_{ij}\otimes E_{i,j}$. We call $X_{ij}$ the off-diagonal entry if $i\ne j$. We denote by $Y_{ijkl}$ the matrix coefficient of an element $Y\in \wUph\otimes\End(\C^n)\otimes\End(\C^n)$ in the matrix basis $E_{ij}\otimes E_{kl}$, i.e., we write $Y=\sum_{i,j,k,l}Y_{ijkl}\otimes E_{i,j}\otimes E_{k,l}$. 

\subsection{The formal solution}\label{formalsolution}
The existence of the following formal normal form of \eqref{hqeq} is a quantum analog of the Hukuhara-Levelt-Turrittin
formal classification theorem for the unramified
regular semisimple case.

\begin{thm}\label{oneformal}
For any fixed $u\in\h_{\rm reg}$, there exists a unique formal power series \[\widehat{H}(z)=1+H_1z+H_2z^{2}+\cdot\cdot\cdot \in \Uph\otimes{\rm End}(\mathbb{C}^n)\llbracket z\rrbracket\] 
and $p$ diagonal matrices  $D_k=\sum_{i=1}^n(D_{k})_{ii}\otimes E_{ii}\in \Uph\otimes{\rm End}(\mathbb{C}^n)$, $k=1,..,p$ with pairwise
commuting diagonal entries
\begin{equation}
    [(D_l)_{ii}, (D_s)_{jj}]=0 \in \Uph, \qquad l,s=1,...,p, \quad i,j=1,...,n,
\end{equation} 
such that
\begin{eqnarray}\label{formalsum}
\widehat{F}(z)=\widehat{H}(z)\cdot  e^{\int D(z)dz}z^{D_1}
\end{eqnarray}
is a formal solution of \eqref{hqeq}. Here the element 
\begin{equation}\label{Dformal}
D(z):=\frac{u}{z^{p+1}}+\frac{D_p}{z^p}+\frac{D_{p-1}}{z^{p-1}}+\cdots +\frac{D_2}{z^2}\in \Uph\otimes{\rm End}(\mathbb{C}^n)[z^{-1}]
\end{equation}
is called the formal irregular type of the equation \eqref{hqeq}, and 
\[e^{\int D(z)}:={\rm exp}{\left(\frac{-u}{pz^p}+\sum_{m=1}^{p-1}\frac{-D_{m+1}}{mz^m}\right)}=e^{\frac{-u}{pz^p}}\prod_{m=1}^{p-1}e^{\frac{-D_{m+1}}{mz^m}} .\]
\end{thm}
We give two proofs. The first one shows
that every entry of the $H_m$ and $D_k$ is a finite algebraic expression,
hence lies in the algebra $\Uph$. The second proof solves the recursion relation modulo $\frak m^N$ and
then passes to the inverse limit, producing the (same) formal solution in the
completed algebra $\wUph$. Although the result of the second proof looks weaker, it is enough for our purpose. The finite quotient
construction will also be used for the formal solutions in Section \ref{SPDE} and the Borel–Laplace summation.

\begin{proof}[First proof of the existence]
First, let us prove that there exists an off-diagonal matrix formal power series 
\[\widehat{H}^{o}(z)=1+\sum_{i\ge 1}H^{o}_i z^i\in \Uph\otimes{\rm End}(\mathbb{C}^n)\llbracket z\rrbracket,\] 
(here off-diagonal means that if we write $H_m^{o}=\sum_{i,j=1}^nH_{m,ij}\otimes E_{ij}$, then $H^{o}_{m,ii}=0$ for $i=1,\ldots, n$) that diagonalizes \begin{equation}\label{Tz}
    T(z)=\frac{u}{z^{p+1}}+\sum_{j=2}^p\frac{\hbar T_{[j]}}{z^j}-\frac{\hbar T}{z}
\end{equation}
under the gauge transformation
\[\widehat{H}^{o}(z)[T(z)]:=\widehat{H}^{o}(z)^{-1}\cdot T(z) \cdot \widehat{H}^{o}(z)-\widehat{H}^{o}(z)^{-1}\cdot\frac{d\widehat{H}^{o}}{dz}.\] That is there exist diagonal matrices $D_k\in \Uph\otimes{\rm End}(\mathbb{C}^n)$ for $k\le p$ such that
\begin{equation}
\widehat{H}^{o}(z)[T(z)]=\frac{u}{z^{p+1}}+\frac{D_p}{z^p}+\cdots+\frac{D_1}{z}+D_{0}+D_{-1}z^1+\cdots.
\end{equation}

Plugging $\widehat{H}^{o}(z)=1+\sum_{k\ge 1}H^{o}_k z^k$ into the above equation gives (here $H^{o}_0=1$ and by convention $H^{o}_m=0=T_{[m+1]}$ if $m<0$ and $T_{[1]}:=-T$)
\begin{align}
[u,H^{o}_{m+1}]=(m-p+1)H^{o}_{m-p+1}+\sum_{l=p-m}^{p} H^{o}_{m-p+l}D_l-\sum_{l=p-m}^p \hbar T_{[l]}H^{o}_{m-p+l}.
\end{align}
On the one hand, 
since $H^{o}_{m+1}$ is off-diagonal and $u_1,\ldots,u_n$ are distinct, it is uniquely determined by the recursion relations. On the other hand, the left hand side is off-diagonal, which forces the  diagonal matrices $D_l$ to satisfy 
\begin{align}
D_{p-m}&={\rm diag}\left(\sum_{l=2}^p \hbar T_{[l]} H^{o}_{m+l-p}-\hbar TH^{o}_{m+1-p}-\sum_{l=p-m+1}^p H^{o}_{m+l-p}D_l\right).
\end{align}
In this way, the above recurrence relations uniquely determine $D_m$ for all $m\le p$ and the off-diagonal matrices $H^{o}_m$ for $m\ge 1$.

Second, by using Lemma \ref{Dcommute} and by applying the above recurrence relations repeatedly, we obtain for $i,j,k=1,...,n$,
\begin{equation}
[(D_1)_{kk}, (H^{o}_m)_{ij}]=-\hbar(\delta_{ki} (H^{o}_m)_{ij}-\delta_{jk}(H^{o}_m)_{ij}),
\end{equation}
and then 
\begin{equation}\label{1commutem}
[(D_1)_{kk}, (D_m)_{ii}]=0 \text{ for  } m=p,...,1,0,-1,....
\end{equation}

To complete the proof, let $W(z)=1+\sum_{k=1} w_k z^k\in \Uph\otimes{\rm End}(\mathbb{C}^n)\llbracket z\rrbracket$ be the diagonal matrix formal power series with the coefficients $w_k$ determined by ($w_0:=1)$ 
\begin{equation}
(k+1)w_{k+1}=\sum_{i=0}^{k}D_{i-k}\cdot w_i.
\end{equation}
That is the unique normalized formal solution of $W'(z)=N(z)W(z)$ with $N(z)=\sum_{k\ge 0}D_{-k}z^k$. Set $\widehat{H}(z):=\widehat{H}^{o}(z)\cdot W(z)$. Then we have
\[\widehat{H}(z)[T(z)]=W^{-1}\left(D(z)+\frac{D_1}{z}+N(z)\right)W-W^{-1}W'=W^{-1}\left(D(z)+\frac{D_1}{z}\right)W=D(z)+\frac{D_1}{z}.\]
Here the last identity holds because the entries $(w_k)_{ii}$ of the diagonal matrix $w_k$ are polynomials in the entries $(D_{-r})_{ii}$, which commute with $(D_1)_{ii}$ by \eqref{1commutem} and with $(D_l)_{ii}$ for $l=2,\ldots,p$ by Lemma \ref{Dcommute} (hence $[w_k, D(z)+\frac{D_1}{z}]=0$ for each $k$ and thus $[W, D(z)+\frac{D_1}{z}]=0$). As a consequence, 
$\widehat{F}(z):=\widehat{H}(z)\cdot  e^{\int D(z)}z^{D_1}$
is the formal solution as required in Theorem \ref{oneformal}.
\end{proof}

\begin{proof}[Second proof of the existence and the uniqueness] Assume that $\widehat{H}(z)=1+\sum_{k>0}H_kz^k$ is the required formal power series.
Set $T_{[1]}=-T$, $H_0=1$ and $H_k=0$ for $k<0$. Comparing the coefficient of
$z^{m-p}$ in of \eqref{hqeq} gives
\begin{equation}\label{Hrecursion}
  [u,H_{m+1}]
  =
  (m-p+1)H_{m-p+1}
  +
  \sum_{k=1}^{p}
  \left(
    H_{m-p+k}D_k
    -
    \hbar T_{[k]}H_{m-p+k}
  \right).
\end{equation}
We solve the equation \eqref{Hrecursion} by induction on the $\frak m$-adic order. 
First, modulo $\frak m$, we see that
$H^{[1]}_0=1$ and $H^{[1]}_m=0$ for $m\ge 1$ are the unique solution of \eqref{Hrecursion}.
Suppose that the solution is known modulo $\frak m^{N-1}$ at every recursion step, and has already been determined modulo $\frak m^{N}$ through step $b\le m$. Let us then solve $H_{m+1}^{[N]}$.

We write $H_m=\sum_{i,j}(H_m)_{ij}\otimes E_{ij}$. Since $[u,H_m]_{ij}=(u_i-u_j)(H_m)_{ij}$, the operator $\ad_u$ is invertible on off-diagonal entries.
Modulo $\frak m^N$, taking the off-diagonal part of \eqref{Hrecursion} determines
the off-diagonal part of $H_{m+1}^{[N]}$
uniquely.  For
$0 \le m\le p-1$, the diagonal part of the same equation determines $D_{p-m}^{[N]}$, successively from $D_p$ down
to $D_1$.

It remains to determine the diagonal part of $H_{m+1}^{[N]}$. Taking the diagonal part of the recursion relation $p$ step later (replacing $m$ by $m+p$ in \eqref{Hrecursion}) gives
\begin{equation}\label{HrecursionN}
  0={\rm diag} ((m+1)H_{m+1}^{[N]})+
  {\rm diag} 
  \sum_{k=1}^{p}
  \left(
    (H_{m+k}D_k)^{[N]}
    -
    (\hbar T_{[k]}H_{m+k})^{[N]}
  \right).
\end{equation}
Note that 
since entries of $\hbar T_{[k]}$ lie in $\frak m$, and the $D_k$
are obtained from these entries, one
gets $D_k\in\frak m $ for all $k$. Thus in \eqref{HrecursionN}, every term $(H_{m+k}D_k)^{[N]}$, $(\hbar T_{[k]}H_{m+k})^{[N]}$ in the sum carries a factor in $\frak m$, and depends only on the already known
$H_{m+k}^{[N-1]}$. Therefore, ${\rm diag}(H_{m+1}^{[N]})$ is uniquely determined by \eqref{HrecursionN}.

In the end, the projection
\[\pi_N^{N+1}:(\Uph/\mathfrak m^{N+1})
\otimes\End(\mathbb C^n)\longrightarrow (\Uph/\mathfrak m^N)
\otimes\End(\mathbb C^n)
\]
maps $H_m^{[N+1]}$ to a solution of the recursion with the
same initial condition. Uniqueness above therefore ensures $\pi_N^{N+1}(H_{m}^{[N+1]})=H_{m}^{[N]}.$ Passing to the inverse limit gives the unique solution $H_m$ in
$H_m\in \wUph
\otimes\End(\mathbb C^n)$. 
\end{proof}
\begin{rmk}
In the construction of solutions in Lemma \ref{YWFfactorization} and Proposition \ref{t1fact}, we shall repeatedly use the above double induction on the recursion step $m$ and the $\frak m$-adic order $N$. The key point is that every undetermined coefficient occurring in a $p$ step delayed equation is multiplied by an element of $\frak m$.
Consequently, modulo $\frak m^N$, it depends only on its already
known class modulo $\mathfrak m^{N-1}$. Uniqueness at every finite
$\frak m$-adic order then implies compatibility under the projections $\Uph/\frak m^{N+1}\to \Uph/\frak m^{N}$ such that
passing to the inverse limit gives a unique solution in $\wUph$.
\end{rmk}
\begin{rmk}
From the first proof we see that the coefficients $H_k$ of the power series part of the formal solution are finite algebraic
expressions for each fixed $k$, i.e., 
\[
  H_k\in
  \Uph\otimes\operatorname{End}(\mathbb C^n).
\]
Thus
\[
  \widehat H(z)=1+\sum_{k\ge1}H_kz^k
  \in
  \left(
    \Uph\otimes\operatorname{End}(\mathbb C^n)
  \right)\llbracket z\rrbracket.
\]
The completion of the coefficient algebra is needed only when one passes from
the formal series $\widehat{H}$ to its order-$p$ Borel-Laplace resummation. 
The full formal solution is written in the form
\[
  \widehat F(z)
  =
  \widehat H(z)\cdot
  e^{\int D(z)dz}z^{D_1},
\]
where the second factor is regarded as the diagonal singular factor. It is understood in the differential extension obtained by adjoining the corresponding exponential and logarithmic factors. After resummation it becomes a multivalued holomorphic function with values in $\wUph \otimes
\operatorname{End}(\mathbb C^n)$. 
\end{rmk}

\begin{lem}\label{Dcommute}
The diagonal elements of the matrices $D_p,...,D_2$ lie in the center of the associative algebra $\Uph$. That is
\begin{align}\label{Dcenter2}
[(D_m)_{kk}, (\hbar T)_{ij}]=0, \quad
[(D_m)_{kk}, (\hbar T_{[l]})_{ij}]&=0  \ \text{ for all } l,m=2,...,p, \text{ and } \ i, j, k=1,...,n.
\end{align}
The diagonal elements of the matrix $D_1$ satisfy
\begin{align}\label{D1T}
[(D_1)_{kk}, (\hbar T)_{ij}]&=-\hbar\left(\delta_{ki} (\hbar T)_{kj}-\delta_{jk} (\hbar T)_{ik}\right),\\ \label{D1T1}
[(D_1)_{kk}, (\hbar T_{[l]})_{ij}]&=-\hbar(\delta_{ki} (\hbar T_{[l]})_{kj}-\delta_{jk} (\hbar T_{[l]})_{ik}) \ \text{ for } l=2,...,p, \text{ and } \ i, j, k=1,...,n,
\end{align}

\end{lem}

\begin{proof}
Recall that the off-diagonal formal gauge transformation $\widehat{H}^0$, constructed in the proof of
Theorem \ref{oneformal} (its construction does
not use the present lemma), diagonalizes the coefficient matrix $T(z)$, i.e.,  
\begin{equation}\label{HdiagT}(\widehat{H}^o)^{-1}\cdot T(z)\cdot \widehat{H}^o-(\widehat{H}^o)^{-1}\cdot \frac{d\widehat{H}^o}{dz}=
\frac{u}{z^{p+1}}+\sum_{m\leq p}D_m z^{-m}.
\end{equation}
We compute the entrywise action of the generators $\hbar e^{(a)}_{ij}$ on the matrix identity \eqref{HdiagT}. The entrywise action can be expressed via the commutator, for example, $[\hbar e^{(a)}_{ij}\otimes 1, T(z)]=\sum_{k,l}[\hbar e^{(a)}_{ij},T(z)_{kl}]\otimes E_{kl}$.
The defining relations of the coefficient algebra give
\begin{equation}\label{xT}
[\hbar e^{(a)}_{ij}\otimes 1, T(z)]
=
[\hbar E_{ji}z^a,T(z)]+R_{\hbar e^{(a)}_{ij}}(z),
\end{equation}
where $R_{\hbar e^{(a)}_{ij}}(z)$ has at most a simple pole at $z=0$, and
\begin{equation}\label{resT}
    {\rm Res}_{z=0}R_{\hbar e^{(a)}_{ij}}(z)
=
-[\hbar E_{ji},\hbar T_{[a+1]}].
\end{equation}
Taking the commutator of ${\hbar e^{(a)}_{ij}}\otimes 1$ with both sides of the identity \eqref{HdiagT} and using the Leibniz rule, formula \eqref{xT}, we get
\begin{equation}\label{XD=XH}
\left[{\hbar e^{(a)}_{ij}}\otimes 1, \frac{u}{z^{p+1}}+\sum_{m\leq p}D_m z^{-m}\right]=\left[L,\frac{u}{z^{p+1}}+\sum_{m\leq p}D_m z^{-m}\right]+\frac{dL}{dz}+M,
\end{equation}
where
\[
L:=(\widehat{H}^o)^{-1}(\hbar E_{ji}z^a)\widehat{H}^o-(\widehat{H}^o)^{-1}\cdot [{\hbar e^{(a)}_{ij}}\otimes 1,\widehat{H}^o],
\quad M:=(\widehat{H}^o)^{-1}\cdot \left(R_{\hbar e^{(a)}_{ij}}(z)-\frac{d}{dz}(\hbar E_{ji}z^a)\right)\cdot \widehat{H}^o.
\]
Since $d(E_{ji}z^a)/dz=O(1)$ and $\widehat{H}^o=1+O(z)$, we have $L=\sum_{r\ge 1}L_rz^r=O(z)$ and $M$ has at most a simple
pole with
\[
{\rm Res}_{z=0}M={\rm Res}_{z=0}R_{\hbar e^{(a)}_{ij}}.
\]

Taking the $(k,k)$-entry of the identity \eqref{XD=XH} and comparing the coefficient of $z^{-m}$,
$m\geq 1$, gives
\begin{equation}\label{commidentity}
[{\hbar e^{(a)}_{ij}},(D_{m})_{kk}]
=
\sum_{\substack{1\leq r\leq p-m}}
\bigl[(L_r)_{kk},(D_{m+r})_{kk}\bigr]
+
\delta_{m,1}({\rm Res}_{z=0}R_{\hbar e^{(a)}_{ij}})_{kk}.
\end{equation}

For $m=p$, the sum on the right hand side is empty, and thus 
$[\hbar e^{(a)}_{ij},(D_p)_{kk}]=0$ for $k=1,\dots,n$. We can then apply descending induction
on $m$ to show that $(D_m)_{kk}$ commutes with
every generator $\hbar e^{(a)}_{ij}$
for $2\le m\le p$. Note that the construction of $D_p,\ldots, D_2$ involves only the
generators $\hbar e^{(a)}_{ij}$, $a=1,\ldots,p-1$, which commute with the residue generators $\hbar e_{ij}$ by \eqref{algebra2}. Thus the entries of $D_m$, $m=p,\ldots,2$, are central. 

Letting $m=1$ in \eqref{commidentity} and using \eqref{resT} give the identity \eqref{D1T1}. In the end, notice that by the construction of $D_1$
\[(D_1)_{kk}=-\hbar T_{kk}+\text{terms in the generators ${\hbar e^{(a)}_{ij}}$ for $a=1,\ldots,p-1$.}\] 
Then the defining relation \eqref{algebra2} leads to \eqref{D1T}.
\end{proof}
\begin{ex}
For $p=3$, the formal type is $
D(z)=\frac{u}{z^{4}}+\frac{D_3}{z^3}+\frac{D_{2}}{z^{2}}$, and the residue part is $\frac{D_1}{z}$ with
\begin{align*}
(D_3)_{ii}&=\hbar ( T_{[3]})_{ii}, \ \ \ (D_2)_{ii}=\hbar (T_{[2]})_{ii}+\hbar^2\sum_{j=1,j\ne i}^n\frac{(T_{[3]})_{ij}(T_{[3]})_{ji}}{u_i-u_j}\\ 
(D_1)_{kk}=&-\hbar T_{kk}+\hbar^2\sum_{i=1,i\ne k}^n\frac{(T_{[2]})_{ki}(T_{[3]})_{ik}}{u_k-u_i}-\hbar^2\sum_{i=1,i\ne k}^n\frac{(T_{[3]})_{ki}(T_{[2]})_{ik}}{u_i-u_k}\\
&+\hbar^3\sum_{i=1,i\ne k}^n\frac{(T_{[3]})_{ki}(T_{[3]})_{ik}(T_{[3]})_{kk}}{(u_i-u_k)(u_k-u_i)}-\hbar^3\sum_{i=1,i\ne k}^n\sum_{j=1,j\ne k}^n\frac{(T_{[3]})_{ki}(T_{[3]})_{ij}(T_{[3]})_{jk}}{(u_i-u_k)(u_k-u_j)}.
\end{align*}
\end{ex}

\subsection{Canonical holomorphic solutions and quantum Stokes matrices}\label{q-Stokes}

Note that the
leading term in the exponential factor $e^{\int D(z)dz}$ of the equation \eqref{hqeq} is $e^{\frac{-u}{pz^p}}$.

\begin{defi}\label{Stokesrays}
The \textbf{anti-Stokes rays} of the equation \eqref{hqeq} are the directions along which ${\rm exp}(\frac{u_i-u_j}{pz^p})$ decays most rapidly as $z\rightarrow 0$ for some $i\ne j$, i.e., the directions along which $\frac{u_i-u_j}{pz^p}$ is real and negative. Denote by
\begin{align*}
        {\rm aSR}(u) &:=\left\{\frac{1}{p}\mathrm{Arg}(u_j-u_i)+\frac{2k\pi}{p}~|~ \text{ for all }  k\in \mathbb{Z}, \text{ and } i\neq j\right\}
    \end{align*}
the set of anti-Stokes rays/directions.
\end{defi}

In this paper we denote a direction/ray by its argument, understood as the lifted argument on the universal cover $\widetilde{\mathbb{C}^\times}$ of $\mathbb{C}^\times=\mathbb{C}\setminus\{0\}$. Let us choose an initial anti-Stokes direction $\tau_0\in\mathbb{R}$ and then
arrange the anti-Stokes directions into a strictly monotonically increasing sequence 
\begin{align*}
        \cdots <\tau_{-1}<\tau_0<\tau_1<\cdots.
    \end{align*}

\begin{defi}
For any direction $d\in (\tau_j,\tau_{j+1})$, the \textbf{Stokes sector} $\Sect_d$ is defined as 
    \begin{align*}
        \Sect_{d}:=\left\{z\in\widetilde{\mathbb{C}^\times} : \mathrm{arg}(z)\in\bigg(\tau_j-\frac{\pi}{2p},\tau_{j+1}+\frac{\pi}{2p}\bigg)\right\}.
    \end{align*}
    \end{defi}

\begin{pro}
\label{CanSol}
Given any fixed $u\in\h_{\rm reg}$ and a direction $d\notin{\rm aSR}(u)$, there exists a unique $\wUph\otimes{\rm End}(\mathbb{C}^n)$ valued holomorphic solution ${F}_d(z)$
on ${\Sect}_d$ of the equation \eqref{hqeq} with the asymptotics 
\begin{equation}\label{preasy}
    {F}_d(z) \cdot z^{-D_1}e^{-\int D(z)}\sim 1, \ \text{ as } z\rightarrow 0 \text{ in } {\Sect}_d.
\end{equation}
It is the canonical sectorial solution associated with $d$, and admits analytic continuation to
$\widetilde{\mathbb{C}\setminus\{0\}}$. 
\end{pro}
\begin{proof}

It follows from a standard construction via the order-$p$ Borel-Laplace transform of the formal power series $\widehat{H}(z)$ along $d$. The order-$p$ Borel-Laplace transform/resummation for power series with coefficients in the infinite dimensional vector space $  \Uph  \otimes \End(\C^n)$ is understood as follows. 

For every $N$, the reduction of the equation modulo $\frak m^N$ gives a
meromorphic linear system with coefficients in the finite dimensional algebra
$ (\Uph/\frak m^N)  \otimes \End(\C^n)$.
Let
\[
  \widehat{F}^{[N]}(z)
  =
  \widehat{H}^{[N]}(z)
  \cdot e^{\int D^{[N]}(z)dz}z^{D_{1}^{[N]}}
\]
be the formal solution of the finite dimensional system over $ \Uph/\frak m^N$. 
Note that the coefficient matrices of these reduced systems are compatible with the
natural projections
\[
  \pi_N^{N+1}: \Uph/\frak m^{N+1} \longrightarrow  \Uph/\frak m^N.
\] Thus, applying
$\pi_N^{N+1}$ termwise to the formal solution over
$\Uph/\frak m^{N+1}$ gives a normalized formal solution over $\Uph/\frak m^N$, that is
\[
  \pi_N^{N+1}(\widehat{H}^{[N+1]}(z))
  =
  \widehat{H}^{[N]}(z),
  \qquad
  \pi_N^{N+1}(D_{r}^{[N+1]})
  =
  D_{r}^{[N]}.
\]

The standard order-$p$ Borel-Laplace resummation theorem then applies in each finite
dimensional quotient. To be more precise, after the
reduction modulo $\frak m^{N}$ the system becomes an ordinary finite rank $\C$-linear system of finite rank, whose leading coefficient has each eigenvalue $u_i$ with multiplicity ${\rm dim}(\Uph/\frak m^{N})$.  Thus the leading coefficient is not
regular semisimple in the usual sense. 
However, all diagonal coefficients
$D_{r}^{[N]}$, $r=1,\ldots,p$, lie in the nilpotent ideal
$\frak m/\frak m^N$.
As a consequence, after separating the scalar exponential $e^{-\frac{u}{p z^p}}$ from $e^{\int D^{[N]}(z)dz}$, the remaining exponential factor
\[
{\rm exp}\left(-\sum_{r=2}^{p}\frac{D_{r}^{[N]}}{(r-1)z^{r-1}}\right)
\]
is a Laurent polynomial in $1/z$. If we absorb the factor into the series part, then the formal solution takes the (unramified) form 
\[ \widehat{F}^{[N]}(z)=\widehat{G}^{[N]}(z) \cdot e^{\frac{-u}{pz^p}}z^{D_{1}^{[N]}}\]
where $\widehat G^{[N]}(z)$ is a formal Laurent series. 
Therefore the only exponential parts are the scalar exponentials
$e^{-\frac{u_i}{p z^p}}$ for $i=1,\ldots,n.$
All differences of exponential parts have the same order \(z^{-p}\).  Hence
the system has a single level $p$.  By the $p$-summability theorem for
single-level meromorphic systems (see e.g., the books
\cite{Balser, LR}, and see \cite{BBRS,MR} for the multisummability theorem), the formal series $\widehat{G}^{[N]}(z)$, and
therefore equivalently the original formal series $\widehat {H}^{[N]}(z)$, is
Borel-Laplace $p$-summable along every direction that is not an anti-Stokes
direction, i.e. along every $d\notin{\rm aSR}(u)$.

Thus, for every $d\notin{\rm aSR}(u)$,  the Borel-Laplace $p$-sum of $\widehat{H}^{[N]}(z)$ gives a holomorphic function
\[
H_{d}^{[N]}(z)\in
(\Uph/\frak m^N)
\otimes\operatorname{End}(\mathbb C^n).
\]
Since the Borel-Laplace $p$-sum is unique in each finite dimensional quotient, and the order-$p$ Borel transform (also known as $p$-Borel transform)
and the order-$p$ Laplace transform (also known as $p$-Laplace transform) commute with the finite dimensional projection
$\pi_N^{N+1}$, we have
\[
  \pi_N^{N+1}(H_{d}^{[N+1]})
  =
  H_{d}^{[N]}.
\]
Therefore the inverse limit defines a holomorphic function
\[
H_d(z)\in
\wUph
\otimes\operatorname{End}(\mathbb C^n).
\]
Note that the
diagonal singular factor $e^{\int D(z)dz}z^{D_1}$
is also a holomorphic
$\wUph
\otimes\End(\C^n)$-valued function (on the covering of $\mathbb{C}\setminus \{0\}$). Hence we can define the holomorphic solution
\[F_d(z):=H_d(z)\cdot e^{\int D(z)}z^{D_1}\in
\wUph
\otimes\End(\C^n).\]
\end{proof}
Holomorphic functions below, as well as the quantum Stokes matrices, are understood in this completed
$\frak m$-adic sense.

\begin{defi}
$(1)$ 
For any fixed $u\in\h_{\rm reg}$, the quantum Stokes factor $\mathcal{K}_{j}(u)\in \wUph\otimes {\rm End}(\mathbb{C}^n)$ of the equation \eqref{hqeq}, associated to the anti-Stokes ray $\tau_j\in {\rm aSR}(u)$, is determined by the identity
   \begin{align}
    F_{d}(z)=F_{d'} (z)\cdot \mathcal{K}_{j}(u).
    \end{align}
Here $d,d'\notin {\rm aSR}(u)$ are two arbitrary directions such that $\tau_{j-1}<d<\tau_j<d'<\tau_{j+1}$ (thus no other anti-Stokes ray lies between $d$ and $d'$).

$(2).$ The quantum Stokes matrix $S_i(u)\in \wUph\otimes{\rm End}(\mathbb{C}^n)$ of the equation \eqref{hqeq} is
\[S_i(u):=\mathcal{K}_{(i+1)l}(u)\cdots \mathcal{K}_{il+1}(u).\]
Here $l=\frac{\#({\rm aSR}(u)\cap [0,2\pi))}{2p}$ (note that ${\rm aSR}(u)$ are invariant under the rotation by $\pi/p$, thus the integer $l$ is well defined, and $l=\frac{n(n-1)}{2}$ for generic $u$). Equivalently,
\begin{align}
    F_{d}(z)=F_{d'} (z)\cdot S_{i}(u),
    \end{align}
if $d,d'\notin {\rm aSR}(u)$ are two arbitrary directions such that $\tau_{il}<d<\tau_{il+1}<\cdots<\tau_{(i+1)l}<d'<\tau_{(i+1)l+1}$.

$(3).$ 
We call a direction $d$ an \textbf{admissible direction} if $d\notin {\rm aSR}(u)$ or if $d=\tau_l\in {\rm aSR}(u)$ but the associated Stokes factor $\mathcal{K}_{l}=1$.
\end{defi}
 
We remark that the Stokes matrices depend on the chosen initial anti-Stokes direction $\tau_0$. Up to a possible permutation of $u_1,\ldots,u_n$, the Stokes matrices $S_{2k+1}(u)$ and $S_{2k}(u)$ for any integer $k$ are lower and upper triangular respectively, and have ones along the diagonal (see e.g., \cite{Boalch2}). Furthermore, the Stokes matrices satisfy 
\begin{equation}S_{k+2p}(u)=e^{-2\pi \I D_1} S_{k}(u) e^{2\pi \I  D_1},
\end{equation}
when the branch of the multivalued function $z^{-D_1}$ in \eqref{preasy} is taken into account.

\section{Formal solutions for a two variable system and their factorization}\label{SPDE}
The proof of Theorem \ref{thmRLL} relies on the study of the Stokes phenomenon of a compatible system of partial differential equations \eqref{kKZ1}-\eqref{kKZ2} with two variables. In this section, we study the formal solutions of the system.
First, in Section \ref{sec33} we construct the formal solution
$\widehat Y$ and show that the coefficients
of its formal power series part have poles only at $t=1,\infty$. Second, in Sections \ref{factYinf}-\ref{factY1} we prove that 
near each of the points $t=0,1,\infty$, $\widehat Y$ has a formal fractorization as a product of two factors. The analytic factorization after the resummation is considered in the next section. 

\subsection{A compatible system of partial differential equations with two variables}\label{twoKZ}
Consider a system of equations, for a $\wUph\otimes {\rm End}(\IC^n)\otimes {\rm End}(\IC^n)$ valued function $Y(z_1,z_2)$ of two complex variables, 
\begin{align}\label{kKZ1}
\frac{{\partial Y}}{\partial z_1}&= \left(\frac{u^{(1)}}{z_1^{p+1}}+\frac{ \hbar T_{[p]}^{(1)}}{z_1^p}+\cdots+\frac{\hbar T_{[2]}^{(1)}}{z_1^2}-\frac{\hbar T^{(1)}}{z_1}+\frac{\hbar P}{z_1}+
\frac{\hbar P}{z_2-z_1}\right)\cdot Y, \\ \label{kKZ2}
\frac{{\partial Y}}{\partial z_2}&= \left(\frac{u^{(2)}}{z_2^{p+1}}+\frac{\hbar T_{[p]}^{(2)}}{z_2^p}+\cdots+\frac{\hbar T_{[2]}^{(2)}}{z_2^2}-\frac{\hbar T^{(2)}}{z_2}+\frac{\hbar P}{z_2}+
\frac{\hbar P}{z_1-z_2}\right)\cdot Y,
\end{align}
Here $u={\rm diag}(u_1,\ldots,u_n)\in\h_{\rm reg}$, and
\begin{align*}
&u^{(1)}= \sum_i 1\otimes u_i E_{ii}\otimes 1, \ u^{(2)}= \sum_i 1\otimes1 \otimes u_i E_{ii},& \\ 
&\hbar T^{(1)}=\sum_{k,l} \hbar e_{kl}\otimes E_{kl}\otimes 1, \ \hbar T^{(2)}=\sum_{k,l} \hbar e_{kl}\otimes 1\otimes E_{kl},\\ 
&\hbar T_{[j]}^{(1)}=\sum_{k,l} \hbar e_{kl}^{(j-1)}\otimes E_{kl}\otimes 1, \ \hbar T_{[j]}^{(2)}=\sum_{k,l} \hbar e_{kl}^{(j-1)}\otimes 1\otimes E_{kl}, \ \hbar P=\hbar \sum_{k,l} 1\otimes E_{kl}\otimes E_{lk},&
\end{align*}
are elements in $\Uph\otimes {\rm End}(\IC^n)\otimes {\rm End}(\IC^n)$. 

To show the system is compatible, we need the following lemma, whose proof follows by direct computation.
\begin{lem}\label{12commute}
We have the identities 
\begin{align}
[u^{(1)},\hbar T^{(2)}]&=[u^{(2)},\hbar T^{(1)}]=0,\\
\label{TP0}
[\hbar T^{(1)},\hbar T_{[m]}^{(2)}]&=[\hbar T^{(2)},\hbar T_{[m]}^{(1)}]=0,  \hspace{3mm}  \text{ for all } m=2,...,p,\\  
 [\hbar T^{(1)}+\hbar T^{(2)}, \hbar P]&=0, \ \ [\hbar T^{(1)}, \hbar P-\hbar T^{(2)}]=0,
\end{align}
and
\begin{align}
[u^{(1)},\hbar T^{(2)}_{[j]}]&=[ u^{(2)},\hbar T^{(1)}_{[j]}]=0 \  \text{ for } j=2,...,p, \\
\hbar T^{(1)}_{[i]}\hbar T^{(2)}_{[j]}&=\hbar T^{(2)}_{[j]}\hbar T^{(1)}_{[i]} \  \text{ for all } i+j \ge p+3,\\
\label{TP1}
    [\hbar T_{[m]}^{(1)}+\hbar T_{[m]}^{(2)}, \hbar P]&=0,   \hspace{3mm}  \text{ for all } m=2,...,p,\\  \label{TP2}
    [\hbar T_{[i]}^{(1)}, \hbar T_{[j]}^{(2)}]&=[\hbar P,\hbar T^{(2)}_{[i+j-1]}], \hspace{3mm}  \text{ for all } i,j\ge 2, \ i+j\le p+1,\\ \label{TP3}
   \hbar (u^{(1)}P-u^{(2)}P)&=\hbar T^{(1)}_{[i]}\hbar T^{(2)}_{[p+2-i]}-\hbar T^{(2)}_{[p+2-i]}\hbar T^{(1)}_{[i]}, \text{ for all } i=2,...,p.
\end{align}
\end{lem}
\begin{proof}
In terms of the basis $E_{ab}\otimes E_{cd}$ we have
\begin{align*}
    [\hbar T_{[i]}^{(1)},\hbar T_{[j]}^{(2)}]&=\sum_{a,b,c,d}[\hbar e_{ab}^{(i-1)}, \hbar e_{cd}^{(j-1)}]\otimes E_{ab}\otimes E_{cd}, \\ [\hbar P, \hbar T_{[i]}^{(2)}]&=\hbar \sum_{a,b,c,d}(\delta_{bc}\hbar e_{ad}^{(i-1)}-\delta_{ad}\hbar e_{cb}^{(i-1)})\otimes E_{ab}\otimes E_{cd},\\
    [\hbar T_{[i]}^{(1)},\hbar T_{[p+2-i]}^{(2)}]&=\hbar \sum_{a,b,c,d}\delta_{bc}\delta_{ad}(u_a-u_b)\otimes E_{ab}\otimes E_{cd},
\end{align*} 
All identities can be checked coefficientwise directly using the defining relations of the algebra.
\end{proof}

\begin{pro}
The linear system of partial differential equations \eqref{kKZ1} and \eqref{kKZ2} is compatible.
\end{pro}
\begin{proof}
Denote the coefficients of the equations \eqref{kKZ1} and \eqref{kKZ2} by $A(z_1,z_2)$ and $B(z_1,z_2)$ respectively. Then the compatibility condition is 
\[A(z_1,z_2)B(z_1,z_2)+\frac{\partial A(z_1,z_2)}{\partial z_2}=B(z_1,z_2)A(z_1,z_2)+\frac{\partial B(z_1,z_2)}{\partial z_1}.\]
We group the terms according to the total pole order in $z_1$ and $z_2$. By a direct computation, all
terms except the two families below cancel by Lemma \ref{12commute}:

(1) for all $m=2,...,p$
\begin{align}\nonumber
    &\sum_{\substack{i+j=m+1 \\ i,j\ge 2}}\frac{\hbar T^{(1)}_{[i]}\hbar T^{(2)}_{[j]}}{z_1^iz_2^j}+\frac{\hbar P\hbar T^{(2)}_{[m]}z_2}{z_1(z_2-z_1)z_2^m}+\frac{\hbar T^{(1)}_{[m]}\hbar Pz_1}{z_2(z_1-z_2)z_1^m}\\ \label{ijm0}=&\sum_{\substack{i+j=m+1 \\ i,j\ge 2}}\frac{\hbar T^{(2)}_{[j]}\hbar T^{(1)}_{[i]}}{z_2^jz_1^i} +\frac{\hbar T^{(2)}_{[m]}\hbar Pz_2}{z_1(z_2-z_1)z_2^m}+\frac{\hbar P\hbar T^{(1)}_{[m]}z_1}{z_2(z_1-z_2)z_1^m};
\end{align}

(2) for $m=p+1$ 
\begin{align}\nonumber
&\sum_{i+j=p+2}\frac{\hbar T^{(1)}_{[i]}\hbar T^{(2)}_{[j]}}{z_1^iz_2^j}
+\frac{u^{(1)}\hbar Pz_1}{z_2z_1^{p+1}(z_1-z_2)}+\frac{\hbar Pu^{(2)}z_2}{z_1z_2^{p+1}(z_2-z_1)}\\ \label{ijmk}
=&\sum_{i+j=p+2}\frac{\hbar T^{(2)}_{[j]}\hbar T^{(1)}_{[i]}}{z_2^jz_1^i} +\frac{\hbar Pu^{(1)}z_1}{z_2z_1^{p+1}(z_1-z_2)}+\frac{u^{(2)}\hbar Pz_2}{z_1z_2^{p+1}(z_2-z_1)}.
\end{align}

These identities follow from the commutation relations in Lemma \ref{12commute}. For instance, identity \eqref{ijm0} simplifies using $[P,T^{(1)}_{[m]}]=-[P,T^{(2)}_{[m]}]$ and the relation \eqref{TP2}. Identity \eqref{ijmk} follows similarly using \eqref{TP3}. Hence, the two expressions are equal, establishing compatibility. 
\end{proof}

\subsection{Construction of the formal solution $\widehat{Y}$ as $z_1\rightarrow 0$ and $z_2/z_1$ fixed}\label{sec33}

In terms of the new coordinate $z=z_1$ and $t=z_2/z_1$, the system \eqref{kKZ1} and \eqref{kKZ2} becomes
\begin{align}\label{ztkKZ1}
  \frac{{\partial Y}}{\partial z}&= \left(  \frac{u^{(1)}+t^{-p} u^{(2)}}{z^{p+1}}+\sum_{r=2}^p\frac{\hbar T_{[r]}^{(1)}+t^{1-r}\hbar T_{[r]}^{(2)}}{z^r}-\frac{\hbar T^{(1)}+\hbar T^{(2)}}{z}+
\frac{\hbar P}{z}\right)\cdot Y, \\ 
\label{ztkKZ2}
  \frac{{\partial Y}}{\partial t}&=\left(  \frac{u^{(2)}}{t^{p+1}z^p}+\sum_{r=2}^p\frac{\hbar T_{[r]}^{(2)}}{t^rz^{r-1}}-\frac{\hbar T^{(2)}}{t}+\frac{\hbar P}{t}+
\frac{\hbar P}{1-t}\right)\cdot Y,
\end{align}

\begin{pro}\label{t1fact}
For any fixed $u\in\h_{\rm reg}$, the  system of equations \eqref{ztkKZ1} and \eqref{ztkKZ2} has a unique formal solution taking the form \begin{eqnarray}\label{Yformalsum}
\widehat{Y}(z;t)=\widehat{Q}(z,t)\cdot e^{\int \left(D^{(1)}(z)+tD^{(2)}(zt)\right)dz}z^{\left(D_1^{(1)}+D_1^{(2)}+\hbar\delta P\right)}t^{D_1^{(2)}+\hbar\delta P}(1-t)^{-\hbar\delta P}, \end{eqnarray}
where 
\[\widehat{Q}=1+Q_1(t)z+Q_2(t)z^{2}+\cdot\cdot\cdot,\] and each coefficient 
\begin{equation}\label{Qtrational}
  Q_m(t)\in \varprojlim_N \left[\left(
\Uph/\mathfrak m^N
\otimes\End(\mathbb C^n)^{\otimes 2}\right){\otimes}\C[t,(t-1)^{-1}]\right].
\end{equation}
That is, for every $m,N\ge 1$, the reduction $Q^{[N]}_m(t)$ modulo $\frak m^N$ is an $
(\Uph/\frak m^N)
\otimes\End(\mathbb C^n)^{\otimes 2}$ valued rational function of $t$ whose poles are contained in
$\{1,\infty\}$. Moreover, its pole orders at $t=\infty$ and $t=1$ satisfy
\begin{equation}\label{Qinfdegree}
{\rm ord}^{\rm pole}_
{t=\infty}(Q^{[N]}_m(t))\le m,
\end{equation}
and
\begin{equation}\label{Q1degree}
{\rm ord}^{\rm pole}_
{t=1}(Q^{[N]}_m(t))
 \leq\left\lfloor\frac mp\right\rfloor.
\end{equation}
\end{pro}

\begin{proof}
Substitution of \eqref{Yformalsum} into the $z$-equation and
comparison of the coefficient of $z^{m-p}$ gives
\begin{align}\nonumber
 &[u^{(1)}+t^{-p}u^{(2)},Q_{m+1}]\\ \nonumber
=&(m-p+1)Q_{m-p+1}
 +\hbar(T^{(1)}+T^{(2)}-P)Q_{m-p+1}+Q_{m-p+1}(D_1^{(1)}+D_1^{(2)}+\hbar\delta P)\\ \label{Qzrecur}
 &+\sum_{l=2}^pQ_{m-p+l}
       (D_l^{(1)}+t^{1-l}D_l^{(2)})-\hbar\sum_{l=2}^p
       (T_{[l]}^{(1)}+t^{1-l}T_{[l]}^{(2)})Q_{m-p+l},
\end{align}
where we take the convention $Q_0=1$ and $Q_a=0$ for $a<0$.  Similarly, comparison of the
coefficient of $z^{m+1-p}$ of the $t$-equation gives
\begin{align}\nonumber
 t^{-p-1}[u^{(2)},Q_{m+1}]
=&\frac{d}{dt}Q_{m-p+1}
 -\sum_{l=2}^p\frac{\hbar T_{[l]}^{(2)}}{t^l}Q_{m-p+l}
 +\frac{\hbar(T^{(2)}-P)}tQ_{m-p+1}-\frac{\hbar P}{1-t}Q_{m-p+1}\\ \label{Qtrecur}
 &
 +\sum_{l=2}^pQ_{m-p+l}\frac{D_l^{(2)}}{t^l}+Q_{m-p+1}\frac{D_1^{(2)}+\hbar\delta P}{t}
 +Q_{m-p+1}\frac{\hbar\delta P}{1-t}.
\end{align}
We solve these identities modulo $\mathfrak m^N$. We write $Q_{m}=
 \sum_{i,j,k,l}
 Q_{m,ijkl}\otimes E_{ij}\otimes E_{kl}$, and denote by $Q_{m}^{[N]}=
 \sum_{i,j,k,l}
(Q_{m,ijkl}^{[N]})\otimes E_{ij}\otimes E_{kl}$ its reduction modulo $\frak m^{N}$.

First, modulo $\frak m$, we have
$Q^{[1]}_0=1$ and $Q^{[1]}_m=0$ for $m\ge 1$. Suppose that the solution is known modulo $\frak m^{N-1}$ at every recursion step, and has been determined modulo $\frak m^N$ through recursion step $m$. Assume also the rationality
and the pole position for these known classes.
To determine the $Q_{m+1}^{[N]}$, we consider the following three cases, where all identities are understood in
$\left(
  (\Uph/\mathfrak m^N)
  \otimes\End(\mathbb C^n)\otimes\End(\C^n)
 \right)[t,(t-1)^{-1}].$

Case I: If $k\ne l$, taking the $(ijkl)$-entry of
\eqref{Qtrecur} gives an identity of the form
\begin{equation}\label{casekl}
 t^{-p-1}(u_k-u_l)Q_{m+1,ijkl}^{[N]}
 =\mathcal{F}_{\le m}^{[N]}(t),
\end{equation}
where the right hand side only involves the already known coefficients $Q_a^{[N]}$ with $a<m+1$, that are rational in $t$ with poles at $\{1,\infty\}$. Hence this $(ijkl)$-entry is uniquely
determined and it is rational in $t$ with poles contained in $\{1,\infty\}$. 

Case II: If $k=l$ and $i\ne j$, the corresponding entry of
\eqref{Qzrecur} gives
\begin{equation}\label{eq:solve-ij}
 (u_i-u_j)Q_{m+1,ijkk}^{[N]}
 =\mathcal{G}_{\le m}^{[N]}(t),
\end{equation}
where the right-hand side is again already known. The entry is therefore uniquely determined and is rational in $t$ with poles contained in $\{1,\infty\}$. 

Case III. If $i=j$ and $k=l$, taking the
$(iikk)$-entry of \eqref{Qzrecur} $p$ steps later,
(i.e., replace $m$ in \eqref{Qzrecur} by $m+p$) gives
\begin{equation}\label{eq:solve-diagonal-delayed}
0=(m+1)Q_{m+1,iikk}^{[N]}+\mathcal{H}_{\le m}^{[N]}(t)+\mathcal{H}_{\ge m+1}^{[N]}(t).
\end{equation}
Here $\mathcal{H}_{\le m}^{[N]}$ and $\mathcal{H}_{\ge m+1}^{[N]}$ are the terms involving $Q_a^{[N]}$ with $a\le m$ and $a\ge m+1$ respectively. 
In this $m+1+p$ step equation, $\mathcal{H}_{\ge m+1}^{[N]}(t)$ contains $Q_{m+1}$ and the future coefficients
$Q_{m+2},\ldots,Q_{m+p}$. However, each of them is multiplied by one of
elements $D_l$, $\hbar T_{[l]}$, $\hbar T$, $\hbar P$, or
$\hbar\delta P$ in $\frak m$, therefore their contribution modulo $\frak m^N$
depends only on the reduction $Q_{m+1}^{[N-1]},\ldots,Q_{m+p}^{[N-1]}$ modulo $\frak m^{N-1}$. These
components are known from the induction. 
Thus the term $\mathcal{H}_{\le m}^{[N]}(t)+\mathcal{H}_{\ge m+1}^{[N]}(t)$ in
\eqref{eq:solve-diagonal-delayed} is already known and is rational with only possible poles at $t=1,\infty$. For $m+1> 0$,
the diagonal $(iikk)$-entry is uniquely determined and is rational in $t$ with poles contained in $\{1,\infty\}$. 

The three cases therefore uniquely determine every $Q_m^{[N]}$ recursively in the finite quotient modulo $\frak m^N$. One checks that the compatibility of the two variable system ensures that the $z$- and $t$-recursions \eqref{Qzrecur} and \eqref{Qtrecur} can be applied alternately.  The same double induction gives estimate \eqref{Qinfdegree} and \eqref{Q1degree}. It completes the
proof.
\end{proof}
\begin{rmk}
One can prove a stronger result, that for a fixed $m$,
$Q_m(t)$ is a rational function of $t$ and has a finite expression in the (uncompleted) coefficient
algebra $\Uph$. The finite quotient statement proved above is enough for the use below.
\end{rmk}

\subsection{Factorization of the formal solution $\widehat{Y}(z,t)$ as $t\rightarrow \infty$}\label{factYinf}
In the following proposition, 
the meaning of the formal series appearing and the identity \eqref{Y=W1} will be explained in the proof.
\begin{pro}\label{YWFfactorization}
After the substitution
$z_1=z$, $z_2=zt$, we have 
\begin{equation}\label{Y=W1}
    \widehat{Y}(z,t)= {_1\widehat{W}}(z_1;z_2)\widehat{F}^{(2)}(z_2)\cdot \left(\frac{t-1}{1-t}\right)^{\hbar\delta P} , \text{ with } \ |t|>1.
\end{equation}
Here 
\begin{equation}\label{formalKZ1}
_1\widehat{W}(z_1;z_2)=\left(1+\sum_{m\ge 1}K_m(z_2)z_1^{m}\right)\cdot e^{\int D^{(1)}(z_1)dz_1}z_1^{D^{(1)}_1+\hbar\delta P}(z_2-z_1)^{-\hbar\delta P}z_2^{\hbar\delta P}
\end{equation}
is the unique formal solution of 
the equation \eqref{kKZ1} of the prescribed form, such that each coefficient $K_m$ is a polynomial of degree at most $m-1$ in $z_2^{-1}$
\begin{equation}\label{Kmpoly}
    K_m(z_2)=\sum_{i=0}^{m-1} K_{m,i} z_2^{-i}\in \wUph\otimes {\rm End}(\IC^n)\otimes {\rm End}(\IC^n)[z_2^{-1}],
\end{equation}
and $\widehat{F}^{(2)}(z_2)$ is the unique formal solution of the equation 
\begin{equation}\label{KZin2}
\frac{{d F}}{dz_2}= \left(\frac{u^{(2)}}{z_2^{p+1}}+\frac{\hbar T_{[p]}^{(2)}}{z_2^p}+\cdots+\frac{\hbar T_{[2]}^{(2)}}{z_2^2}-\frac{\hbar T^{(2)}}{z_2}\right)\cdot F.
\end{equation}
Note that $\widehat{F}^{(2)}(z_2)$ is just the extension of the formal solution $\widehat{F}(z=z_2)$ from $\wUph\otimes {\rm End}(\IC^n)$ to $\wUph\otimes {\rm End}(\IC^n)\otimes {\rm End}(\IC^n)$.
\end{pro}
\begin{proof} 
We work first in the finite quotient $\Uph/\frak m^N$. All expansions and products below are
taken in the algebra of formal Laurent series in two variables $z_1/z_2$ and $z_2$
\[\left(\Uph/\mathfrak m^N
 \otimes\End(\mathbb C^n)^{\otimes2}\right)[[z_1/z_2,z_2]][(z_1/z_2)^{-1},z_2^{-1}].\]

Let us consider the ratio 
\[
_1\widehat{W}(z_1;z_2):=\widehat{Y}(z,t)\cdot \left(\frac{t-1}{1-t}\right)^{-\hbar\delta P}\cdot \widehat{F}^{(2)}(z_2)^{-1}=\widehat{Q}(z,t)\cdot z_2^{\hbar\delta P}(z_2-z_1)^{-\hbar\delta P}e^{\int D^{(1)}(z_1)}z_1^{D^{(1)}_1+\hbar\delta P}\cdot (\widehat{H}^{(2)})^{-1}.\]
It satisfies the equation \eqref{kKZ1}. Lemma \ref{Dcommute} and Lemma \ref{12commute} imply the commutator
\[\left[\frac{u^{(1)}}{z_1^{p+1}}+\frac{D^{(1)}_p}{z_1^p}+\cdots +\frac{D^{(1)}_2}{z_1^2}+\frac{D^{(1)}_1+\hbar\delta P}{z_1}, \ \ \frac{u^{(2)}}{z_2^{p+1}}+\frac{\hbar T_{[p]}^{(2)}}{z_2^p}+\cdots+\frac{\hbar T_{[2]}^{(2)}}{z_2^2}-\frac{\hbar T^{(2)}}{z_2}\right]=0. \]
Thus, $e^{\int D^{(1)}(z_1)}z_1^{D^{(1)}_1+\hbar\delta P}$ commutes with the coefficient matrix of the equation \eqref{KZin2}. Since it also commutes with the regularized initial condition, it commutes with the formal power series part $\widehat{H}^{(2)}$ of the solution $\widehat{F}^{(2)}(z_2)$. Using this, we write 
\begin{align}\nonumber
_1\widehat{W}&=\widehat{Q}(z,t)\cdot \left(1-\frac{z_1}{z_2}\right)^{-\hbar\delta P}\cdot (\widehat{H}^{(2)}(z_2))^{-1} \cdot e^{\int D^{(1)}(z_1)}z_1^{D^{(1)}_1+\hbar\delta P}\\ \label{Wztform}
&=\widehat{Q}(z,t) \left(1-\frac{z_1}{z_2}\right)^{-\hbar\delta P} (\widehat{H}^{(2)}(z_2))^{-1} \left(1-\frac{z_1}{z_2}\right)^{\hbar\delta P} \cdot e^{\int D^{(1)}(z_1)dz_1}z_1^{D^{(1)}_1+\hbar\delta P}(z_2-z_1)^{-\hbar\delta P}z_2^{\hbar\delta P}
\end{align}
Recall that 
\[\widehat{Q}=1+Q_1(t)z+Q_2(t)z^{2}+\cdot\cdot\cdot,\]
with each $Q_m(t)\in \left(
\Uph/\mathfrak m^N
\otimes\End(\mathbb C^n)^{\otimes 2}\right)[t,(t-1)^{-1}].$ Expanding $Q_m(t)=\sum_{k\le m}Q_{m,k}t^k$ at $t=\infty$, and replacing $t$ by $t=z_2/z_1$ we get
\begin{equation}\label{Qz1z2}
    \widehat{Q}\left(z_1,\frac{z_2}{z_1}\right)=\sum_{m\ge 0}\sum_{k\le m}Q_{m,k}z_2^k z_1^{m-k}.
\end{equation}
Note that $\widehat{Y}\left(z_1,\frac{z_2}{z_1}\right)$ satisfies the equation \eqref{kKZ2}.
Comparing the $z_1^0$ coefficient of the equation \eqref{kKZ2} shows that $
\left.
\widehat Q\!\left(z_1,\frac{z_2}{z_1}\right)
\right|_{z_1=0}
=
1+\sum_{m\ge1}Q_{m,m}z_2^m$ satisfies the equation \[
\begin{aligned}
\frac{d\left.
\widehat Q
\right|_{z_1=0}}{dz_2}
=
\left(
\frac{u^{(2)}}{z_2^{p+1}}
+\sum_{r=2}^{p}\frac{\hbar T_{[r]}^{(2)}}{z_2^r}
-\frac{\hbar T^{(2)}}{z_2}
\right)\cdot \left.
\widehat Q
\right|_{z_1=0}-\left.
\widehat Q
\right|_{z_1=0}\cdot
\left(
\frac{u^{(2)}}{z_2^{p+1}}
+\sum_{r=2}^{p}\frac{D_r^{(2)}}{z_2^r}
+\frac{D_1^{(2)}}{z_2}
\right).
\end{aligned}
\]
Since by definition, $\widehat{H}^{(2)}(z_2)=1+O(z_2)$ satisfies the same equation with the same initial condition, we have 
\begin{equation}\label{Qz1z22}
\left.
\widehat Q\!\left(z_1,\frac{z_2}{z_1}\right)
\right|_{z_1=0}=\widehat{H}^{(2)}(z_2).\end{equation}

Now by \eqref{Qz1z2} and \eqref{Qz1z22}, we have $\widehat{Q}\left(z_1,\frac{z_2}{z_1}\right)=\widehat{H}^{(2)}(z_2)+O(z_1)$. Together with the expansion 
\begin{equation}\label{texp}
\left(1-\frac{z_1}{z_2}\right)^{\hbar\delta P}=1+O(z_1),
\end{equation}
we can rewrite \eqref{Wztform} in the form of \eqref{formalKZ1}, where at this stage each coefficient $K_m(z_2)=\sum_{i\le m-1} K_{m,i} z_2^{-i}$ is a Laurent series. 
The rest is to show that for each $m$, $K_{m,i}=0$ for $i<0$.

Let $\widehat{H}(z)=1+O(z)$ be the formal power series part of $F(z)$ from Theorem
\ref{oneformal}.
Let us lift $\widehat{H}(z_1)$ to the first tensor factor to get $\widehat{H}^{(1)}(z_1)$, and define
\[(\widehat{H}^{(1)}(z_1))^{-1}\cdot P\cdot \widehat{H}^{(1)}(z_1)=\sum_{r\ge 0}P_r z_1^r,\]
with $P_0=P$.
Then we take $\widehat{V}(z_1;z_2)=1+\sum_{r\ge1}V_r(z_2)z_1^r$ such that 
\begin{equation}\label{WHV}
_1\widehat{W}(z_1;z_2)=\widehat{H}^{(1)}(z_1)\cdot \widehat{V}(z_1;z_2)\cdot
e^{\int D^{(1)}(z_1)dz_1}
z_1^{D_1^{(1)}+\hbar\delta P}
(z_2-z_1)^{-\hbar\delta P}z_2^{\hbar\delta P}
\end{equation}
solves \eqref{kKZ1}. 
Expanding 
\begin{equation}\label{Pexp}
\frac{P}{z_2-z_1}=\frac{P}{z_2}\cdot \left(1+\sum_{k\ge 1}\left(\frac{z_1}{z_2}\right)^k\right) \ \ \text{for } |z_1|<|z_2|
\end{equation} 
and 
comparing the coefficient in \eqref{kKZ1} give
\begin{align}\nonumber
[u^{(1)},V_{m+1}]
=&(m+1-p)V_{m+1-p}
+\sum_{a=1}^{p}
[V_{m+a-p},D_a^{(1)}] \\ \label{Vrecursion}
&-\hbar
\sum_{l=-1}^{m-p}
z_2^{-(l+1)}
\left(
\sum_{r=0}^{m-p-l}
P_rV_{m-p-l-r}
-
V_{m-p-l}\delta P
\right).
\end{align}
Here we take the convention $V_0=1$ and $V_r=0$ for $r<0$, $P_0=P$.

Let us assume inductively that for all $j\le m$, $V_{j}$ has the degree bound ${\rm deg}_{z_2^{-1}}(V_j)\le j-1$.
Then for $(ijkl)$-entries of $V_{m+1}$ with $i\ne j$, the left hand side of \eqref{Vrecursion} is $(u_i-u_j)(V_{m+1})_{ijkl}$, and the terms on the right are determined by lower orders $V_j$
uniquely. For the degree bound, just notice that by induction every term $V_j$, for $j\ge 1$, on the right has $z_2^{-1}$ degree less than or equal to $j-1$, and $V_0=1$.

For the $(iikl)$-entry of $V_{m+1}$, taking the
$(iikl)$-entry of \eqref{Vrecursion} $p$ steps later
(replace $m$ in \eqref{Vrecursion} by $m+p$) gives
\begin{equation*}
(m+1)(V_{m+1})_{iikl}-{\rm ad}_{(D_1)_{ii}}(V_{m+1})_{iikl}-(\hbar PV_{m+1}-V_{m+1}\hbar \delta P)_{iikl}
=
\text{terms involving only $V_j$ with $j\le m$}.
\end{equation*}
Here we use Lemma
\ref{Dcommute} to eliminate the contribution of terms $V_i$ with $i\ge m+1$,
\[
[V_{m+a},D_a^{(1)}]_{iikl}=[(V_{m+a})_{iikl}, (D_a)_{ii}]
=0, \qquad a=2,\ldots,p.
\]
The terms $D_1$ and $\hbar P, \hbar \delta P$ multiplying $V_{m+1}$ belong to $\frak m$ and preserve the space of polynomials in $z_2^{-1}$
of degree at most
$m$. Hence the left hand side is seen as action of an invertible operator on $V_{m+1}$, which is $(m + 1){\rm id}$ plus a nilpotent operator, in
every finite quotient. The right hand side has degree at most $m$ by the induction hypothesis. It proves the
required degree bound for the $(iikl)$-entries.

Therefore, the coefficients of the formal power series part of ${}_1\widehat W(z_1;z_2)$ are
\[K_m(z_2)=\sum_{i=0}^mH_i^{(1)}V_{m-i}(z_2).\]
Since $H_i^{(1)}$ is independent of $z_2$, ${\rm deg}_{z_2^{-1}}(K_m)\le m-1$ as claimed.
\end{proof}
\begin{rmk}\label{tbranchlem}
In the proof of Proposition \ref{YWFfactorization}, we have used the expansion \ref{texp} at $z_1/z_2=0$. Accordingly, the factorizations at $t=0$ and $t=\infty$ fix, respectively,
the branches of $\log(1-t)$ and $\log(t-1)$ by
\[
\left.\log(1-t)\right|_{t=0}=0
\quad \text{ and } \quad 
\left.\log(1-t^{-1})\right|_{t^{-1}=0}=0.
\]
If the $\log(t-1)-\log(1-t)=(2\nu+1)\pi\mathi$ 
for some $\nu\in\mathbb Z$, then the last factor in \eqref{Y=W1}
\[
\left(\frac{t-1}{1-t}\right)^{\hbar\delta P}
=
e^{(2\nu+1)\pi\mathi\hbar\delta P}.
\]
\end{rmk}
\begin{rmk}
From the proof, we see that the factorization above also determines the boundary value of
$\widehat Q$ at $t=\infty$. That is
\[
 \left.\widehat Q(z_1,z_2/z_1)\right|_{z_1=0}= \left.\widehat Q(z_2/t,t)\right|_{t^{-1}=0}
 =
 \widehat H^{(2)}(z_2).
\]
Here the restriction means taking the constant coefficient in
$t^{-1}=z_1/z_2$, with $z_2$ retained as a formal variable.
\end{rmk}

\subsection{Factorization of the formal solution $\widehat{Y}(z,t)$ as $t\rightarrow 0$}
Parallel to Proposition \ref{YWFfactorization}, we have
\begin{pro}\label{t2fact}
 After substituting $z=z_1$ and $t=z_2/z_1$, we have
\begin{equation}\label{Y=W2}
\widehat{Y}(z,t)= {_2\widehat{W}}(z_2;z_1)\widehat{F}^{(1)}(z_1), \text{ with } \ 0<|t|<1.
\end{equation}
Here 
\begin{equation}
_2\widehat{W}(z_2;z_1)=\left(1+\sum_{m\ge 1}K'_m(z_1)z_2^{m}\right)\cdot e^{\int D^{(2)}(z_2)dz_2}z_2^{D^{(2)}_1+\hbar\delta P}\left(z_1-{z_2}\right)^{-\hbar\delta P} z_1^{\hbar\delta P}, 
\end{equation}
is the unique formal solution of the equation \eqref{kKZ2} with the prescribed form, and each coefficient \begin{equation}\label{K'm}
K'_m(z_1)=\sum_{i=0}^{m-1} K'_{m,i} z_1^{-i}\in \wUph\otimes {\rm End}(\IC^n)\otimes {\rm End}(\IC^n)[z_1^{-1}]
\end{equation} 
is a polynomial of degree at most $m-1$ in $z_1^{-1}$, and  $\widehat{F}^{(1)}(z_1)$ is the unique formal solution of the equation 
\begin{equation}\label{KZin1}
 \frac{{d F}}{dz_1}= \left( \frac{u^{(1)}}{z_1^{p+1}}+\frac{\hbar T_{[p]}^{(1)}}{z_1^p}+\cdots+\frac{\hbar T_{[2]}^{(1)}}{z_1^2}-\frac{\hbar T^{(1)}}{z_1}\right)\cdot F,
\end{equation}
\end{pro}

Note that $\widehat{F}^{(1)}(z_1)$ is the extension of the formal solution $\widehat{F}(z=z_1)$, given in \eqref{formalsum}, from $\Uph\otimes {\rm End}(\IC^n)$ to $\Uph\otimes {\rm End}(\IC^n)\otimes {\rm End}(\IC^n)$. It follows from the similar argument in the proof of Proposition \ref{YWFfactorization} that
the product $_2\widehat{W}(z_2;z_1)\hat{F}^{(1)}(z_1)$ 
 takes the form 
\begin{align}\nonumber
&_2\widehat{W}(z_2;z_1)\widehat{F}^{(1)}(z_1)\\
=&\left(1+\sum_{i=0}^\infty \left(\sum_{j=-i+1}^\infty g_{ij} z_1^{j}\right) z_2^{i}\right)\cdot e^{\int D^{(1)}(z_1)dz_1}e^{\int D^{(2)}(z_2)dz_2}z_1^{D_1^{(1)}+\hbar\delta P}z_2^{D_1^{(2)}+\hbar\delta P}(z_1-z_2)^{-\hbar\delta P},
\end{align}
where each coefficient $g_{ij}\in \Uph\otimes {\rm End}(\IC^n)\otimes {\rm End}(\IC^n)$. The identity \eqref{Y=W2} is then understood after expanding each rational coefficient $Q_m(t)$ of $\widehat{Y}$ for $m\ge 1$ as a Laurent series at $t=0$, and substituting $z=z_1$ and $t=z_2/z_1$.

\subsection{Factorization of the formal solution $\widehat{Y}(z,t)$ as $t\rightarrow 1$}\label{factY1}
Let
\begin{equation}
    \widehat{X}(\omega)=(1+\sum_{m=1} C_m \omega^m) \cdot e^{\frac{-u^{(2)}}{p\omega^p}}\omega^{p\hbar \delta P}
\end{equation}
be the unique formal solution of the prescribed form of the equation 
\begin{align}\label{Xeq}
\frac{{dX}}{d \omega}=\left(\frac{ u^{(2)}}{\omega^{p+1}}+p\frac{\hbar P}{\omega}\right)X.
\end{align}
(This equation is exactly solvable and its Stokes matrices have a closed expression, see Lemma \ref{PR0}.)

We will construct, in the proof of Proposition \ref{t3fact}, the unique formal solution $\widehat{U}(z,1)=\widehat I_0(z)\cdot E_U(z)$, with the regular part $\widehat I_0(z)=1+O(z)$, of the equation \eqref{ztkKZ1} at $t=1$ \begin{align}\label{Uzeq}
\frac{d \widehat{U}(z,1)}{dz}= \left(\frac{u^{(1)}+u^{(2)}}{z^{p+1}}+\hbar \sum_{j=2}^p\frac{T_{[j]}^{(1)}+T_{[j]}^{(2)}}{z^j}-\hbar \frac{T^{(1)}+T^{(2)}}{z}+\frac{\hbar P}{z}\right)\cdot \widehat{U}(z,1)-\widehat{U}(z,1)\cdot \frac{p\hbar P}{z},
\end{align}
where the singular term
\begin{align}
 E_U(z)=
 \exp\!\left(-\frac{u^{(1)}+u^{(2)}}{pz^p}
 -\sum_{r=2}^p
 \frac{D_r^{(1)}+D_r^{(2)}}{(r-1)z^{r-1}}\right)\cdot
 z^{D_1^{(1)}+D_1^{(2)}-(p-1)\hbar\delta P}.
\end{align}

\begin{pro}\label{t3fact}
We have
\begin{equation}\label{Y=UX}
    \widehat{Y}(z,t)=\widehat{U}(z,t)\cdot\widehat{X}\left((t^{-p}-1)^{-1/p}z\right)\cdot p^{\hbar\delta P}, \text{ as a formal identity near } \ t=1.
\end{equation}
where
\begin{equation}\label{Uformal}
\widehat U(z,t)
 =\left(\widehat I_0(z)
       +\sum_{k\geq1}\widehat{I}_k(z)(t-1)^k\right)E_U(z)
\end{equation}
is the unique formal solution of $t$-equation 
\begin{align}\label{Uteq}
\frac{d U}{d t}
=
\left(
\frac{u^{(2)}}{t^{p+1}z^p}
+
\sum_{r=2}^{p}
\frac{\hbar T_{[r]}^{(2)}}{t^rz^{r-1}}
-\frac{\hbar T^{(2)}}t
+\frac{\hbar P}{t}
+\frac{\hbar P}{1-t}
\right)\cdot U-
U\cdot \left(
\frac{u^{(2)}}{t^{p+1}z^p}
+\frac{p\hbar P}{t(1-t^p)}
\right).
\end{align}
with the initial condition
\begin{equation}\label{Uinitial}
\widehat{U}(z,1)=\widehat I_0(z)\cdot E_U(z).
\end{equation}
Moreover, for every $N\geq1$ and $k\geq1$, the reduction
$\widehat{I}_k^{[N]}(z)$ of $\widehat{I}_k(z)$ modulo $\mathfrak m^N$ satisfies the bound
\begin{equation}\label{Ikbound}
 \widehat{I}_k^{[N]}
 \in
 \sum_{r=0}^{N-1}
 z^{-(p-1)r}
 \left(\mathfrak m^r/\mathfrak m^N\right)\otimes {\rm End}(\mathbb{C}^n)^{\otimes 2}[[z]].
\end{equation}
\end{pro}

\begin{proof}
Let us fix $N$ and work modulo $\mathfrak m^N$. To prove that the $\widehat{Y}^{[N]}\cdot (p^{-\hbar \delta P})^{[N]} \cdot(\widehat{X}^{[N]})^{-1}$ has the required form, let us work in the algebra of Laurent series in two variables
\[
 \mathcal L_N=\left(\Uph/\mathfrak m^N
 \otimes\End(\mathbb C^n)^{\otimes2}\right)[[x,y]][y^{-1}], \quad \text{with } x=(t^{-p}-1)^{1/p} \text{ and } y=(t^{-p}-1)^{-1/p}z.
\]
Let us first show that the involved formal series are in $\mathcal L_N$.
Let us write $\widehat Y$ and $\widehat X$ as products of regular and singular terms
\begin{equation}\label{YXfactors}
 \widehat Y^{[N]}=\widehat Q^{[N]}\cdot E_Y^{[N]},
 \qquad
 \widehat X^{[N]}=\widehat C^{[N]}\cdot E_X^{[N]},
\end{equation}
where, after the substitution
$\omega=(t^{-p}-1)^{-1/p}z$,
\begin{equation}
 E_X(z,t)=
 \exp\!\left(-\frac{(t^{-p}-1)u^{(2)}}{pz^p}\right)\cdot 
 z^{p\hbar\delta P}(t^{-p}-1)^{-\hbar\delta P}.
\end{equation}

On the one hand, recall that $\widehat{Q}^{[N]}=1+\sum_m Q_m^{[N]}(t)z^m$, where each $Q_m^{[N]}(t)$ is a Laurent series in $t-1$.
By the pole order estimate \ref{Q1degree}, ${\rm ord}^{\rm pole}_{t=1}Q_m^{[N]}
 \leq\left\lfloor\frac mp\right\rfloor.$ We can write
\[\widehat{Q}^{[N]}=1+\sum_m \sum_{j\ge -\left\lfloor\frac mp\right\rfloor}Q_{m,j}^{[N]}\cdot (t-1)^jz^m.\]
Since
$t-1=-\frac{x^p}{p}+O(x^{2p})$, 
a term $z^m(t-1)^{j}$ in $Q^{[N]}$ becomes 
$x^{m+pj}y^m$
times a power series in $x^p$. Thus
\begin{equation}\label{Qregular}
 \widehat Q^{[N]}(z,t)=\widehat Q^{[N]}(xy,t(x))\in \Uph/\mathfrak m^N
 \otimes\End(\mathbb C^n)^{\otimes2} [[x,y]],\qquad \left.Q^{[N]}(xy,t(x))\right|_{y=0}=1.
\end{equation}

On the other hand, the recursion relation 
\[
 [u^{(2)},C_m]
 =
 (m-p)C_{m-p}
 +pC_{m-p}\hbar\delta P
 -p\hbar P C_{m-p},
\]
for the equation \eqref{Xeq} gives
\begin{equation}
C^{[N]}_m=0 \ \text{if }p\nmid m,
 \qquad
 \widehat C^{[N]},  (\widehat C^{[N]})^{-1}\in 1+y^p\cdot \Uph/\mathfrak m^N
 \otimes\End(\mathbb C^n)^{\otimes2} [[y^p]].
\end{equation}

Furthermore, a direct comparison of
the singular terms gives
\begin{equation}\label{Phi0}
E^{[N]}_Y (p^{-\hbar\delta P})^{[N]} (E^{[N]}_X)^{-1}=\Phi^{[N]} E_U^{[N]},
\end{equation}
where $\Phi$ is the diagonal matrix
\begin{equation}\label{Phidefine}
\Phi=
 \exp\!\left(-\sum_{r=2}^p
 \frac{(t^{1-r}-1)D_r^{(2)}}{(r-1)z^{r-1}}\right)\cdot
 t^{D_1^{(2)}+\hbar\delta P}
 \cdot\left(\frac{t^{-p}-1}{p(1-t)}\right)^{\hbar\delta P}.
\end{equation}
Since
$t=(1+x^p)^{-1/p}$ and $z=xy$, we have for $2\leq r\leq p$,
\[
 \frac{t^{1-r}-1}{z^{r-1}}
 \in
 x^{p-r+1}y^{1-r} \cdot \mathbb C[[x^p]].
\]
In particular, every nonconstant term has strictly positive
$x$ degree. We also have
\[
 \log t\in x^p\mathbb C[[x^p]], \quad \text{and } \frac{t^{-p}-1}{p(1-t)}
 =
 \frac{x^p}
 {p\left(1-(1+x^p)^{-1/p}\right)}
 \in1+x^p\mathbb C[[x^p]].
\]
Since $D_r,D_1,\hbar\delta P\in\mathfrak m$, the corresponding
exponential series in \eqref{Phidefine} truncate modulo $\frak m^N$. Hence
\begin{equation}\label{Phiregular}
 \Phi^{[N]}\in 1+x\cdot( \Uph/\frak m^N)[[x]][y^{-1}]\subset \mathcal L_N
\end{equation}
Here note that $\left.\Phi^{[N]}(x,y)\right|_{x=0}=\left.\Phi^{[N]}(z,t)\right|_{t=1}=1.$ 

Therefore, we can introduce
\begin{equation}\label{Idefinition}
\widehat U^{[N]}=\widehat I^{[N]}\cdot E_U^{[N]} \quad \text{with } \widehat I^{[N]}
 :=
 \widehat Q^{[N]}\Phi^{[N]}
 (\widehat C^{[N]})^{-1}\in \mathcal{L}_N
 .
\end{equation}
Equations \eqref{Idefinition} and \eqref{Phi0} immediately give
\begin{equation}
\widehat U^{[N]}\widehat X^{[N]}p^{\hbar\delta P}= \widehat Q^{[N]}\Phi^{[N]} 
 (\widehat C^{[N]})^{-1} E_U^{[N]} \cdot \widehat C^{[N]} E_X^{[N]}p^{\hbar\delta P}=\widehat Q^{[N]}\Phi^{[N]} 
 E_U^{[N]} E_X^{[N]}p^{\hbar\delta P}=\widehat Y^{[N]}.
\end{equation}
Here the second identity uses the fact that $\widehat C^{[N]}E_U^{[N]}=E_U^{[N]}\widehat C^{[N]}$. It is because the coefficients of $\widehat C$ are obtained recursively from
$u^{(2)},P$ and $\delta P$.  These operators commute with the operators
$u^{(1)}+u^{(2)}$, $D_r^{(1)}+D_r^{(2)}$, and $\delta P$ defining $E_U$.
It then follows that $\widehat C^{[N]}$ commutes with
$E_U^{[N]}$.

It remains to prove the required form \eqref{Ikbound} of the coefficients of $\widehat I^{[N]}$. In terms of the coordinates $x$ and $y$, by the chain rule the system \eqref{ztkKZ1}-\eqref{ztkKZ2} for $\widehat{Y}$ becomes
\begin{align}\label{Ixeq}
 \frac{\partial \widehat I}{\partial x}
 &=\left[\frac{u^{(1)}+u^{(2)}}{x^{p+1}y^p},\widehat I\right]
 +\sum_{ l=2}^p
 \frac{\hbar(T_{[ l]}^{(1)}+t^{p+1- l}T_{[ l]}^{(2)})\widehat I
       -\widehat I(D_l^{(1)}+D_l^{(2)})}
      {x^l y^{l-1}}\nonumber\\
 &-\frac{(\hbar T^{(1)}+t^p\hbar T^{(2)}
                  +(t+\cdots+t^{p-1})\hbar P)\widehat I}{x}
 -\frac{\widehat I(D_1^{(1)}+D_1^{(2)}-(p-1)\hbar\delta P)}{x},
\\
\label{Iyeq}
 \left(\frac{\partial}{\partial y}-\frac{x}{y}\frac{\partial}{\partial x}\right)\widehat I
 &=\left[\frac{u^{(2)}}{y^{p+1}}+\frac{p\hbar P}{y},\widehat I\right]
 \nonumber\\
 &+\left(
 \sum_{ l=2}^p
 \frac{t^{p+1- l}x^{p- l+1}\hbar T_{[ l]}^{(2)}}{y^l}
 -\frac{t^px^p\hbar T^{(2)}}{y}
 +\frac{(1+t+\cdots+t^{p-1}-p)\hbar P}{y}
 \right)\widehat I.
\end{align}
Here $t$ denotes the function
$t(x)=(1+x^p)^{-1/p}$.

Recall that $\widehat I^{[N]}(x,y)$ has nonnegative $x$ powers. Let us first determine the $x^0$ coefficient $\widehat I^{[N]}(0,y)$ of $\widehat I^{[N]}(x,y)$. By \eqref{Qregular}, \eqref{Phiregular} and the fact that $(\widehat C^{[N]}(y))^{-1}$ is a nonnegative power series in $y$, we have
\begin{equation}\label{x0coeff1}
\text{the $x^0$ coefficient of } I^{[N]}(x,y)=\widehat Q^{[N]}(z(xy),t(x))\Phi^{[N]}(x,y)(\widehat C^{[N]}(y))^{-1} \ \text{ is in $1+y\cdot \Uph/\frak m^N[[y]]$.}
  \end{equation}
Then comparing the $x^0$ coefficient in the equation \eqref{Iyeq} gives
\[
 \frac{\partial \widehat I^{[N]}(0,y)}{\partial y}=\left[\frac{u^{(2)}}{y^{p+1}}+\frac{p\hbar P}{y},\widehat I^{[N]}(0,y)\right],\qquad \left.\widehat I^{[N]}(0,y)\right|_{y=0}=1.
\]
One checks that the induction modulo powers of $\mathfrak m$ gives 
\begin{equation}\label{xyN-face}
 \widehat I^{[N]}(0,y)=1.
\end{equation}

Now let us check the constraint that the equation \eqref{Ixeq} put on the formal solution $\widehat{I}^{[N]}$.
Using the expansion $t(x)=(1+x^p)^{-1/p}=1-\frac{x^p}{p}+O(x^{2p})$, we see that in \eqref{Ixeq} the scalar coefficients other than the leading commutator are of
monomials
\begin{equation}\label{xyform}
 x^{pk- l}y^{1- l},\qquad
 1\le l\le p,\quad k\ge 0,
\end{equation}
and their operator coefficients $\hbar T_{[m]}, \hbar P, D_m$ and so on belong to
$\frak m$. 
We can then write down the underlying recursion relations and use the coefficient expansion \eqref{xyform} to show that, by induction on the $\frak m$ adic order, 
\begin{equation}
 \widehat I^{[N]}=1+
 \sum_{r=1}^{N-1}
 \sum_{\substack{a\ge 1,\ j\ge 0\\a-pj+(p-1)r\ge 0}}
 I_{a,j,r}^{[N]}x^a y^{a-pj},
 \qquad
 I_{a,j,r}^{[N]}\in
 (\mathfrak m^r/\mathfrak m^N)\otimes\End(\C^n)^{\otimes2}.
\end{equation}

In the end, let us rewrite $\widehat{I}^{[N]}$ in $z,t$ variables. Note that $ x^a y^{a-pj}=z^{a-pj}(t^{-p}-1)^j$. For a fixed power $(t-1)^k$, only finitely many $j\ge 0$ in \eqref{t-1expansion} can contribute.
Thus 
\eqref{t-1expansion} gives
\begin{equation}\label{t-1expansion1}
 \widehat I^{[N]}
 =
\widehat{I}_{0}^{[N]}(z)+
 \sum_{k\geq1}\widehat{I}_{k}^{[N]}(z)(t-1)^k,
\end{equation}
with $\widehat{I}_{0}^{[N]}(z)=1+O(z)$ and the bound \eqref{Ikbound}. 

We have that $\widehat{U}^{[N]}=\widehat{I}^{[N]}\cdot E_U^{[N]}$ satisfies the equation \eqref{Uteq}. Since the residue of \eqref{Uteq} at $t=1$ is
$-\operatorname{ad}_{\hbar P}$, the power series solution \eqref{t-1expansion1} implies that $[\hbar P,\widehat{I}_0^{[N]}E_U^{[N]}]=0$.
Taking $t=1$ (i.e., the constant coefficient in in $t-1$) in the $z$-equation of $\widehat{I}^{[N]}$, we see that the initial condition $\widehat{U}(z,1)=\widehat{I}_{0}^{[N]}\cdot E_U^{[N]}$ satisfies the equation \eqref{Uzeq}. The uniqueness of the formal
solution follows from a standard argument.

The above construction is compatible with
$\Uph/\mathfrak m^{N+1}\to\Uph/\mathfrak m^N$.  Passing to the inverse
limit proves \eqref{Uformal} and \eqref{Y=UX}.
\end{proof}
\begin{rmk}
Note that the $t$ equation \eqref{Uteq} is a Fuchsian equation, i.e., has a form with regular singularity
\begin{equation}
 \frac{\partial U}{\partial t}
 =-\frac{[\hbar P,U]}{t-1}+\mathcal R(t;z)(U),
\end{equation}
where $\mathcal R(t;z)$ is a holomorphic operator at $t=1$. Since $[\hbar P,\widehat I_0(z)E_U(z)]=0$ (which follows from the fact that the coefficient of the equation \eqref{Uzeq} and regularized initial data of the solution $U_0(z):=U(z,1)=\widehat I_0(z) E_U(z)$ commute with $\hbar P$), it has a unique power series solution $U_0(z)+\sum_{k\geq1}U_k(z)\cdot (t-1)^k$ with initial condition $U_0(z)$. This solution coincides with $\widehat{U}$ in the above proposition. That is $U_k(z)=\widehat{I}_k(z) E_U(z)$.
\end{rmk}

\section{Stokes phenomenon for the two variable system and a proof of Theorem \ref{thmRLL}}\label{SPDE2}
In this section, we study the resummation (i.e., the order-$p$ Borel-Laplace tranform) of the formal solution $\widehat{Y}(z_1,z_2):=\widehat{Y}(z,t)$ of the system \eqref{kKZ1}-\eqref{kKZ2}. The resulting holomorphic solutions depend on the summation directions and regions. In particular, for two summation directions $d_1, d_2\notin {\rm aSR}(u)$ of $z_1,z_2$, we get holomorphic solutions
\begin{itemize}
    \item $Y_{d_1,d_2}(z_1,z_2)$, if $d_1,d_2$ are not bounded by two adjacent anti-Stokes rays. In the following, it is denoted by $Y_{d,\mathbb{D}_k}$, where $\tau_i<d=d_1<\tau_{i+1}$ and $\tau_{i+k}<\frac{(k+1)\pi}{p}<\tau_{i+k+1}$ for some integers $i$ and $k\ne 0$.

    \item $Y^{(-)}_{d_1,d_2}(z_1,z_2)$ and $Y^{(+)}_{d_1,d_2}(z_1,z_2)$ if $d_1,d_2$ are bounded by two adjacent. In the following, they are denoted by $Y_{d,\mathbb{D}_{(0,1)}}$ and $Y_{d,\mathbb{D}_{(1,\infty)}}$ respectively, where $\tau_i<d=d_1,d_2<\tau_{i+1}$.
\end{itemize}
In Corollary \ref{mainpro}, we obtain the factorization of these holomorphic solutions, via the factorization of the underlying formal solution, at $t=z_2/z_1=0,1$ and $\infty$. It enables us to derive the explicit connection formula or Stokes matrices for the solutions. In Section \ref{proofThm}, by comparing the connection formula along homotopy paths in the space $z_1,z_2\in (\mathbb{C}\times \mathbb{C})\setminus\{z_1=0,z_2=0,z_1=z_2\}$, we finally prove the identities in Theorem \ref{thmRLL}.

\subsection{Resummation of the formal solution $\widehat{Y}(z;t)$}\label{sec34}
We now consider the resummation of the formal solution $\widehat{Y}(z,t)$ with respect to the variable $z$. The resummation procedure is as in Proposition \ref{CanSol}: following Proposition \ref{t1fact}, after the reduction modulo $\frak m^N$, the coefficients $Q^{[N]}_m(t)$ of the power series part \[\widehat{Q}^{[N]}(z,t)=1+\sum_{m\ge 1} Q^{[N]}_m(t) z^m \] are rational functions of $t$ valued in finite dimensional space. For any fixed $t$ and every $N$, we take the Borel-Laplace $p$-sum of $\widehat{Q}^{[N]}(z,t)$ along an admissible direction $d$, and then pass to the inverse limit to get $Q_d(z,t)$.

A priori, for any fixed $t$, the set of anti-Stokes rays determined by the leading coefficient ${u^{(1)}+t^{-p} u^{(2)}}$ of equation \eqref{ztkKZ1} is
\[{\rm aSR}(u^{(1)}+t^{-p}u^{(2)})=\left\{\frac{1}{p}\mathrm{Arg}(u_i-u_j+t^{-p}(u_k-u_l))+\frac{2m\pi}{p}~:~ i,j,k,l=1,...,n, \ m\in \mathbb{Z}\right\}.\]  However, if the coefficient matrix is special, some of these apparent anti-Stokes rays
are genuine, while others carry trivial Stokes factors and are therefore admissible. The following proposition identifies
those rays.

For fixed \(u\in\mathfrak h_{\rm reg}\) and
\(d\notin{\rm aSR}(u)\), we introduce two collections of
real curves or walls in the $t$-plane
\begin{align}\label{Ct0}
\mathcal C(u;d)_0
&:=
\bigcup_{i\ne j}
\left\{
 \ t
 \ \middle|\
 e^{-\mathi pd}\cdot t^{-p}(u_i-u_j)\in\mathbb R_{>0}
\right\},
\\ \label{Ct1}
\mathcal C(u;d)_1
&:=
\bigcup_{i\ne j}
{\rm Component}_{t=1}
\left\{
 \ t
 \ \middle|\
 e^{-\mathi pd}\cdot (1-t^{-p})(u_i-u_j)\in\mathbb R_{>0}
\right\}.
\end{align}
Here ${\rm Component}_{t=1}$ means the connected
component of the real curve whose closure contains $t=1$.

\begin{pro}\label{Yztresum}
For fixed $u\in\h_{\rm reg}$ and $d\notin {\rm aSR}(u)$, the lifts of the walls in $\mathcal{C}(u;d)_0$ and $\mathcal{C}(u;d)_1$
cut the universal cover space of $\mathbb{C}\setminus\{0,1\}$ into connected components. For each such connected component $\D$, there exists a unique holomorphic solution $Y_{d,\D}(z,t) \in \wUph\otimes {\rm End}(\IC^n)\otimes {\rm End}(\IC^n)$ of the system \eqref{ztkKZ1}-\eqref{ztkKZ2} on the universal cover of $(\mathbb{C}\setminus\{0\})\times (\mathbb{C}\setminus \{0,1\})$, such that for all $t\in \D$
\begin{equation}\label{YDasy}
    Y_{d,\D}(z,t)\cdot e^{-\int \left(D^{(1)}(z)+tD^{(2)}(zt)\right)dz}z^{-(D_1^{(1)}+D_1^{(2)}+\hbar\delta P)}t^{-D_1^{(2)}-\hbar\delta P}(1-t)^{\hbar\delta P} \sim \widehat{Q}(z;t), 
\end{equation}
as $z\rightarrow 0$ within the sector $-\frac{\pi}{2p}< {\rm arg}(z)-d<\frac{\pi}{2p}$. 
\end{pro}
\begin{proof}
The holomorphic solution $Y_{d,\mathbb{D}}$ can be constructed via the order-$p$ Borel-Laplace transform along the direction $d$ for $t\in \mathbb{D}$. For this purpose, we only need to prove, for all $t\in \mathbb{D}$, the summability of the formal power series part $\widehat{Q}$ along $d$. Recall that summability requires analytic continuation of the order-$p$ Borel transform from $0$ along the direction $d$, i.e., the ray $e^{\mathi pd}\cdot \mathbb{R}_{>0}$,
together with an exponential growth estimate. The exponential growth estimate is standard, in the following, we only prove the analytic continuation part.

Note that along a direction, $\widehat Q$ is not summable only if the analytic continuation of the order-$p$ Borel transform from $0$ to infinity along the direction encounters a singularity. Thus we only need to identify the singularities of the order-$p$ Borel transform. 
To avoid the ramification in the order-$p$ Borel transform, let us divide the formal power series in $z$ into $p$ power series in $\omega=z^p$ and conduct the ordinary Borel transform.
Let us write
\begin{equation}
 \widehat Q (z,t)
 =\sum_{r=0}^{p-1}z^r\widehat Q_{(r)} (z^p,t),
 \qquad \text{with} \quad
 \widehat Q_{(r)}(w,t)
 :=\sum_{k\ge 0}Q_{r+pk}(t)w^k,
\end{equation}
and denote the (ordinary order one) Borel transform of $\widehat Q_{(r)} (w,t)$ in the variable $w=z^p$ by
\begin{equation}\label{Btrans}
\mathcal{B}_{(r)} (\xi,t):=\sum_{k\ge 0}Q_{r+pk} (t)
  \frac{(\xi/p)^k}{k!},
 \qquad 0\leq r<p.
\end{equation}
These $p$ many ordinary Borel
transforms determine the usual order-$p$ Borel transform of 
$\widehat Q(z,t)$
 by
the standard beta fractional integrals and the finite recombination
over $r$. 

The series
$\mathcal B_{(r)}^{[N]}(\xi,t)$ initially defines a convergent
holomorphic germ along the zero section $\{\xi=0\}\times (\mathbb{C}\setminus\{0,1\})\subset \mathbb{C}_\xi\times \mathbb{C}_t$. 
In the following, we study their analytic continuation via the associated differential equations. 

First applying \eqref{Btrans} directly to the
$z$ and $t$ recursions \eqref{Qzrecur}-\eqref{Qtrecur} and taking the $(ijkl)$-entry, we get the equations for the Borel transforms
\begin{align}
&\left(
 (u_i-u_j+t^{-p}(u_k-u_l)-\xi)\frac{\partial}{\partial\xi}
 -\frac{r}{p}
 \right)(\mathcal  B_{(r)} )_{ijkl}
 \nonumber\\\nonumber
=&\frac{1}{p}\Biggl(
 \hbar(T^{(1)}+T^{(2)}-P)\mathcal  B_{(r)} 
 +\mathcal  B_{(r)} (D_1^{(1)}+D_1^{(2)}+\hbar\delta P) \\ \label{zBoreleq} & \qquad +\sum_{m=2}^p\mathcal J_{r,m} 
       (D_m^{(1)}+t^{1-m}D_m^{(2)})
 -\hbar\sum_{m=2}^p
       (T_{[m]}^{(1)}+t^{1-m}T_{[m]}^{(2)})
       \mathcal J_{r,m} 
 \Biggr)_{ijkl}, 
 \end{align}
 and
 \begin{align}\nonumber
&\left(
 \frac{\partial}{\partial t}-p(u_k-u_l)t^{-p-1}\frac{\partial}{\partial\xi}
 \right)(\mathcal  B_{(r)} )_{ijkl}
 \\ \nonumber
=&\Biggl(
 \sum_{m=2}^p{\hbar T_{[m]}^{(2)}\over t^m}
       \mathcal J_{r,m} 
 -{\hbar(T^{(2)}-P)\over t}\mathcal  B_{(r)} 
 +{\hbar P\over1-t}\mathcal  B_{(r)} 
 \\  \label{tBoreleq}
&\qquad
 -\sum_{m=2}^p\mathcal J_{r,m} 
       {D_m^{(2)}\over t^m}
 -\mathcal  B_{(r)} {D_1^{(2)}+\hbar\delta P\over t}
 -\mathcal  B_{(r)} {\hbar\delta P\over1-t}
 \Biggr)_{ijkl}.
\end{align}
Here \begin{equation}\label{shiftJ}
 \begin{split}
 \mathcal J_{r,m} (\xi,t)
 :=
 \begin{cases}
  \mathcal B_{(r+m-1)} (\xi,t),
     &r+m-1<p,\\[2pt]
  p\,\partial_\xi
     \mathcal B_{(r+m-1-p)} (\xi,t),
     &r+m-1\geq p.
 \end{cases}
 \end{split}
\end{equation}
Just as in the standard resummation theory, the obvious singularities of the equation \eqref{zBoreleq} tell us that the possible singular locus of the Borel transforms is contained in the union of the
curves
\begin{align}
\Sigma_{abcd}=\{(\xi,t)~|~ \xi=u_a-u_b+t^{-p}(u_c-u_d)\}, \qquad a\ne b, \text{ or } \ c\ne d.
\end{align}
Among them, there are the following distinguished curves (that will be shown to be the genuine singular curves)
\begin{align}\label{curve1}
\Sigma_{1,ab}&=\{(\xi,t)~|~ \xi=u_a-u_b\},\qquad a\ne b,\\ \label{curve2}
\Sigma_{2,ab}&=\{(\xi,t)~|~ \xi=t^{-p}(u_a-u_b)\},\qquad a\ne b,\\ \label{curve3}
\Sigma^{(1)}_{3,ab}&=\left\{(\xi,t)~|~ \xi=(1-t^{-p})(u_a-u_b), \ t\ne 0,1, \ -\frac{\pi}{p}<{\rm Arg}(t)<\frac{\pi}{p} \right\},\qquad a\ne b.
\end{align}

Now fix a curve $\Sigma_{ijkl}$, let us denote
\begin{equation*}\Sigma_{ijkl}^{0}:=\Sigma_{ijkl}\setminus
 \overline{
 \bigcup_{a\ne b}
 \left(
 \Sigma_{1,ab}\cup\Sigma_{2,ab}\cup\Sigma_{3,ab}^{(1)}
 \right)}.
\end{equation*}
Here the closure is taken in
$\mathbb C_\xi\times(\mathbb C_t\setminus\{0,1\})$.
Let us prove by induction on the $\frak m$-adic order that every entry of
$\mathcal{B}_{(r)}(\xi,t)$,  for all $r=0,\ldots,p-1$, extends
holomorphically through $\Sigma^0_{ijkl}$.

Modulo
$\frak m$, we have $\mathcal{B}_{(0)}^{[1]}=1$ and $\mathcal{B}_{(r)}^{[1]}=0$ for $r=1,\ldots,p-1$. Assume
that the continuation statement has
been proved modulo $\frak m^{N-1}$. Let us prove it modulo $\frak m^N$. First note that
every summand on the
right hand sides of \eqref{zBoreleq} and \eqref{tBoreleq} is the product of $B_{(r)}$, or its $\xi$ derivative, with one of the elements
\[D_l,\hbar T_{[l]}, \hbar T,\hbar P,\hbar \delta P\]
in $\frak m$. Therefore, modulo $\frak m^N$, they 
depend only on the $\mathcal B^{[N-1]}_{(r)}$ modulo $\frak m^{N-1}$. The required correction in $\mathfrak m^{N-1}/\mathfrak m^N$ satisfies an inhomogeneous linear
system whose right hand side is already known and holomorphic away from the three families \eqref{curve1}-\eqref{curve3} by the induction hypothesis. 

Case I. For the $(abcd)$-entry of $B_{(r)}^{[N]}(\xi,t)$ indexed by a tuple $(abcd)\ne (ijkl)$. Take any point \[P_*=(u_i-u_j+t_*^{-p}(u_k-u_l),t_*)\in\Sigma_{ijkl}^{0}\setminus\Sigma_{abcd}.\]
The coefficient
of $\partial/\partial \xi$ in the first equation for $\mathcal{B}_{(r)}^{[N]}(\xi,t)_{abcd}$ (i.e., replacing indices $(ijkl)$ by $(abcd)$ in \eqref{zBoreleq}) is nonzero near $P_*$. Since the right hand side of the equation only depends on $\mathcal{B}_{(a)}^{[N-1]}$ and their $\xi$ derivatives, $a=1,\ldots, p$, and is therefore holomorphic at $P_*$ by the induction hypothesis, the solution $B_{(r)}^{[N]}(\xi,t)_{abcd}$ extends holomorphically across
$\Sigma^0_{ijkl}$ near $P_*$. 

Case II. For the $(ijkl)$-entry of $B_{(r)}^{[N]}(\xi,t)$ indexed by the same tuple $(ijkl)$. Introduce the new coordinate
\[
 c=\xi-(u_i-u_j+t^{-p}(u_k-u_l)), \qquad t=t.
\]
Since 
\[\frac{\partial}{\partial t}-p(u_k-u_l)t^{-p-1} \frac{\partial}{\partial \xi}
 =\left. \frac{\partial}{\partial t} \right|_c,\]
the equation \eqref{tBoreleq} becomes
\begin{equation}\label{ctequation}
  \frac{\partial}{\partial t} (\mathcal  B^{[N]}_{(r)}(c,t) )_{ijkl}=\mathcal{F}^{[N]}_{(r)}(c,t)_{ijkl},
\end{equation}
where $\mathcal{F}^{[N]}_{(r)}(c,t)_{ijkl}$ 
depends only on the $\mathcal B^{[N-1]}_{(a)}$ modulo $\frak m^{N-1}$. 
Thus the equation \eqref{ctequation} is a
differential equation transporting the unknown $\mathfrak m^{N-1}/\mathfrak m^N$ component along the curves $c=$constant.

It is a key observation that there exists $t_0\neq 0,1,\infty$ such that $(c=0,t_0)\in \Sigma_{ijkl}^0$, i.e., 
\[u_i-u_j+t_0^{-p}(u_k-u_l)=0.\]
The convergent germ of the Borel transform $\mathcal B_{(r)}^{[N]}$ along the zero section gives an initial value at $(c,t)=(0,t_0)$.
Then for a fixed path $\gamma$ on the $t$-plane such that $(c=0,\gamma)\in \Sigma_{ijkl}^0$,
integrating the equation \eqref{ctequation} for sufficiently small $|c|$ gives 
\begin{equation}\label{TransportB}
 (\mathcal B_{(r)}^{[N]})_{ijkl}(c,t)
 =
 (\mathcal B_{(r)}^{[N]})_{ijkl}(c,t_0)
 +
 \int_{t_0}^{t}F_{{(r)}}^{[N]}(c,\tau)_{ijkl}\,d\tau ,
\end{equation}
where, by the induction hypothesis, the 
$\mathcal F^{[N]}(c,\tau)_{ijkl}$ is holomorphic for $|c|$ small enough and $\tau$ in a sufficiently small neighbourhood of the path $\gamma$ on $t$-plane. Thus \eqref{TransportB} gives a holomorphic extension of $(\mathcal B_r^{[N]})_{ijkl}$ across $\Sigma^0_{ijkl}$.

It proves that the Borel transforms $\mathcal B_{(r)}^{[N]}$ extend holomorphically across all apparent singular curves other than the distinguished families \eqref{curve1}-\eqref{curve3}. 

Now let $d\notin{\rm aSR}(u)$ be the given resummation direction. The above results say that for all $t \notin \mathcal{C}(u;d)_0\cup\mathcal{C}(u;d)_1$, the Borel transforms extend holomorphically from $\xi=0$ to $\xi=\infty$ along $d$. The exponential growth estimate (locally uniform with respect to $t$), needed for the Laplace transform, follows from the standard argument. That is if $K$ is a compact subset of the
chamber $\mathbb{D}$, then there exist positive real $C_{K,N}$ and $\beta_{K,N}$ such that
\[|\mathcal B_{(r)}^{[N]}(\xi,t)|\le C_{K,N} e^{\beta_{K,N}|\xi|}, \qquad \text{for } t\in K, \ \ \xi\in e^{\mathi pd}\cdot \mathbb{R}_{>0}.\]
The standard
beta-integral recombination over $r=0,1,\ldots,p-1$ gives the corresponding statement for the order-$p$ Borel transform
of $\widehat{Q}$. In particular, the formal series $\widehat Q(z,t)$ is
$p$-summable in the direction $d$, and its $p$-sum is
\[
Q_d^{[N]}(z,t)
=
\sum_{r=0}^{p-1}
z^r
\frac1{p z^p}
\int_0^{e^{ipd}\cdot \infty}
e^{-\xi/(p z^p)}
\mathcal B^{[N]}_{(r)}(\xi,t)\,d\xi.
\]
  
Therefore, the resummation $Y_{d,t}^{[N]}(t)$ via the standard order-$p$ Borel-Laplace transform exists and depends holomorphically on $t$ in every component $\mathbb{D}$. One defines $Y_{d,t}(t)$ by passing to the inverse limit, and uniqueness follows from the prescribed asymptotics. \end{proof}

\begin{ex}
To illustrate, the following is the graph of the curves and regions in the $t$-plane for the case $p=4$, $u_k\in i\mathbb{R}$ (purely imaginary) for all $k=1,\dots,n$, and $d=0$.
\begin{figure}[h] 
\centering
\begin{tikzpicture}[scale=2.0]

\clip (-2.2,-2.2) rectangle (2.2,2.2);


\draw[gray, dashed, thin] (0,0) circle (1);

\foreach \angle in {22.5, 67.5, 112.5, 157.5, 202.5, 247.5, 292.5, 337.5} {
    \draw[blue!65!black, line width=0.9pt, opacity=0.85] (0,0) -- (\angle:2.2);
}

\draw[red!70!black, line width=1.0pt, opacity=0.85]
        (0.0100,0.0000) --
        (0.1954,0.0808) --
        (0.2093,0.0865) --
        (0.2232,0.0922) --
        (0.2371,0.0979) --
        (0.2510,0.1036) --
        (0.2650,0.1092) --
        (0.2789,0.1149) --
        (0.2929,0.1205) --
        (0.3069,0.1260) --
        (0.3208,0.1315) --
        (0.3348,0.1370) --
        (0.3489,0.1425) --
        (0.3629,0.1478) --
        (0.3770,0.1531) --
        (0.3911,0.1584) --
        (0.4053,0.1635) --
        (0.4195,0.1686) --
        (0.4337,0.1736) --
        (0.4480,0.1785) --
        (0.4623,0.1832) --
        (0.4766,0.1879) --
        (0.4910,0.1923) --
        (0.5055,0.1967) --
        (0.5200,0.2008) --
        (0.5345,0.2047) --
        (0.5491,0.2085) --
        (0.5637,0.2120) --
        (0.5785,0.2153) --
        (0.5933,0.2183) --
        (0.6081,0.2210) --
        (0.6230,0.2235) --
        (0.6378,0.2256) --
        (0.6528,0.2273) --
        (0.6677,0.2287) --
        (0.6828,0.2297) --
        (0.6978,0.2303) --
        (0.7129,0.2304) --
        (0.7280,0.2300) --
        (0.7431,0.2291) --
        (0.7581,0.2277) --
        (0.7730,0.2258) --
        (0.7878,0.2232) --
        (0.8025,0.2201) --
        (0.8172,0.2163) --
        (0.8315,0.2118) --
        (0.8457,0.2067) --
        (0.8596,0.2009) --
        (0.8732,0.1943) --
        (0.8863,0.1871) --
        (0.9001,0.1784) --
        (0.9124,0.1696) --
        (0.9241,0.1601) --
        (0.9352,0.1499) --
        (0.9456,0.1390) --
        (0.9552,0.1275) --
        (0.9641,0.1153) --
        (0.9721,0.1026) --
        (0.9792,0.0893) --
        (0.9853,0.0755) --
        (0.9904,0.0614) --
        (0.9944,0.0469) --
        (0.9974,0.0322) --
        (0.9993,0.0173) --
        (1.0000,0.0023) --
        (1.0000,-0.0000) --
        (0.9994,-0.0150) --
        (0.9977,-0.0299) --
        (0.9949,-0.0447) --
        (0.9911,-0.0593) --
        (0.9861,-0.0735) --
        (0.9801,-0.0873) --
        (0.9732,-0.1007) --
        (0.9653,-0.1135) --
        (0.9565,-0.1258) --
        (0.9470,-0.1374) --
        (0.9367,-0.1484) --
        (0.9257,-0.1587) --
        (0.9141,-0.1683) --
        (0.9019,-0.1772) --
        (0.8892,-0.1854) --
        (0.8761,-0.1928) --
        (0.8627,-0.1995) --
        (0.8489,-0.2054) --
        (0.8347,-0.2108) --
        (0.8204,-0.2153) --
        (0.8058,-0.2193) --
        (0.7912,-0.2226) --
        (0.7764,-0.2253) --
        (0.7615,-0.2273) --
        (0.7465,-0.2289) --
        (0.7314,-0.2299) --
        (0.7164,-0.2303) --
        (0.7012,-0.2303) --
        (0.6861,-0.2299) --
        (0.6712,-0.2290) --
        (0.6561,-0.2277) --
        (0.6412,-0.2260) --
        (0.6262,-0.2240) --
        (0.6114,-0.2216) --
        (0.5965,-0.2189) --
        (0.5817,-0.2160) --
        (0.5670,-0.2128) --
        (0.5523,-0.2093) --
        (0.5377,-0.2056) --
        (0.5232,-0.2017) --
        (0.5087,-0.1976) --
        (0.4942,-0.1933) --
        (0.4798,-0.1889) --
        (0.4655,-0.1843) --
        (0.4512,-0.1796) --
        (0.4369,-0.1747) --
        (0.4227,-0.1698) --
        (0.4085,-0.1647) --
        (0.3944,-0.1596) --
        (0.3803,-0.1544) --
        (0.3662,-0.1491) --
        (0.3522,-0.1437) --
        (0.3381,-0.1383) --
        (0.3241,-0.1328) --
        (0.3101,-0.1273) --
        (0.2962,-0.1218) --
        (0.2822,-0.1162) --
        (0.2683,-0.1106) --
        (0.2543,-0.1049) --
        (0.2404,-0.0993) --
        (0.2265,-0.0936) --
        (0.2126,-0.0879) --
        (0.1987,-0.0822) --
        (0.0100,0.0000);

\node[font=\boldmath\large] at (45:1.55)  {$\D_1$};
\node[font=\boldmath\large] at (90:1.55)  {$\D_2$};
\node[font=\boldmath\large] at (135:1.55) {$\D_3$};
\node[font=\boldmath\large] at (180:1.55) {$\D_4$};
\node[font=\boldmath\large] at (225:1.55) {$\D_5$};
\node[font=\boldmath\large] at (270:1.55) {$\D_6$};
\node[font=\boldmath\large] at (315:1.55) {$\D_7$};

\node[font=\boldmath, text=red!70!black] at (0.50, 0) {$\D_{(0,1)}$};
\node[font=\boldmath, text=red!70!black] at (1.65, 0) {$\D_{(1,+\infty)}$};

\fill[black] (0,0) circle (0.02);
\node[below left, gray, font=\small] at (-0.04,-0.04) {$0$};

\fill[black] (1,0) circle (0.025);
\node[below right, font=\small] at (1.02,-0.04) {$t{=}1$};

\draw[black, thin] (-2.3,-2.3) rectangle (2.3,2.3);

\foreach \x in {-2,-1,0,1,2} {
    \draw[black, thin] (\x,-2.3) -- (\x,-2.24);
    \draw[black, thin] (\x,2.3) -- (\x,2.24);
    \node[below, font=\footnotesize] at (\x,-2.3) {$\x$};
}
\foreach \y in {-2,-1,0,1,2} {
    \draw[black, thin] (-2.3,\y) -- (-2.24,\y);
    \draw[black, thin] (2.3,\y) -- (2.24,\y);
    \node[left, font=\footnotesize] at (-2.3,\y) {$\y$};
}

\node[below, font=\normalsize] at (0, -2.58) {$\mathrm{Re}(t)$};
\node[left, font=\normalsize, rotate=90] at (-2.58, 0) {$\mathrm{Im}(t)$};

\node[above, font=\large] at (0, 2.45) {$\mathrm{aSR}(u;\,d)$ \text{with} $p=4,\; d=0$};

\node[below, font=\normalsize, gray!70!black] at (0, -2.82) {$h=1,\; u_1=\tfrac{i}{2},\; u_2=-\tfrac{i}{2}$};

\begin{scope}[shift={(1.05, 1.85)}]
    \fill[white, opacity=0.92, rounded corners=2pt] (-0.05,-0.6) rectangle (1.6,0.15);
    \draw[gray!50, thin, rounded corners=2pt] (-0.05,-0.6) rectangle (1.6,0.15);
    \draw[blue!65!black, line width=0.9pt] (0.05,0) -- (0.3,0);
    \node[right, font=\small] at (0.32,0) {$\mathcal{C}(u;d)_0$};
    \draw[red!70!black, line width=1.0pt] (0.05,-0.38) -- (0.3,-0.38);
    \node[right, font=\small] at (0.32,-0.38) {$\mathcal{C}(u;d)_1$};
\end{scope}

\end{tikzpicture}
\end{figure}
\end{ex}
\subsection{Factorization of $Y_{d,\mathbb{D}}(z;t)$ at $t=0,1$ and $\infty$}

Proposition \ref{Yztresum} states that the resulting holomorphic solution, obtained by the resummation of $\widehat{Y}(z;t)$ along $d$, depends on the connected component (cut out by the two sets $\mathcal{C}(u;d)_0$ and $\mathcal{C}(u;d)_1$) to which the parameter $t$ belongs. To derive the connection formula of $Y_{d,\D}$ on different components $\D$, we first derive the factorization properties of the holomorphic solutions. 

First, recall from the proof of Proposition \ref{Yztresum} that for fixed $t$, the singular directions of the formal solution $\widehat{Y}(z;t)$, along which the formal power series part is not summable with respect to $z$, are contained in the
following three families
\begin{align}
    {\rm aSR}(u)&=\left\{\frac{1}{p}\mathrm{Arg}(u_i-u_j)+\frac{2m\pi}{p}\right\},\\
    {\rm aSR}(t^{-p}u)&=\left\{\frac{1}{p}\mathrm{Arg}(t^{-p}(u_k-u_l))+\frac{2m\pi}{p}\right\},\\
    {\rm aSR}((1-t^{-p})u)&=\left\{\frac{1}{p}\mathrm{Arg}((1-t^{-p})(u_k-u_l))+\frac{2m\pi}{p}\right\}.
\end{align}

Secondly, factorizations at $t=\infty,0,1$ in Propositions \ref{YWFfactorization}, \ref{t2fact} and \ref{t3fact} separate out the corresponding poles and place them into the second factor: \begin{itemize}
    \item in terms of $z,t$ coordinates, we have
\begin{align}\label{Yde1}
\widehat{Y}(z,t)= {_1\widehat{W}}(z_1(z,t);z_2(z,t))\widehat{F}^{(2)}(zt)\cdot \left(\frac{t-1}{1-t}\right)^{\hbar\delta P},
\end{align}
where the formal power series part of ${_1\widehat{W}}(z;t)={_1\widehat{W}}(z_1(z,t);z_2(z,t))$ is a series in $z$, whose coefficients are holomorphic at $t=\infty$; 
\item in terms of coordinates $z_2=tz$ and $t$, we have
\begin{equation}\label{Yde3}
\widehat{Y}(z,t)={_2\widehat{W}}(z_2;z_1(z_2,t))\widehat{F}^{(1)}(z),
\end{equation}
where the formal power series part of ${_2\widehat{W}}(z_2;t)={_2\widehat{W}}(z_2;z_1(z_2,t))$ is a series in $z_2$, whose coefficients depend on $t$ and are holomorphic at $t=0$. 

\item in terms of $z,t$ coordinates, we have
\begin{align}\label{Yde2}
\widehat{Y}(z,t)=\widehat{U}(z,t)\cdot\widehat{X}\left((t^{-p}-1)^{-1/p}z\right)\cdot p^{\hbar\delta P}.
\end{align}
\end{itemize}

At $t=\infty,0,1$, respectively, the corresponding first factors ${_1\widehat{W}}(z;\infty)$, ${_2\widehat{W}}(z_2;0)$ and ${\widehat{U}}(1;z)$ satisfy limiting linear equations in $z$. Every direction $d\notin {\rm aSR}(u)$ used below is admissible for the limiting equations in $z$ and remains
admissible in a sufficiently small parameter neighbourhood of $t=\infty,0,1$. 
Then

\begin{pro}\label{factresum}
Given $d\notin{\rm aSR}(u)$, there exist small neighbourhoods $B_\infty$, $B_0$ and $B_1$ of $t=\infty, 0, 1$ respectively, such that
\begin{enumerate}
\item[(a).]
For any $t\in B_\infty$, the resummation $_1{W}_d(z;t)$ of $_1\widehat{W}(z;t)$ with respect to $z$ along direction $d$ is well defined;

\item[(b).]
(set $z_2=zt$) For any $t\in B_0$, the resummation of $_2\widehat{W}(z_2;t)$ with respect to $z_2$ along the direction $d$ is well defined;

\item[(c).]
For any $t\in B_1$,  the resummation ${U}_d(z;t)$ of $\widehat{U}(z;t)$ with respect to $z$ along the direction $d$ is well defined.
\end{enumerate}

\end{pro}

\begin{proof}
Proof of (a). It follows from the equations \eqref{kKZ1}-\eqref{kKZ2} satisfied by ${_1\widehat{W}}(z_1;z_2)\cdot \hat{F}^{(2)}(z_2)$ that ${_1\widehat{W}}(z_1;z_2)$ satisfies a linear system with respect to $z_1$ and $z_2$.
Under the change of coordinates $z=z_1$, $t=z_2/z_1$, the system becomes
\begin{align}\label{Wzeq}
\frac{{\partial {_1\widehat{W}}(z;t)}}{\partial z}= \left(\frac{u^{(1)}}{z^{p+1}}+\sum_{k=2}^p\frac{\hbar T_{[k]}^{(1)}}{z^k}-\frac{\hbar T^{(1)}}{z}+\frac{\hbar P}{z}\right)\cdot {_1\widehat{W}}+\left[\frac{u^{(2)}}{t^pz^{p+1}}+\sum_{k=2}^p\frac{\hbar T_{[k]}^{(2)}}{t^{k-1}z^k}-\frac{\hbar T^{(2)}}{z}, {_1\widehat{W}}\right].
\end{align}
In a finite quotient $\Uph/\frak m^N$, we write the regular part of $_1\widehat{W}^{[N]}(z,t)$ as
\[
1+\sum_{m\ge 1}K_m^{[N]}(zt)z^m=
1+\sum_{m\ge1}
\left(\sum_{j\ge 0}K_{m+j,j}^{[N]}t^{-j}\right)z^m=1+\sum_{m\ge 1} X_m(t)z^m.
\]
Comparing the coefficient of $z^m$ in the regular part of the both sides of \eqref{Wztform}, we see that $X^{[N]}_m(t)=\sum_{j\ge 0}K_{m+j,j}^{[N]}t^{-j}$ is holomorphic at $t=\infty$. By Proposition \ref{Yztresum}, the possible singular directions (genuine anti-Stokes rays) for $_1\widehat{W}^{[N]}(z,t)$ are 
\begin{equation}\label{1Wrays}
    \frac{1}{p}{\rm arg}(u_a-u_b), \ \frac{1}{p}{\rm arg}(t^{-p}(u_c-u_d)) \ \text{ and } \ \frac{1}{p}{\rm arg}((1-t^{-p})(u_a-u_b)).
\end{equation}

In the underlying recursion relation of equation \eqref{Wzeq}, the $(ijkl)$ entry of the leading term is multiplied by the scalar $(u_i-u_j)+t^{-p}(u_k-u_l)$. Then on a small circle $|1/t|=\varepsilon$, all non-identically zero scalars have a fixed positive distance from $0$. Using this, one checks that the coefficient recursion relation gives uniform estimates on the circle
\[
 \sup_{|1/t|=\var}\|X_m^{[N]}(t)\|
 \le C_N \rho^m\Gamma\left(1+\frac mp\right),
\]
for constants $\rho$ and $C_N$. Since the coefficients $X_m^{[N]}(t)$ are holomorphic, the maximum principle extends these estimates to the disc $D_\var(\infty):=\{t~;~|1/t|\le \varepsilon\}$. Thus the order-$p$ Borel transform of the regular part of $_1\widehat{W}^{[N]}(z,t)$ for all $t\in D_\var(\infty)$ has a common initial radius of convergence $r$ at the origin, with $r>0$ independent of $t\in D_\var(\infty)$.
After shrinking the disc if necessary, we can assume that all the points $t^{-p}(u_c-u_d)$ have norm less than $r$, thus are analytic points (on the Borel plane) of the order-$p$ Borel transform. Thus the rays of the second kind in \eqref{1Wrays} (if they do not coincide with the rays from first and third kinds) are admissible. On the other hand, the rays of first kind belong to ${\rm aSR}(u)$, and as $t$ sufficiently close to $\infty$ (i.e., $\var$ small enough), the rays $\frac{1}{p}{\rm arg}(u_a-u_b)$ and $\frac{1}{p}{\rm arg}((t^{-p}-1)(u_a-u_b))$, from the first and third kinds respectively, get sufficiently close to each other. Therefore, for any given fixed $d\notin{\rm aSR}(u)$, there exists $\var>0$ small enough (that depend on the choice of $d$) such that (the regular part of) $_1\widehat{W}^{[N]}(z,t)$ is $p$-summable along the direction $d$ for all $t\in B_\infty:=D_\var(\infty)$.

Here we add the following remark. Note that the evaluation of the formal $_1\widehat{W}^{[N]}(z,t)$ at $t=\infty$ 
\[{_1\widehat{W}}(z;\infty)=\left(1+\sum_{m\ge1}K_{m,0}^{[N]}z^m\right)
e^{\int D^{(1)}(z)\,dz}
z^{D_1^{(1)}+\hbar\delta P}\] satisfies the equation
\begin{align}\label{Winf}
\frac{{\partial {_1\widehat{W}}(z;\infty)}}{\partial z}= \left(\frac{u^{(1)}}{z^{p+1}}+\sum_{k=2}^p\frac{\hbar T_{[k]}^{(1)}}{z^k}-\frac{\hbar T^{(1)}}{z}+\frac{\hbar P}{z}\right)\cdot {_1\widehat{W}(z;\infty)}+\left[-\frac{\hbar T^{(2)}}{z}, {_1\widehat{W}}(z;\infty)\right].
\end{align}
The order-$p$ Borel–Laplace transform of the formal power series part of ${_1\widehat{W}}(z;t)$, with respect to $z$ along the direction $d\notin {\rm aSR}(u)$, can be taken locally uniformly in the neighbourhood $D_\var(\infty)$ of $t=\infty$. It can be seen from the fact that the order-$p$ Borel transform extends along $d$ with growth bounds locally uniform in $t^{-1}$, and the Laplace integrals converge locally uniformly in $t^{-1}$ as well. 
Therefore, the order-$p$ Borel-Laplace transform commutes with the evaluation at $t=\infty$. That is the resummation ${_1{W}}_{d}(z;t)$ is holomorphic around $t=\infty$ and its value ${_1{W}}_{d}(z;\infty)$ at $t=\infty$ is the resummation of ${_1\widehat{W}}(z;\infty)$ along the same direction $d$. 

Thus one can also understand ${_1{W}}_{d}(z;t)$ from the $t$-equation. First ${}_1\widehat{W}$ satisfies the $t$ equation 
\begin{align}\label{Wteq}
\frac{\partial{}_1\widehat{W}}{\partial t}
=
\left[
\frac{u^{(2)}}{t^{p+1}z^p}
+\sum_{k=2}^{p}
\frac{\hbar T_{[k]}^{(2)}}{t^kz^{k-1}}
-\frac{\hbar T^{(2)}}t,
{}_1\widehat{W}
\right]+\frac{\hbar P}{t(1-t)}{}_1\widehat{W}.
\end{align}
It has a regular singularity at $t=\infty$. Note that ${_1\widehat{W}}(z;\infty)$, therefore ${_1{W}}_d(z;\infty)$, commutes with $\hbar T^{(2)}$, i.e., $[\hbar T^{(2)},{}_1W_d(z;\infty)]=0$. Starting with the initial value $W_0(z):={_1{W}}_d(z;\infty)$ that commutes with the residue $\hbar T^{(2)}$ at the regular singularity, the equation \eqref{Wteq} has a unique power series solution in $t^{-1}$ with $z$ regarded as a parameter
\[ W_0(z)+\sum_{k\ge 1} W_k(z) t^{-k}.\]
It is convergent at $t=\infty$ and coincides with ${}_1W_d(z;t)$.

Part (b) is proved in the same way. 

Proof of (c). By Proposition \ref{t3fact} we have
\begin{equation*}
    \widehat U(z,t)
 =\left(I_0(z)
       +\sum_{k\ge 1}{I}_k(z)(t-1)^k\right)E_U(z)=\widehat{U}_0(z)+\sum_{k\ge 1} \widehat{U}_k(z)(t-1)^k.
\end{equation*}
We can rewrite the part $\sum_{k\ge 0}{I}_k(z)(t-1)^k$ as a series in $z$. It follows from
the bound \eqref{Ikbound} that it is a formal Laurent series with
the degree of $z$ at least $-(p-1)(N -1)$. After separating the finitely many negative powers, one can then take its order-$p$ Borel-Laplace transform.
Then by Proposition \ref{Yztresum}, the possible singular directions (genuine anti-Stokes rays) for $\widehat{U}^{[N]}$ are those in \eqref{1Wrays}. Similar to the proof of part $(a)$, as $t$ in a sufficiently small neighbourhood $B_1$ of $t=1$, the rays of the third kind are admissible, the rays of first kind belong to ${\rm aSR}(u)$,  while the rays of second kind get close to the first kind in ${\rm aSR}(u)$. Therefore, for the given fixed $d\notin {\rm aSR}(u)$, the resummation $U_d(z,t)$ is defined for all $t$ in $B_1$.

Here we also remark that ${U}_{d}(z,t)$ satisfies the $t$-equation \eqref{Uteq} with a regular singularity at $t=1$, and its evaluation at $t=1$ coincides with the resummation $U_{d}(z,1)$ of the formal $\widehat{U}(z,1)=\widehat{I}_0 E_U$ along the same ${d}$. 
Actually, we can regard $z$ as parameter and take $U_d(z)_0:=U_d(z,1)$ as an initial condition of the $t$ equation \eqref{Uteq} with the regular singularity term $\frac{{\rm ad}_{\hbar P}}{1-t}$. Just like in the formal setting, since $[\hbar P, U_d(z,1)]=0$, the singular term ${\rm ad}_{\hbar P}/(t-1)$ annihilates the initial value. The recursion for the equation \eqref{Uteq} therefore gives a unique holomorphic solution with the prescribed form $U_d(z)_0+\sum_{k\ge 1} U_d(z)_k\cdot (t-1)^k$ around the simple pole $t=1$. This holomorphic function around $t=1$ simply coincides with $U_d(z,t)$.
\end{proof}

Therefore, in each case the second factor captures the local singular behavior, and the connection matrices of $Y_{d,\D}$ are determined by (the Stokes phenomenon of) the second factors. Taking the resummation of the both sides of identities \eqref{Yde1}, \eqref{Yde3} and \eqref{Yde2} with respect to the variable $z$ along an admissible direction $d$, the involved estimates are locally uniform in $t$, the factorizations of formal solutions yield (here again the resummation is taken in every finite quotient $\Uph/\frak m^N$ in standard way and then pass to the inverse limit)
\begin{cor}\label{mainpro}
For each $k=1,...,2p-1$ let $\D_k$ be the connected component of 
\[({\mathbb{C}\setminus\{0,1\})}\setminus (\mathcal{C}(u;d)_0\cup\mathcal{C}(u;d)_1)\] that contains the ray ${\rm arg}(t) = \frac{k\pi}{p}$, and let $\D_{(0,1)}$, $\D_{(1,+\infty)}$ be the connected components that contain the intervals $(0,1)$ and $(1,+\infty)$, respectively. On the interval $(1,+\infty)$ we choose ${\rm arg}(t-1)=0$ and ${\rm arg}(1-t)=-\pi$, then we have the factorization identities (recall that $z_1=z, z_2=zt$ and $\omega=(t^{-p}-1)^{-1/p}z$)
\begin{align}\label{YWF1}
_2W_{d+\frac{k\pi}{p}}(z_2;z_1)F^{(1)}_{d}(z_1)&=Y_{d,\D_k}(z,t)={_1W}_{d}(z_1;z_2)F^{(2)}_{d+\frac{k\pi}{p}}(z_2)\cdot e^{\pi\mathi \hbar\delta P}, \ \ k=1,...,2p-1,\\ \label{YUX}
{_2W}_{d}(z_2;z_1)F^{(1)}_{d}(z_1)&=Y_{d,\D_{(0,1)}}(z,t)={U}_d(z;t){X}_{d}(\omega)\cdot p^{\hbar\delta P},\\ \label{YWF2}
{U}_d(z;t){X}_{d+\frac{\pi}{p}}(\omega)\cdot p^{\hbar\delta P}&=Y_{d,\D_{(1,+\infty)}}(z,t)={_1W}_{d}(z_1;z_2)F^{(2)}_{d}(z_2)\cdot e^{\pi\mathi \hbar\delta P}.
\end{align}
Here 
\begin{itemize}
    \item $F^{(2)}_{d_2}(z_2=zt)$ is the function obtained by the resummation of $\widehat{F}_{d_2}(z_2=zt)$ along the admissible direction ${\rm arg}(z)=d_2-{\rm arg}(t)$, or equivalently  ${\rm arg}(z_2)=d_2\notin {\rm aSR}(u)$; for $t\in B_\infty$ a small neighbourhood of $t=\infty$, the function $_1W_{d_1}(z_1;z_2)$ is the resummation of $_1\widehat{W}(z;t)$ along an admissible direction ${\rm arg}(z)=d_1\notin {\rm aSR}(u)$. 
    
    \item $F^{(1)}_{d_1}(z_1)$ is the resummation of $\widehat{F}^{(1)}(z_1)$ along the admissible direction $d_1$; and $t\in B_0$ a small neighbourhood of $t=0$, the function $_2W_{d_2}(z_2;z_1)$ is the resummation of $_2\widehat{W}(z_2=zt;z_1=z)$ along the direction ${\rm arg}(z)=d_2-{\rm arg}(t)$, or equivalently  ${\rm arg}(z_2)=d_2\notin {\rm aSR}(u)$. 
    \item $U_{d}(z;t)$ denotes the function obtained by the  resummation of $\widehat{U}(z;t)$ along the admissible direction $d$, $t\in B_1$ a small neighbourhood of $t=1$,; and $X_{d}(\omega=(t^{-p}-1)^{-1/p}z)$ and $X_{d+\frac{\pi}{p}}(\omega)$ denote resummations of $\widehat{X}(\omega)$ along $d$ and $d+\frac{\pi}{p}$ respectively (that is for $z$ along $d$, while $t\in \D_{(0,1)}$ and $t\in \D_{(1,+\infty)}$ respectively).
\end{itemize}
\end{cor}

\subsection{The proof of Theorem \ref{thmRLL}}\label{proofThm}

Without loss of generality, we can make 
\begin{assumption*}
The initial anti-Stokes ray $\tau_0$ and the $u\in\h_{\rm reg}$ are such that the Stokes matrices $S_{2k+1}(u)$ and $S_{2k}(u)$ are lower and upper triangular, respectively.
\end{assumption*}
This assumption gives the identities in Theorem \ref{thmRLL} in the stated
form. Otherwise, the same identities would involve the
corresponding permutation of the Stokes matrices. Then as a corollary of the factorization identities \eqref{YWF1}-\eqref{YWF2}, 

\begin{cor}[Theorem \ref{thmRLL}]
We have the commutation relations
\begin{align}\label{SS=SS}
e^{\pi\mathi \hbar\delta P}S^{(1)}_i(u)e^{-\pi\mathi \hbar\delta P}S_{i+k}^{(2)}(u)&=S^{(2)}_{i+k}(u)e^{\pi\mathi \hbar\delta P}S^{(1)}_i(u)e^{-\pi\mathi \hbar\delta P}, \ \ \text{for any $i$ and } 1<k<2p-1,\\
\label{pRLL1} R^{12}S^{(1)}_i(u)e^{-\pi\mathi \hbar\delta P}S^{(2)}_i(u) &=S^{(2)}_i(u)e^{-\pi\mathi \hbar\delta P}S^{(1)}_i(u) R^{12}, \ \ \text{for $i$ odd} \\ \label{pRLL12}
R^{12}S^{(2)}_i(u)e^{-\pi\mathi \hbar\delta P}S^{(1)}_i(u) &=S^{(1)}_i(u)e^{-\pi\mathi \hbar\delta P}S^{(2)}_i(u)R^{12}, \ \ \text{for $i$ even} \\ \label{pRLL3}
S^{(2)}_{i-1}(u)e^{-\pi\mathi \hbar\delta P} S^{(1)}_{i}(u)e^{\pi\mathi \hbar\delta P} &=e^{-\pi\mathi \hbar\delta P}S^{(1)}_{i}(u) R^{12} S^{(2)}_{i-1}(u), \ \ \text{for $i$ odd},\\ \label{pRLL4}
S^{(1)}_{i-1}(u)e^{-\pi\mathi \hbar\delta P}S^{(2)}_{i}(u)e^{\pi\mathi \hbar\delta P} &=e^{-\pi\mathi \hbar\delta P}S^{(2)}_{i}(u) R^{12} S^{(1)}_{i-1}(u), \ \ \text{for $i$ even}.
\end{align}
\end{cor}
\begin{proof}
To prove \eqref{SS=SS},
by Corollary \ref{mainpro}, for any chosen admissible direction $d$ such that $\tau_{il}<d<\tau_{il+1}$, let us consider the following solutions (the terms $e^{\pm \pi\mathi\hbar\delta P}$ account for the argument of $(t-1)$ in the multivalued function $\left(\frac{t-1}{1-t}\right)^{\hbar\delta P}$)
\begin{align}\label{WYW1}
_2W_{d+\frac{k\pi}{p}}(z_2;z_1)F^{(1)}_{d}(z_1)&=Y_{d,\D_k}(z,t)={_1W}_{d}(z_1;z_2)F^{(2)}_{d+\frac{k\pi}{p}}(z_2)\cdot e^{\pi\mathi \hbar\delta P},\\ \label{WYW2}
_2W_{d+\frac{(k+1)\pi}{p}}(z_2;z_1)F^{(1)}_{d}(z_1)&=Y_{d,\D_{k+1}}(z,t)={_1W}_{d}(z_1;z_2)F^{(2)}_{d+\frac{(k+1)\pi}{p}}(z_2)\cdot e^{\pi\mathi \hbar\delta P},\\
_2W_{d+\frac{(k+1)\pi}{p}}(z_2;z_1)F^{(1)}_{d+\frac{\pi}{p}}(z_1)&=Y_{d+\frac{\pi}{p},\D_{k+1}}(z,t)={_1W}_{d+\frac{\pi}{p}}(z_1;z_2)F^{(2)}_{d+\frac{(k+1)\pi}{p}}(z_2)\cdot e^{\pi\mathi \hbar\delta P},\\
_2W_{d+\frac{k\pi}{p}}(z_2;z_1)F^{(1)}_{d+\frac{\pi}{p}}(z_1)&=Y_{d+\frac{\pi}{p},\D_k}(z,t)={_1W}_{d+\frac{\pi}{p}}(z_1;z_2)F^{(2)}_{d+\frac{k\pi}{p}}(z_2)\cdot e^{\pi\mathi \hbar\delta P}.
\end{align}

It then follows from the second identity in \eqref{WYW1} and the second identity in \eqref{WYW2} that
\begin{equation}
Y_{d,\D_{k+1}}(z,t)^{-1}\cdot Y_{d,\D_k}(z,t)=e^{-\pi\mathi \hbar\delta P} F^{(2)}_{d+\frac{(k+1)\pi}{p}}(z_2)^{-1}\cdot F^{(2)}_{d+\frac{k\pi}{p}}(z_2)e^{\pi\mathi \hbar\delta P}=e^{-\pi\mathi \hbar\delta P} S_{i+k}^{(2)}(u) e^{\pi\mathi \hbar\delta P}.
\end{equation}
Here the second identity follows from the definition of Stokes matrices: since $\tau_{il}<d<\tau_{il+1}$, we have $\tau_{(i+k)l}<d+\frac{k\pi}{p}<\tau_{(i+k)l+1}$ and $\tau_{(i+k+1)l}<d+\frac{(k+1)\pi}{p}<\tau_{(i+k+1)l+1}$. Here recall $l=\frac{\#({\rm aSR}(u)\cap[0,2\pi))}{2p }$.

Similarly, we have 
\begin{align*}
  Y_{d+\frac{\pi}{p},\D_{k+1}}(z,t)^{-1}\cdot Y_{d,\D_{k+1}}(z,t)&=F^{(1)}_{d+\frac{\pi}{p}}(z_1)^{-1}\cdot F^{(1)}_{d}(z_1)=S^{(1)}_i(u)\\
    Y_{d+\frac{\pi}{p},\D_k}(z,t)^{-1}\cdot Y_{d,\D_{k}}(z,t)&=F^{(1)}_{d+\frac{\pi}{p}}(z_1)^{-1}\cdot F^{(1)}_{d}(z_1)=S^{(1)}_i(u)\\
    Y_{d+\frac{\pi}{p},\D_{k+1}}(z,t)^{-1}\cdot Y_{d+\frac{\pi}{p},\D_{k}}(z,t)&=e^{-\pi\mathi \hbar\delta P}F^{(2)}_{d+\frac{(k+1)\pi}{p}}(z_2)^{-1}\cdot F^{(2)}_{d+\frac{k\pi}{p}}(z_2)e^{\pi\mathi \hbar\delta P}=e^{-\pi\mathi \hbar\delta P}S^{(2)}_{i+k}(u)e^{\pi\mathi \hbar\delta P}.
\end{align*}
Then the identity \eqref{SS=SS}
can be obtained by computing the monodromy from $Y_{d,\D_k}(z,t)$ to $Y_{d+\frac{\pi}{p},\D_{k+1}}(z,t)$ in two equivalent ways: from 
$Y_{d,\D_k}(z,t)$ to 
$Y_{d,\D_{k+1}}(z,t)$ to $Y_{d+\frac{\pi}{p},\D_{k+1}}(z,t)$
, and from 
$Y_{d,\D_k}(z,t)$ to 
$Y_{d+\frac{\pi}{p},\D_k}(z,t)$ to 
$Y_{d+\frac{\pi}{p},\D_{k+1}}(z,t)$.

To prove the identity \eqref{pRLL12}, by Corollary \ref{mainpro}, we consider the following solutions of the system \eqref{kKZ1}-\eqref{kKZ2} (the terms $e^{\pm \pi\mathi\hbar\delta P}$ account for the argument of $(t-1)$):
\begin{align*}{U}_d(z;t){X}_{d+\frac{\pi}{p}}(\omega)\cdot p^{\hbar\delta P}=&Y_{d,\D_{(1,+\infty)}}(z,t)={_1W}_{d}(z_1;z_2)F^{(2)}_{d}(z_2)\cdot e^{\pi\mathi \hbar\delta P},\\
{_2W}_{d}(z_2;z_1)F^{(1)}_{d}(z_1)=&Y_{d,\D_{(0,1)}}={U}_d(z;t){X}_{d}(\omega)\cdot p^{\hbar\delta P}, \\
{_2W}_{d}(z_2;z_1)F^{(1)}_{d+\frac{\pi}{p}}(z_1)=&Y_{d+\frac{\pi}{p},\D_{2p-1}}={_1W}_{d+\frac{\pi}{p}}(z_1;z_2)F^{(2)}_{d}(z_2)\cdot e^{-\pi\mathi \hbar\delta P},  \\
{_2W}_{d+\frac{\pi}{p}}(z_2;z_1)F^{(1)}_{d}(z_1)=&Y_{d,\D_{1}}={_1W}_{d}(z_1;z_2)F^{(2)}_{d+\frac{\pi}{p}}(z_2)\cdot e^{\pi\mathi \hbar\delta P},\\
{_2W}_{d+\frac{\pi}{p}}(z_2;z_1)F^{(1)}_{d+\frac{\pi}{p}}(z_1)=&Y_{d+\frac{\pi}{p},\D_{(0,1)}}={U}_{d+\frac{\pi}{p}}(z;t){X}_{d+\frac{\pi}{p}}(\omega)\cdot p^{\hbar\delta P}, \\
{U}_{d+\frac{\pi}{p}}(z;t){X}_{d}(\omega)\cdot p^{\hbar\delta P}=&Y_{d+\frac{\pi}{p},\D_{(1,+\infty)}}={_1W}_{d+\frac{\pi}{p}}(z_1;z_2)F^{(2)}_{d+\frac{\pi}{p}}(z_2)\cdot e^{-\pi\mathi \hbar\delta P}.
\end{align*} 
Figure 1 illustrates the asymptotic regions defining these solutions and the analytic continuation paths for calculating the connection matrices among them.

\begin{figure}[h]
\centering
\begin{tikzpicture}[scale=0.8, every node/.style={font=\small}, >=Stealth]

\def\lenA{1.0}
\def\lenB{1.8}
\def\cosang{0.866}
\def\sinang{0.5}

\begin{scope}[shift={(0,0)}]
    \draw[thick] (0,0) -- (2.5,0);
    \draw[thick] (0,0) -- (2.165,1.25);
    \fill (0,0) circle (1.5pt) node[below left] {$0$};
    \fill (\lenA,0) circle (1pt) node[below] {$z_1$};
    \fill (\lenB,0) circle (1pt) node[below] {$z_2$};
    \node at (1.25,-0.8) {$Y_{d,\D_{(1,+\infty)}}$};
\end{scope}

\begin{scope}[shift={(4,2.5)}]
    \draw[thick] (0,0) -- (2.5,0);
    \draw[thick] (0,0) -- (2.165,1.25);
    \fill (0,0) circle (1.5pt) node[below left] {$0$};
    \fill (\lenA,0) circle (1pt) node[below] {$z_2$};
    \fill (\lenB,0) circle (1pt) node[below] {$z_1$};
    \node at (1.25,-0.8) {$Y_{d,\D_{(0,1)}}$};
\end{scope}

\begin{scope}[shift={(8,2.5)}]
    \draw[thick] (0,0) -- (2.5,0);
    \draw[thick] (0,0) -- (2.165,1.25);
    \fill (0,0) circle (1.5pt) node[below left] {$0$};
    \fill (\lenA,0) circle (1pt) node[below] {$z_2$};
    \fill (\lenB*\cosang,\lenB*\sinang) circle (1pt) node[above right] {$z_1$};
    \node at (1.25,-0.8) {$Y_{d+\frac{\pi}{p},\D_{2p-1}}$};
\end{scope}

\begin{scope}[shift={(11.5,2.5)}]
    \draw[thick] (0,0) -- (2.5,0);
    \draw[thick] (0,0) -- (2.165,1.25);
    \fill (0,0) circle (1.5pt) node[below left] {$0$};
    \fill (\lenB,0) circle (1pt) node[below] {$z_2$};
    \fill (\lenA*\cosang,\lenA*\sinang) circle (1pt) node[above right] {$z_1$};
    \node at (1.25,-0.8) {$Y_{d+\frac{\pi}{p},\D_{2p-1}}$};
\end{scope}

\begin{scope}[shift={(16,0)}]
    \draw[thick] (0,0) -- (2.5,0);
    \draw[thick] (0,0) -- (2.165,1.25);
    \fill (0,0) circle (1.5pt) node[below left] {$0$};
    \fill (\lenA*\cosang,\lenA*\sinang) circle (1pt) node[above left] {$z_1$};
    \fill (\lenB*\cosang,\lenB*\sinang) circle (1pt) node[above right] {$z_2$};
    \node at (1.25,-0.8) {$Y_{d+\frac{\pi}{p},\D_{(1,+\infty)}}$};
\end{scope}

\begin{scope}[shift={(4,-2.5)}]
    \draw[thick] (0,0) -- (2.5,0);
    \draw[thick] (0,0) -- (2.165,1.25);
    \fill (0,0) circle (1.5pt) node[below left] {$0$};
    \fill (\lenA,0) circle (1pt) node[below] {$z_1$};
    \fill (\lenB*\cosang,\lenB*\sinang) circle (1pt) node[above right] {$z_2$};
    \node at (1.25,-0.8) {$Y_{d,\D_1}$};
\end{scope}

\begin{scope}[shift={(7.5,-2.5)}]
    \draw[thick] (0,0) -- (2.5,0);
    \draw[thick] (0,0) -- (2.165,1.25);
    \fill (0,0) circle (1.5pt) node[below left] {$0$};
    \fill (\lenB,0) circle (1pt) node[below] {$z_1$};
    \fill (\lenA*\cosang,\lenA*\sinang) circle (1pt) node[above right] {$z_2$};
    \node at (1.25,-0.8) {$Y_{d,\D_1}$};
\end{scope}

\begin{scope}[shift={(12,-2.5)}]
    \draw[thick] (0,0) -- (2.5,0);
    \draw[thick] (0,0) -- (2.165,1.25);
    \fill (0,0) circle (1.5pt) node[below left] {$0$};
    \fill (\lenA*\cosang,\lenA*\sinang) circle (1pt) node[above left] {$z_2$};
    \fill (\lenB*\cosang,\lenB*\sinang) circle (1pt) node[above right] {$z_1$};
    \node at (1.25,-0.8) {$Y_{d+\frac{\pi}{p},\D_{(0,1)}}$};
\end{scope}

\draw[->,thick] (2.5,0.8) -- (3.5,1.8);
\draw[->,thick] (2.5,-0.8) -- (3.5,-1.8);
\draw[->,thick] (6.8,2.5) -- (7.5,2.5);
\draw[thick] (10.8,2.45) -- (11.2,2.45);
\draw[thick] (10.8,2.55) -- (11.2,2.55);
\draw[->,thick] (14.3,1.8) -- (15.5,0.8);
\draw[thick] (6.8,-2.55) -- (7.2,-2.55);
\draw[thick] (6.8,-2.45) -- (7.2,-2.45);
\draw[->,thick] (10.3,-2.5) -- (11.5,-2.5);
\draw[->,thick] (14.7,-1.8) -- (15.5,-0.8);

\end{tikzpicture}
\caption{The two connection paths for Relation~\ref{pRLL12}.}
\end{figure}

Note that we have $\tau_{(i+1)l}<d+\frac{\pi}{p}<\tau_{(i+1)l+1}$. 
It then follows from the definition of Stokes matrices that
\begin{align}\label{connect21}
Y_{d,\D_{(0,1)}}&=Y_{d+\frac{\pi}{p},\D_{2p-1}}\cdot S^{(1)}_{i}(u),\\
\label{connect3}
Y_{d+\frac{\pi}{p},\D_{2p-1}}&=Y_{d+\frac{\pi}{p},\D_{(1,+\infty)}}\cdot e^{\pi\mathi\hbar\delta P} S^{(2)}_{i}(u)e^{-\pi\mathi\hbar\delta P},\\
Y_{d,\D_{(1,+\infty)}}&=Y_{d,\D_{1}}\cdot e^{-\pi\mathi\hbar\delta P} S^{(2)}_{i}(u)e^{\pi\mathi\hbar\delta P},\\
Y_{d,\D_{1}}&=Y_{d+\frac{\pi}{p},\D_{(0,1)}}\cdot S^{(1)}_{i}(u).
\end{align}
The opposite conjugations in the second and third identities follow from the the branch convention in Remark \ref{tbranchlem}.
Meanwhile, by Lemma \ref{PR0}, if $i$ is even, then we have 
\begin{align}\label{connect2}
Y_{d,\D_{(1,+\infty)}}^{-1}\cdot Y_{d,\D_{(0,1)}}& =X_{d+\frac{\pi}{p}}(\omega)^{-1}\cdot X_{d}(\omega)=R^{(12)}_0.
\end{align}
Here the matrix $R_0\in {\rm End}(\IC^n)\otimes {\rm End}(\IC^n)\llbracket\hbar\rrbracket$ is 
\begin{equation}
    R_0=\sum_{i\ne j, i,j=1}^n E_{ii}\otimes E_{jj}+\sum_{i=1}^n E_{ii}\otimes E_{ii}+(e^{{\pi\mathi \hbar}}-e^{-{\pi\mathi \hbar}})\sum_{1\le j<i\le n}E_{ij}\otimes E_{ji},
\end{equation} 
and satisfies $R=e^{\pi\mathi \hbar\delta P}R_0=R_0e^{\pi\mathi \hbar\delta P}$.

Thus, we have two ways to compute the connection matrix between $Y_{d,\D_{(1,+\infty)}}$ and $Y_{d+\frac{\pi}{p},\D_{(1,+\infty)}}$:
on the one hand, we have (along the upper path in Figure 1)
\begin{align}\nonumber
    Y_{d+\frac{\pi}{p},\D_{(1,+\infty)}}&=Y_{d+\frac{\pi}{p},\D_{2p-1}}\cdot e^{\pi\mathi \hbar\delta P} S^{(2)}_{i}(u)^{-1} e^{-\pi\mathi \hbar\delta P} \\ \nonumber
    &=Y_{d,\D_{(0,1)}}\cdot S^{(1)}_{i}(u)^{-1}e^{\pi\mathi \hbar\delta P}S^{(2)}_{i}(u)^{-1} e^{-\pi\mathi \hbar\delta P}\\ \label{Rll1}
&=Y_{d,\D_{(1,+\infty)}}\cdot R_0^{(12)} S^{(1)}_{i}(u)^{-1} e^{\pi\mathi \hbar\delta P}S^{(2)}_{i}(u)^{-1} e^{-\pi\mathi \hbar\delta P}.
\end{align}
On the other hand, we have (along the lower path in Figure 1)
\begin{align}\label{Rll2}
Y_{d+\frac{\pi}{p},\D_{(1,+\infty)}}=Y_{d,\D_{(1,+\infty)}}\cdot  e^{-\pi\mathi \hbar\delta P}S^{(2)}_{i}(u)^{-1} e^{\pi\mathi \hbar\delta P} S^{(1)}_{i}(u)^{-1} R_0^{(12)}.
\end{align}
Comparing the identities \eqref{Rll1} and \eqref{Rll2} gives
\begin{equation}\label{R0LL}
R_0^{12}e^{\pi\mathi \hbar\delta P}S^{(2)}_i(u)e^{-\pi\mathi \hbar\delta P}S^{(1)}_i(u)=S^{(1)}_i(u)e^{-\pi\mathi \hbar\delta P}S^{(2)}_i(u)e^{\pi\mathi \hbar\delta P} R_0^{12}.
\end{equation} 
Together with $R=e^{\pi \mathi \hbar\delta P}R_0=R_0e^{\pi \mathi \hbar\delta P}$, it proves the identity \eqref{pRLL12}.

The same argument applied to the Stokes matrix $S_{i+1}(u)$ (exchanging the roles of upper and lower Stokes matrices) yields the identity \eqref{pRLL1}.

To prove the identity \eqref{pRLL4}, let us consider the solutions
\begin{align*}
{U}_d(z;t){X}_{d+\frac{\pi}{p}}(\omega)\cdot p^{\hbar\delta P}&=Y_{d,\D_{(1,+\infty)}}={_1W}_{d}(z_1;z_2)F^{(2)}_{d}(z_2)\cdot e^{\pi\mathi \hbar\delta P},\\
{_2W}_{d}(z_2;z_1)F^{(1)}_{d}(z_1)&=Y_{d,\D_{(0,1)}}={U}_d(z;t){X}_{d}(\omega)\cdot p^{\hbar\delta P}, \\
{_2W}_{d}(z_2;z_1)F^{(1)}_{d-\frac{\pi}{p}}(z_1)&=Y_{d-\frac{\pi}{p},\D_{1}}={_1W}_{d-\frac{\pi}{p}}(z_1;z_2)F^{(2)}_{d}(z_2)\cdot e^{\pi\mathi \hbar\delta P},  \\ 
{_2W}_{d+\frac{\pi}{p}}(z_2;z_1)F^{(1)}_{d}(z_1)&=Y_{d,\D_{1}}={_1W}_{d}(z_1;z_2)F^{(2)}_{d+\frac{\pi}{p}}(z_2)\cdot e^{\pi\mathi \hbar\delta P},\\
{_1W}_{d-\frac{\pi}{p}}(z_1;z_2)F^{(2)}_{d+\frac{\pi}{p}}(z_2)\cdot e^{\pi\mathi \hbar\delta P}&=Y_{d-\frac{\pi}{p},\D_{2}}={_2W}_{d+\frac{\pi}{p}}(z_2;z_1)F^{(1)}_{d-\frac{\pi}{p}}(z_1).
\end{align*} 
Figure 2 illustrates the asymptotic regions defining these solutions and the paths for calculating the connection matrices among them.

\begin{figure}[h]
\centering
\begin{tikzpicture}[scale=0.75, every node/.style={font=\small}, >=Stealth]

\def\lenA{1.0}
\def\lenB{1.8}
\def\cosang{0.9397}
\def\sinang{0.3420}

\begin{scope}[shift={(0,0)}]
    \draw[thick] (0,0) -- (2.5,0);
    \draw[thick] (0,0) -- (2.5*\cosang, 2.5*\sinang);
    \draw[thick] (0,0) -- (2.5*\cosang, -2.5*\sinang);
    \fill (0,0) circle (1.5pt) node[below left] {$0$};
    \fill (\lenA,0) circle (1pt) node[below] {$z_1$};
    \fill (\lenB,0) circle (1pt) node[below] {$z_2$};
    \node at (1.25,-1.0) {$Y_{d,\D_{(1,+\infty)}}$};
\end{scope}

\begin{scope}[shift={(4,2.7)}]
    \draw[thick] (0,0) -- (2.5,0);
    \draw[thick] (0,0) -- (2.5*\cosang, 2.5*\sinang);
    \draw[thick] (0,0) -- (2.5*\cosang, -2.5*\sinang);
    \fill (0,0) circle (1.5pt) node[below left] {$0$};
    \fill (\lenA,0) circle (1pt) node[below] {$z_2$};
    \fill (\lenB,0) circle (1pt) node[below] {$z_1$};
    \node at (1.25,-1.0) {$Y_{d,\D_{(0,1)}}$};
\end{scope}

\begin{scope}[shift={(8,2.7)}]
    \draw[thick] (0,0) -- (2.5,0);
    \draw[thick] (0,0) -- (2.5*\cosang, 2.5*\sinang);
    \draw[thick] (0,0) -- (2.5*\cosang, -2.5*\sinang);
    \fill (0,0) circle (1.5pt) node[below left] {$0$};
    \fill (\lenA,0) circle (1pt) node[above] {$z_2$};
    \fill (\lenB*\cosang,-\lenB*\sinang) circle (1pt) node[below right] {$z_1$};
    \node at (1.25,-1.25) {$Y_{d-\frac{\pi}{p},\D_1}$};
\end{scope}

\begin{scope}[shift={(11.5,2.7)}]
    \draw[thick] (0,0) -- (2.5,0);
    \draw[thick] (0,0) -- (2.5*\cosang, 2.5*\sinang);
    \draw[thick] (0,0) -- (2.5*\cosang, -2.5*\sinang);
    \fill (0,0) circle (1.5pt) node[below left] {$0$};
    \fill (\lenB,0) circle (1pt) node[above] {$z_2$};
    \fill (\lenA*\cosang,-\lenA*\sinang) circle (1pt) node[below right] {$z_1$};
    \node at (1.25,-1.25) {$Y_{d-\frac{\pi}{p},\D_1}$};
\end{scope}

\begin{scope}[shift={(16,0)}]
    \draw[thick] (0,0) -- (2.5,0);
    \draw[thick] (0,0) -- (2.5*\cosang, 2.5*\sinang);
    \draw[thick] (0,0) -- (2.5*\cosang, -2.5*\sinang);
    \fill (0,0) circle (1.5pt) node[below left] {$0$};
    \fill (\lenA*\cosang,-\lenA*\sinang) circle (1pt) node[below right] {$z_1$};
    \fill (\lenB*\cosang,\lenB*\sinang) circle (1pt) node[above right] {$z_2$};
    \node at (1.25,-1.25) {$Y_{d-\frac{\pi}{p},\D_2}$};
\end{scope}

\begin{scope}[shift={(4,-2.7)}]
    \draw[thick] (0,0) -- (2.5,0);
    \draw[thick] (0,0) -- (2.5*\cosang, 2.5*\sinang);
    \draw[thick] (0,0) -- (2.5*\cosang, -2.5*\sinang);
    \fill (0,0) circle (1.5pt) node[below left] {$0$};
    \fill (\lenA,0) circle (1pt) node[below] {$z_1$};
    \fill (\lenB*\cosang, \lenB*\sinang) circle (1pt) node[above right] {$z_2$};
    \node at (1.25,-1.0) {$Y_{d,\D_1}$};
\end{scope}

\begin{scope}[shift={(7.5,-2.7)}]
    \draw[thick] (0,0) -- (2.5,0);
    \draw[thick] (0,0) -- (2.5*\cosang, 2.5*\sinang);
    \draw[thick] (0,0) -- (2.5*\cosang, -2.5*\sinang);
    \fill (0,0) circle (1.5pt) node[below left] {$0$};
    \fill (\lenB,0) circle (1pt) node[below] {$z_1$};
    \fill (\lenA*\cosang,\lenA*\sinang) circle (1pt) node[above right] {$z_2$};
    \node at (1.25,-1.0) {$Y_{d,\D_1}$};
\end{scope}

\begin{scope}[shift={(12,-2.7)}]
    \draw[thick] (0,0) -- (2.5,0);
    \draw[thick] (0,0) -- (2.5*\cosang, 2.5*\sinang);
    \draw[thick] (0,0) -- (2.5*\cosang, -2.5*\sinang);
    \fill (0,0) circle (1.5pt) node[below left] {$0$};
    \fill (\lenB*\cosang,-\lenB*\sinang) circle (1pt) node[below right] {$z_1$};
    \fill (\lenA*\cosang,\lenA*\sinang) circle (1pt) node[above right] {$z_2$};
    \node at (1.25,-1.25) {$Y_{d-\frac{\pi}{p},\D_2}$};
\end{scope}

\draw[->,thick] (2.5,0.8) -- (3.5,1.9);
\draw[->,thick] (2.5,-0.8) -- (3.5,-1.9);
\draw[->,thick] (6.8,2.7) -- (7.5,2.7);
\draw[thick] (10.8,2.65) -- (11.2,2.65);
\draw[thick] (10.8,2.75) -- (11.2,2.75);
\draw[->,thick] (14.3,1.9) -- (15.5,0.8);
\draw[thick] (6.8,-2.75) -- (7.2,-2.75);
\draw[thick] (6.8,-2.65) -- (7.2,-2.65);
\draw[->,thick] (10.3,-2.7) -- (11.5,-2.7);
\draw[thick] (14.8,-2.0) -- (15.15,-1.45);
\draw[thick] (15.0,-2.05) -- (15.35,-1.5);

\end{tikzpicture}
\caption{The two connection paths for Relation~\ref{pRLL4}.}

\end{figure}

We have two ways to compute the connection matrix between $Y_{d,\D_{(1,+\infty)}}$ and $Y_{d-\frac{\pi}{p},\D_{2}}$:
on the one hand, we have (along the upper path in Figure 2)
\begin{align*}
    Y_{d-\frac{\pi}{p},\D_{2}}&=Y_{d-\frac{\pi}{p},\D_{1}}\cdot e^{-\pi\mathi \hbar\delta P}S^{(2)}_{i}(u)^{-1}e^{\pi\mathi \hbar\delta P} \\
    &=Y_{d,\D_{(0,1)}}\cdot S^{(1)}_{i-1}(u) e^{-\pi\mathi \hbar\delta P}S^{(2)}_{i}(u)^{-1}e^{\pi\mathi \hbar\delta P} \\
&=Y_{d,\D_{(1,+\infty)}}\cdot R_0^{(12)} S^{(1)}_{i-1}(u) e^{-\pi\mathi \hbar\delta P}S^{(2)}_{i}(u)^{-1}e^{\pi\mathi \hbar\delta P},
\end{align*}
on the other hand, we have (along the lower path in Figure 2)
\begin{align*}
Y_{d-\frac{\pi}{p},\D_{2}}=Y_{d,\D_{1}}\cdot S^{(1)}_{i-1}(u)=Y_{d,\D_{(1,+\infty)}}\cdot e^{-\pi\mathi \hbar\delta P}S^{(2)}_{i}(u)^{-1} e^{\pi\mathi \hbar\delta P} S^{(1)}_{i-1}(u).
\end{align*}
It proves the identity \eqref{pRLL4}. 
\end{proof}

\begin{lem}\label{PR0}
The Stokes matrix $X_{d+\frac{\pi}{p}}(\omega)^{-1}\cdot X_{d}(\omega)$, associated to the direction $d$ with $\tau_{il}<d<\tau_{il+1}$ for $i$ even, of the equation \eqref{Xeq} is equal to $R^{12}_0$. For $i$ odd the Stokes matrix is the lower triangular counterpart.
\end{lem}
\begin{proof}
Under the coordinate change $\eta=\omega^p=(t^{-p}-1)^{-1}z^p$, the equation \eqref{Xeq} becomes
\begin{align}\label{Xetaeq}
\frac{{dX}}{d \eta}=\left(\frac{ u^{(2)}}{p\eta^{2}}+\frac{\hbar P}{\eta}\right)X.
\end{align}
The transition matrix $X_{d+\frac{\pi}{p}}(\omega)^{-1}\cdot X_{d}(\omega)$ is then the Stokes matrix $X_{pd+\pi}(\eta)^{-1}X_{pd}(\eta)$ of \eqref{Xetaeq}.
The Stokes matrices of the equation \eqref{Xetaeq} were computed in \cite{Xu}, see also \cite{LWX}. We include a brief proof here for the convenience of the reader. 

The coefficient of the equation \eqref{Xetaeq} is in $\Uph\otimes {\rm End}(\IC^n)\otimes {\rm End}(\IC^n)$. Hence the equation may be viewed as an $n^2\times n^2$ block linear system. The equation is then decomposed to multiple $2\times 2$ and $1\times 1$ block systems: let $\{v_i\}_{1\le i\le n}$ be the standard basis of $\mathbb{C}^n$, and $\{v_i\otimes v_j\}$ the basis of $\mathbb{C}^n\otimes \mathbb{C}^n$, we have for any $v\in \wUph$
\begin{align}
    P(v\otimes v_i\otimes v_j)&=v\otimes v_j\otimes v_i,\\
    u^{(2)}(v\otimes v_i\otimes v_j)&=v\otimes v_i\otimes u_jv_j, \ \text{for } i,j=1,...,n
\end{align}
Thus, for fixed index $i\ne j$, the equation \eqref{Xetaeq} can be restricted to the subspace spanned by the two vectors $v_i\otimes v_j$ and $v_j\otimes v_i$, reducing to a $2\times 2$ block system of the form 
\begin{equation}\label{pd2by2}
\frac{d\mathcal{M}}{d\eta}=\left({\left(
  \begin{array}{cc}
    u_j/p & 0  \\
    0 & u_i/p
  \end{array}
\right)}\frac{1}{\eta^{2}}
+ {\left(
  \begin{array}{cc}
    0 &\hbar \\
    \hbar & 0
  \end{array}
\right)}\frac{1}{\eta}\right)\mathcal{M}.
\end{equation}
This equation can be solved exactly and (under the assumption that we made on $u$ and the choice of $\tau_{il}<d<\tau_{il+1}$ with $i$ even) shown to have the Stokes matrix ${\left(\begin{array}{cc}
1&e^{{\pi\mathi \hbar}}-e^{-{\pi\mathi \hbar}}\\
0&1
\end{array}\right)}$. (Note that if $\tau_{il}<d<\tau_{il+1}$, and $i$ is odd, the corresponding Stokes matrix is lower triangular.)

Thus, the Stokes matrix of the $2\times 2$ subsystem, associated to each pair $i\ne j$, contributes to the $E_{ii}\otimes E_{jj}+(e^{{\pi\mathi \hbar}}-e^{-{\pi\mathi \hbar}})E_{ij}\otimes E_{ji}$ term of $R_0$. The Stokes matrices of the $1\times 1$ block systems are trivial and contribute the terms $ E_{ii}\otimes E_{ii}$ in $R_0$. It concludes the computation.
\end{proof}

\begin{rmk}
The Stokes matrices have ones along the diagonal. Sometimes it is more convenient to use the normalized Stokes matrices $\mathcal{S}_i(u)=S_i(u)e^{\pi \mathi D_1} $ by adding the formal monodromy part.
Note that we have the commutativity conditions 
\begin{equation}\label{commutativity}
\big[S_i^{(2)}(u),\; D_1^{(1)}+\hbar\delta P\big]=0,
\qquad
\big[S_i^{(1)}(u),\;D_1^{(2)}+\hbar \delta P\big]=0,
\qquad
\big[R_0,\; D_1^{(1)}+ D_1^{(2)}\big]=0.
\end{equation}
For example, the first identity follows from the fact that
${D^{(1)}_1+\hbar\delta P}$ commutes with the coefficient matrix of the equation \eqref{KZin2}.
Multiply $e^{\pi \mathi ( D_1^{(1)}+ D_1^{(2)}+\hbar\delta P)}$ to both sides of \eqref{R0LL}, and applying the commutativity conditions \eqref{commutativity} gives the identity 
\[R^{12}\mathcal{S}^{(2)}_i(u)\mathcal{S}^{(1)}_i(u) =\mathcal{S}^{(1)}_i(u)\mathcal{S}^{(2)}_i(u) R^{12}.\]
In the special case $p = 1$, the above identity gives the FRT realization of the quantum group $U_\hbar(\mathfrak{gl}_n)$. 
\end{rmk}

\section{Quantization of Riemann-Hilbert-Birkhoff maps}\label{QRHBmap}
This section interprets Theorem \ref{thmRLL} as a quantization of the classical RHB maps. In Section \ref{PoissonStokes} we define an associative algebra $U_{\hbar}^{(p)}(R)$ via the exchange relations in Theorem \ref{thmRLL}, and compute the induced Poisson brackets in the semiclassical limit on the space of classical Stokes matrices. In Section \ref{PoissonDeRham} we recall the Poisson brackets on the space of meromorphic differential equations, as the semiclassical limit of the coefficient algebra $\Uph$. In Section \ref{qRHBmap} we use the quantum Stokes matrices to define an algebra homomorphism from $U_{\hbar}^{(p)}(R)$ to $\Uph$, and show that its classical limit gives the pull back of the classical RHB maps. The last section then interpret the Poisson geometric nature of the RHB maps from this new viewpoint.

\subsection{Poisson brackets on the space of Stokes matrices via the semiclassical limit}\label{PoissonStokes}
\begin{defi}\label{defUpR}
Let $R$ be the standard $R$-matrix. The algebra $U_{\hbar}^{(p)}(R)$ is the topological $\C\llbracket\hbar\rrbracket$-algebra generated by elements $\{(l_{2k-1})_{ji},(l_{2k})_{ij}\}_{k=1,...,p,1\le i< j\le n}$ and $\{K_a, K_a^{-1}\}_{a=1,...,n}$, subject to the following relations. Set
\begin{equation*}\label{eq:Lpm-def}
L_{1}=\sum_{i,j=1}^{n}(l_1)_{ij}\otimes E_{ij}, \quad \cdots \quad , L_{2p}=\sum_{i,j=1}^{n}(l_{2p})_{ij}\otimes E_{ij}, \quad K=\sum_{a=1}^n K_a\otimes E_{aa}
\end{equation*}
as elements in $U_{\hbar}^{(p)}(R)\otimes \mathrm{End}(\mathbb{C}^{n})$,
where we set 
\[(l_{2k})_{ij}=0 \ \text{ for } i>j ; \quad (l_{2k-1})_{ij}=0 \ \text{ for } i<j;  \quad \text{ and } \ (l_{2k-1})_{ii}=1=(l_{2k})_{ii}.\]
Then the defining relations are the following identities in $U_{\hbar}^{(p)}(R)\otimes \mathrm{End}(\mathbb{C}^{n})\otimes \mathrm{End}(\mathbb{C}^{n})$:
\begin{align*}
e^{\pi\mathi \hbar\delta P}L^{(1)}_ie^{-\pi\mathi \hbar\delta P}L_{i+k}^{(2)}&=L^{(2)}_{i+k}e^{\pi\mathi \hbar\delta P}L^{(1)}_ie^{-\pi\mathi \hbar\delta P}, \ \ \text{for any $i$ and } 1<k<2p-1,\\
R^{12}L^{(1)}_ie^{-\pi\mathi \hbar\delta P}L^{(2)}_i &=L^{(2)}_ie^{-\pi\mathi \hbar\delta P}L^{(1)}_i R^{12}, \ \ \text{for $i$ odd} \\ 
R^{12}L^{(2)}_ie^{-\pi\mathi \hbar\delta P}L^{(1)}_i &=L^{(1)}_ie^{-\pi\mathi \hbar\delta P}L^{(2)}_iR^{12}, \ \ \text{for $i$ even} \\
L^{(2)}_{i-1}e^{-\pi\mathi \hbar\delta P} L^{(1)}_{i}e^{\pi\mathi \hbar\delta P} &=e^{-\pi\mathi \hbar\delta P}L^{(1)}_{i} R^{12} L^{(2)}_{i-1}, \ \ \text{for $i$ odd},\\
L^{(1)}_{i-1}e^{-\pi\mathi \hbar\delta P}L^{(2)}_{i}e^{\pi\mathi \hbar\delta P} &=e^{-\pi\mathi \hbar\delta P}L^{(2)}_{i} R^{12} L^{(1)}_{i-1}, \ \ \text{for $i$ even},
\end{align*}
and the generators $K_a, K_a^{-1}$ satisfy
\begin{equation}
K_a K_a^{-1}=1=K_a^{-1}K_a, \qquad
[K_a,K_b ]=0, \qquad K^{(1)}L_i^{(2)}(K^{(1)})^{-1}=
e^{-2\pi i\hbar\delta P}L_i^{(2)}e^{2\pi i\hbar\delta P}.
\end{equation}
Here the indices are extended cyclically by $L_{i+2p}=K^{-1} L_i K$.
\end{defi}
\begin{rmk}
The exchange relations are reminiscent of the
$R$-matrix exchange relations in the lattice current algebras introduced in \cite{AFFS, AGS}. The lattice
current variables $J_n$ are full matrix (not just triangular) variables and satisfy
\[J^{(1)}_{i}J^{(2)}_{j}=J^{(2)}_{j}J^{(1)}_{i} \text{ for } |i-j|>1, \text{ and } \quad J^{(1)}_{i}R^{12}J^{(2)}_{i+1}=J^{(2)}_{i+1}J^{(1)}_{i},
  \quad
R^{21}J^{(1)}_{i}J^{(2)}_{i}=J^{(2)}_{i}J^{(1)}_{i}R^{12}.
\]
The $L$ matrices, or the triangular Stokes matrices $S_1,\ldots,S_{2p}$, play an analogous role: they form a
cyclic family of $R$-matrix exchange variables attached to the Stokes
rays of an irregular singularity.
\end{rmk}

The semiclassical limit of the associative algebra $U_{\hbar}^{(p)}(R)$ induces a Poisson bracket on $(U_-\times U_+)^p\times \h$: for $i=1,\ldots,2p$, set $(s_i)_{ab}=\overline{(l_i)_{ab}}$ the image of the generator $(l_i)_{ab}$ in $U_{\hbar}^{(p)}(R)/\hbar U_{\hbar}^{(p)}(R)$, regarded as the $(ab)$-entry coordinates on the $i$-th component of \[(s_1,s_2,\ldots, s_{2p-1},s_{2p}, \lambda)\in U_-\times U_+\times\cdots \times U_-\times U_+\times \h.\]
Then
the Poisson bracket is defined in the usual way
\begin{equation}
 \{(s_i)_{ab},(s_j)_{cd}\}
 :=\overline{\hbar^{-1}[(l_i)_{ab},(l_j)_{cd}]}.
\end{equation}
The Poisson bracket can be explicitly written as follows.
Note that the first order expansion of the (quantum) $R$-matrix is
\begin{equation}\label{Rexp}
R=1+2\pi\mathi\hbar\,r+O(\hbar^2),
\end{equation}
where $r$ is the standard classical $r$-matrix
\[
r=\frac{1}{2}\sum_{a=1}^nE_{aa}\otimes E_{aa}+\sum_{1\leq b<a\leq n}E_{ab}\otimes E_{ba}.
\]
The first order expansion of $e^{\pi\mathi \hbar\delta P}$ is
\begin{equation}\label{hPexp}
e^{\pi\mathi \hbar\delta P}=1+2\pi\mathi \hbar \Omega_\h+O(\hbar^2),
\end{equation}
where 
\[
\Omega_\h=\frac{1}{2}\sum_{a=1}^nE_{aa}\otimes E_{aa}
\]
is just the Cartan part of the $r$-matrix. 
Using \eqref{Rexp} and \eqref{hPexp}, a direct computation leads to
\begin{pro}\label{StokesPoisson}
For $i=1,...,2p$, let $s_i=((s_i)_{ab})_{a,b=1,...,n}$ and $\lambda={\rm diag}(\lambda_{1},\ldots,\lambda_{n})$ be the matrix valued functions on $(U_-\times U_+)^p\times \h$.
Then the Poisson bracket takes the form
\begin{align}\label{bracket1}
 \frac{1}{2\pi\mathi}\{s_i^{(1)},s_{i+k}^{(2)}\}
 &=
 s_{i+k}^{(2)}\Omega_\h s_i^{(1)}-s_{i+k}^{(2)} s_i^{(1)} \Omega_\h
 -\Omega_\h s_i^{(1)}s_{i+k}^{(2)}+s_i^{(1)}\Omega_\h s_{i+k}^{(2)},
 &&1<k<2p-1, \\
\frac{1}{2\pi\mathi} \{s_i^{(1)},s_i^{(2)}\}
 &=
 -s_i^{(2)}\Omega_\h s_i^{(1)}
 +s_i^{(2)}s_i^{(1)}r
 -r s_i^{(1)}s_i^{(2)}
 +s_i^{(1)}\Omega_\h s_i^{(2)},
 &&i\ \text{odd},\\
\frac{1}{2\pi\mathi} \{s_i^{(1)},s_i^{(2)}\}
 &=
 s_i^{(1)}\Omega_\h s_i^{(2)}
 -s_i^{(1)}s_i^{(2)}r
+r s_i^{(2)}s_i^{(1)}
 -s_i^{(2)}\Omega_\h s_i^{(1)},
 &&i\ \text{even}, \\ \label{4Pidentity}
\frac{1}{2\pi\mathi} \{s_i^{(1)},s_{i-1}^{(2)}\}
 &=
 \Omega_\h s_i^{(1)}s_{i-1}^{(2)}
 -s_{i-1}^{(2)}\Omega_\h  s_i^{(1)}+s_{i-1}^{(2)}s_i^{(1)}\Omega_\h
 -s_i^{(1)}r s_{i-1}^{(2)},
 &&i\ \text{odd}, \\
\frac{1}{2\pi\mathi} \{s_i^{(2)},s_{i-1}^{(1)}\}
 &=
\Omega_\h s_i^{(2)}s_{i-1}^{(1)}
 -s_{i-1}^{(1)}\Omega_\h s_i^{(2)}+s_{i-1}^{(1)}s_i^{(2)}\Omega_\h
 -s_i^{(2)}r s_{i-1}^{(1)},
 &&i\ \text{even},\\ \label{6Pidentity}
 \{\lambda_a,\lambda_b\}&=0,
 \qquad
 \{\lambda_a,(s_i)_{bc}\}
 =
 -(\delta_{ab}-\delta_{ac})(s_i)_{bc}, &&\
a,b,c=1,\ldots,n.
\end{align}
Here we use the tensor notation:
let $X=(x_{rs})$ and $Y=(y_{uv})$ be two
matrix-valued functions on $(U_-\times U_+)^p\times \frak{h}$, set
\[
 X^{(1)}=X\otimes 1
 =\sum_{r,s}x_{rs}\otimes E_{rs}\otimes 1,\quad
 Y^{(2)}=1\otimes Y
 =\sum_{u,v}y_{uv}\otimes 1\otimes E_{uv},\]
then
\[ \{X^{(1)},Y^{(2)}\}:=\sum_{r,s,u,v}\{x_{rs},y_{uv}\}
 \otimes E_{rs}\otimes E_{uv}.\]
\end{pro}
Note that by our convention $s_0=e^{2\pi\mathi \lambda} s_{2p}e^{-2\pi\mathi \lambda}$, then the bracket for the pair $s_1$ and $s_{2p}$ can be derived from \eqref{4Pidentity} and \eqref{6Pidentity} by setting $i=1$.
One can check directly the brackets above satisfy the Jacobi identity. 
It gives, in a new way, the Poisson structure on the space of classical Stokes matrices. See Corollary \ref{twoPoisson}. The mixed terms between different Stokes matrices as in \eqref{bracket1}
have a similar structure to the quasi-Poisson fusion rule in the sense of \cite{AKSM}, or to the mixed product in the sense of \cite{LuM}. 

\subsection{Poisson brackets on the space of meromorphic connections via the semiclassical limit}\label{PoissonDeRham}

We first interpret $\Uph$ via a truncation of current algebra.
Take the Lie algebra
\[
 J_p=t\mathfrak{gl}_n[t]/t^p\mathfrak{gl}_n[t],
\]
and write \(X^{(a)}=Xt^a\).  The bilinear form
\[
 \omega_u(X^{(a)},Y^{(b)})
 =\delta_{a+b,p}\operatorname{tr}(u[X,Y])
\]
is a Lie-algebra $2$-cocycle.  Let
\(\widetilde J_{p,u}=J_p\oplus\mathbb Cc\) be the corresponding central
extension and take the direct sum of Lie algebras
\[
 \widetilde{\frak l}_{p,u}
 =\frak{gl}_n\oplus\widetilde J_{p,u}.
\]
Let
$U_\hbar(\widetilde{\frak l}_{p,u})$ denote the
$\mathbb C[\hbar]$-algebra with relations
$xy-yx=\hbar[x,y]$, then the defining relations of $\Uph$ give a canonical isomorphism
\begin{equation}
 \Uph\cong U_\hbar(\widetilde{\frak l}_{p,u})/(c-1)~:~ \hbar T_{[a+1]}\mapsto X^{(a)}, \hbar T\mapsto X.
\end{equation}
Using the PBW theorem,
the following lemma is standard.
\begin{lem}
Ordered monomials in the generators $\hbar e_{ij}$ and
$\hbar e_{ij}^{(a)}$ form a $\C[\hbar]$-basis of $\Uph$.  In
particular, 
\begin{equation}
\Uph/\hbar\Uph\cong\operatorname{Sym}(\frak l_{p}), \text{  with the vector space } {\frak l}_{p}:=\frak{gl}_n\oplus J_{p}.
\end{equation}
Furthermore, reducing 
modulo $\hbar$ gives a canonical isomorphism of complete commutative
algebras
\begin{equation}\label{eq:completed-special-fibre}
  \wUph/\hbar\wUph
  \cong
  \widehat{{\rm Sym}}(\mathfrak l_{p})
  :=
  \varprojlim_N
  {\rm Sym}(\mathfrak l_{p})/\mathfrak m_0^N.
\end{equation}
where $\mathfrak m_0$ is the maximal ideal generated by
$\mathfrak l_{p}$ in $\operatorname{Sym}(\mathfrak l_{p})$. 
\end{lem}
The semiclassical limit of $\Uph$ induces a Poisson bracket on ${\rm Sym}(\frak l_{p})$ in the usual way
\begin{equation}
 \{\bar x,\bar y\}
 =\overline{\hbar^{-1}[x,y]}, \qquad x,y\in \Uph.
\end{equation}
Note that the dual $\frak l_{p}^*$, via the trace pairing, is isomorphic to
the affine
coefficient space of differential equations \eqref{heq}
\begin{equation}
\frak l_{p}^*\cong {\mathcal M}_{dR}^{(p)}(u)=\{(A_p,...,A_2,B)\in \frak{gl}_n\times \cdots \times \frak{gl}_n\}.
\end{equation}
Then in terms of the coordinate functions $\{b_{ij}\}_{i,j=1,...,n}$ of $B$ and $\{a^{(m-1)}_{ij}\}_{i,j=1,...,n}$ of $A_m$ for $m=2,...,p$, the induced Poisson bracket is given by
\begin{align}\label{Poisson1}
\{a^{(m)}_{ij},a^{(s)}_{kl}\}&=\left\{\begin{array}{lr} \delta_{jk}a_{il}^{(m+s)}-\delta_{li} a_{kj}^{(m+s)},  & \text{ if } m+s\le p-1 \\
\delta_{jk}\delta_{il}(u_i-u_j), & \text{ if }  m+s=p,\\
0, & \text{ if }  m+s\ge p+1,
             \end{array}
\right.\\ \label{Poisson2}
    \{b_{ij},b_{kl}\}&= \delta_{jk}b_{il}-\delta_{li} b_{kj}, \hspace{3mm} \ \{b_{ij},a^{(m)}_{kl}\}=0, \text{ for } i,j,k,l=1,\ldots,n, \text{ and } m=1,\ldots,p-1. 
\end{align}

Thus $\wUph$ is a quantization of the completed local Poisson algebra $\widehat{\mathcal{O}}_{{\mathcal M}_{dR}^{(p)}(u)}$. 

\begin{rmk}\label{modiloop}
The space ${\mathcal M}_{dR}^{(p)}(u)$, equipped with the above Poisson structure, is a Poisson submanifold of the dual of the modified loop algebra $L_r\frak{gl}_n$, see \cite{RSTS}. 
Here the Lie algebra $L_r\frak{gl}_n$ is induced by the split rational classical $r$-matrix $r=P_+-P_-$ on the ordinary loop algebra $L\frak{gl}_n=L_+\frak{gl}_n\oplus L_-\frak{gl}_n$. More explicitly, $L_r\frak{gl}_n$ is isomorphic to the direct sum of the original Lie algebra structure on $L_+\frak{gl}_n$ and the opposite bracket on $L_-\frak{gl}_n$
\[
 L_r\frak{gl}_n
 \cong
 L_+\frak{gl}_n\oplus
(L_-\frak{gl}_n)^{\operatorname{op}}.
\]
See \cite{BHH} for its explicit realization on the space of meromorphic connections. In particular, the minus sign in the residue matrix $-B$, as in our convention, accounts for the opposite bracket on $L_-\frak{gl}_n$, while the vanishing of the mixed bracket between $L_-\frak{gl}_n$ and $L_+\frak{gl}_n$ explains the decoupling of $B$ from the irregular coefficients $A_m$.
\end{rmk}

\subsection{Quantum Riemann-Hilbert-Birkhoff maps}\label{qRHBmap}

As a corollary of Theorem \ref{thmRLL}, we have
\begin{thm}
For any fixed $u\in\h_{\rm reg}$, the map 
\[
  \nu_\hbar(u):\mathcal U_\hbar^{(p)}(R)\longrightarrow \wUph~;~ L_i\mapsto S_i(u) \text{ for } i=1,\ldots,2p, \text{ and } K\mapsto e^{2\pi\mathi D_1},
\]
associating the quantum Stokes matrices $S_i$ and the formal monodromy matrix $e^{2\pi\mathi D_1}$ to the generators $L_i$ and $K$ respectively, is an associative algebra homomorphism. 
Taking the semiclassical limit (s.c.l) leads to the commutative diagram:
\[
\begin{CD}
U_\hbar^{(p)}(R)
@>{\nu_\hbar(u)}>>
\wUph
\\
@VV{s.c.l}V
@VV{s.c.l}V
\\
\widehat{\mathcal O}_{\mathcal M_B^{(p)}}
@>{\nu(u)^*}>>
 \widehat{\mathcal O}_{{\mathcal M}_{dR}^{(p)}}.
\end{CD}
\]
\end{thm}
\begin{proof}
We denote by ${\rm scl}(X)$ the image of an element $X\in \wUph$ under the semiclassical limit
\[ \wUph/\hbar{\wUph}
  \cong
  \widehat{\mathcal O}_{{\mathcal M}_{dR}^{(p)}}.
  \]
In particular, we have ${\rm scl}(\hbar e^{(m)}_{ij})=a^{(m)}_{ij}$ and ${\rm scl}(\hbar e_{ij})=b_{ij}$, for $m=1,\ldots,p-1$ and $i,j=1,\ldots,n$. Abusing notation slightly, we may write ${\rm scl}(\hbar T_{[m]})= A_m$ and ${\rm scl}(\hbar T)= B$, where $A_m$ and $B$ are coefficient matrices of the classical differential equation \eqref{heq}.

Let us first check the semiclassical limit of the formal solution of the (quantum) equation \eqref{hqeq} is the formal solution of the classical equation \eqref{heq}. Let us take the unique formal solution of the differential equation \eqref{hqeq}
\[
\widehat F(z)
=
\left(1+H_1z+H_2z^2+\cdots\right)
\cdot \exp\int
\left(
\frac{u}{z^{p+1}}
+\frac{D_p}{z^p}
+\cdots
+\frac{D_1}{z}
\right)dz,
\]
Following \eqref{Hrecursion}, the coefficients $H_m\in \Uph\subset\wUph$ are uniquely determined by the recursion \eqref{Hrecursion}. Taking the semiclassical limit on both sides of \eqref{Hrecursion} and setting 
\[h_m={\rm scl}(H_m), \quad d_i={\rm scl}(D_i), \quad A_m={\rm scl}(\hbar T_m), \quad B={\rm scl}(\hbar T)\in {\mathcal O}_{{\mathcal M}_{dR}^{(p)}}\otimes \End(\C^n),\] 
we get (under the convention $A_1=-B$, $h_0=1$ and $h_k=0$ for $k<0$)
\begin{equation}\label{cHrecursion}
  [u,h_{m+1}]
  =
  (m-p+1)h_{m-p+1}
  +
  \sum_{k=1}^{p}
  \left(
    h_{m-p+k}d_k
    -
    A_kh_{m-p+k}
  \right).
\end{equation}
The equations \eqref{cHrecursion}, for the $n\times n$ matrices $h_m$, $m\ge 1$, and the diagonal matrices $d_i$, $i=1,\ldots,p$, ensure that \[
\widehat f(z)
:=
\left(1+h_1z+h_2z^2+\cdots\right)\cdot
\exp\int
\left(
\frac{u}{z^{p+1}}
+\frac{d_p}{z^p}
+\cdots
+\frac{d_1}{z}
\right)dz,
\]
is the unique formal solution (under the initial condition $h_0=1$) of 
\begin{equation}
\frac{df}{dz} = \left(\frac{u}{z^{p+1}}+\frac{A_p}{z^p}+\cdots+\frac{A_2}{z^2}-\frac{B}{z}\right)\cdot f.\end{equation}

Note that taking the semiclassical limit is compatible with the order-$p$ Borel-Laplace resummation. Therefore, from ${\rm scl}(\widehat{H}(z))=\widehat{h}(z)=1+\sum_m h_mz^m$, we identity the semiclassical limit ${\rm scl}(H_d(z))$, of the Borel-Laplace $p$-sum $H_d(z)$ along an admissible direction $d$, with the formal Taylor series of the $p$-sum $h_d(z)$ in the coefficient parameters. Thus, the canonical solutions satisfy
\[f_d(z)={\rm scl}(F_d(z))\in {\widehat{\mathcal O}}_{{\mathcal M}_{dR}^{(p)}}\widehat{\otimes} \End(\C^n).\]
Here $f_d(z)$, as analytic functions on $(A_p,\ldots,A_2,B)\in\mathcal{M}_{dR}^{(p)}$, is regarded as its Taylor expansion at the origin $0\in \mathcal{M}_{dR}^{(p)}$. Therefore, the quantum and classical Stokes matrices, as the transition matrices between the solutions $F_d(z)$ and $f_d(z)$ respectively, relate to each other by
\[s_i={\rm scl}(S_i)\in {\widehat{\mathcal O}}_{{\mathcal M}_{dR}^{(p)}}\widehat{\otimes} \End(\C^n).\]

Finally, in terms of the generators, the semiclassical limit of the homomorphism
\[\nu_\hbar(u)~:~ U_\hbar^{(p)}(R)
\rightarrow
\wUph~;~ L_i\mapsto S_i \]
gives 
\[{\rm scl}(\nu_\hbar(u))~:~ \widehat{\mathcal O}_{\mathcal M_B^{(p)}}\rightarrow
 \widehat{\mathcal O}_{{\mathcal M}_{dR}^{(p)}};~ {\rm scl}((L_i)_{ab})\mapsto (s_i)_{ab}. \]
Here as in Section \ref{PoissonStokes}, ${\rm scl}((L_i)_{ab})$ are regarded as the coordinate functions on the $i$-th component of $(U_-\times U_+)^p\times \h$. Therefore, ${\rm scl}(\nu_\hbar(u))$ coincides with the pull-back of the classical RHB map 
\[\nu(u): (A_p,...,A_2,B)\in {\mathcal M}_{dR}^{(p)}(u)\rightarrow (s_1,...,s_{2p},d_1)\in\mathcal{M}_B^{(p)}\cong (U_-\times U_+)^{p}\times \h\]
associating the classical Stokes matrices to the differential equation \eqref{heq}. 
\end{proof}

\subsection{Poisson geometric nature of the classical Riemann-Hilbert-Birkhoff maps}\label{PoissonRHB}

Since the semiclassical limit of the associative algebra homomorphism $\nu_\hbar(u)$ automatically yields a Poisson map, it interprets the Poisson geometric nature of the classical RHB map in the following way. 

Let us first recall the results in \cite{Boalch2, Boalch4}.
Originally, the Poisson structure on the space of classical Stokes matrices
was induced from the irregular Atiyah-Bott construction \cite{Boalch4}, and admits an explicit form described by the quasi-Hamiltonian approach \cite{AMM}. That is $\mathcal{M}_B^{(p)}\cong \widetilde{\mathcal{M}}/GL_n$ is the quotient space of a quasi-Hamiltonian manifold $\widetilde{\mathcal{M}}$ described as follows. 

\begin{thm}[{\cite{Boalch4}}]\label{qHamiltonian}
Put
\[
 \widetilde{\mathcal M}
 =G\times(U_+\times U_-)^p\times\mathfrak h.
\]
Write a point as $(C,\Sigma_1,\ldots,\Sigma_{2p},\lambda)$, with
$\Sigma_{\rm odd}\in U_+$ and $\Sigma_{\rm even}\in U_-$. Then
$\widetilde{\mathcal M}$ is a complex quasi-Hamiltonian
$G\times H$-space for the action
\[
(g,t)\cdot(C,\Sigma_1,\ldots,\Sigma_{2p},\lambda)
 =(tCg^{-1},t\Sigma_1t^{-1},\ldots,t\Sigma_{2p}t^{-1},\lambda).
\]
Set $\epsilon=e^{\pi\mathi\lambda/p}$, and, for $1\leq j\leq p$, define
\begin{align*}
 D_j&=\epsilon^{-j}
 \Sigma_{2p+1-j}^{-1}\cdots \Sigma_{2p-1}^{-1}\Sigma_{2p}^{-1}C,\\
 E_j&=\epsilon^{j-2p}\Sigma_j\cdots \Sigma_2\Sigma_1\epsilon^{2p}C,
 \qquad D_0=E_0=C,
\end{align*}
and $D=D_p$, $E=E_p$. If
\[
 \mathcal D_j=D_j^*\theta,\qquad
 \mathcal E_j=E_j^*\theta,\qquad
 \overline{\mathcal D}=D^*\overline\theta,\qquad
 \overline{\mathcal E}=E^*\overline\theta,
\]
where $\theta$ and $\overline\theta$ are the left and right
Maurer-Cartan forms on ${\rm GL}_n$, then the quasi-Hamiltonian two-form is
\[
 \omega
 =\frac12(\overline{\mathcal D},\overline{\mathcal E})
 +\frac12\sum_{j=1}^{p}
 \bigl(
   (\mathcal D_j,\mathcal D_{j-1})
   -(\mathcal E_j,\mathcal E_{j-1})
 \bigr).
\]
\end{thm}
\begin{rmk}\label{opposite}
In the theorem, the convention is $\Sigma_{\rm odd}\in U_+$, whereas our initial anti-Stokes ray convention
requires the odd Stokes matrices in $U_-$. The two descriptions are
related by a cyclic relabelling
\[
s_i=\Sigma_{i+1} \text{ for } 1\le i<2p, \text{ and  }
s_{2p}=e^{-2\pi i\lambda}\Sigma_1 e^{2\pi i\lambda}.
\]
\end{rmk}
\begin{rmk}
In the case $p=1$, as pointed out in \cite{Boalch1, Boalch4}, the space of Stokes matrices equipped with the inherited Poisson structure is identified with the dual Poisson-Lie group ${\rm GL}_n^*$ \cite{LuW}. For general $p>1$, there is also a geometric construction of the quasi-Hamiltonian geometry due to Li-Bland-\v{S}evera \cite{LBS}.
\end{rmk}

Since $u$ is regular, there exists a unique formal power series $\hat{h}(z)\in {\rm End}(\mathbb{C}^n)\llbracket z\rrbracket$ with $\hat{h}(0)=1$, such that the formal gauge transform 
\begin{equation}\label{clsirrtype}
    \hat{h}(z)[A(z)]=\hat{h}(z)^{-1} A(z) \hat{h}(z)-\hat{h}(z)^{-1}\frac{d\hat{h}}{dz}(z)=\frac{u}{z^{p+1}}+\frac{d_p}{z^p}+\cdots+\frac{d_2}{z^2}+\frac{d_1}{z},
\end{equation}
where $A(z)$ is the coefficient matrix of \eqref{heq}, and $d_p,d_{p-1},...,d_1$ is a series of $n\times n$ diagonal matrices. The irregular part $\frac{u}{z^{p+1}}+\sum_{m=2}^p\frac{d_m}{z^m}$ is called the irregular type of equation \eqref{heq}. Theorem \ref{oneformal} is its quantum analog.

The $d_p,\ldots,d_2$, seen as functions on ${\mathcal M}_{dR}^{(p)}(u)$, are Casimirs, so the Poisson structure is fibrewise on each subspace ${\mathcal M}_{dR}^{(p)}(u,d_p,\ldots,d_2)$ with fixed irregular type. Then we have

\begin{thm}\cite{Boalch2}\label{kthm}
The Riemann-Hilbert-Birkhoff (also known as irregular Riemann-Hilbert) map \[\nu: {\mathcal M}_{dR}^{(p)}(u,d_p,...,d_2)\rightarrow \mathcal{M}_B^{(p)}~;~(A_p,...,A_2,B)\mapsto (\Sigma_1,...,\Sigma_{2p},d_1)\]
associating the classical Stokes matrices to the differential equation \eqref{heq}, is a generally local analytic Poisson isomorphism.
\end{thm}
 
\begin{rmk}
The original theorem in \cite{Boalch2} is given in the framed moduli space setting. After the framing $GL_n$ reduction, we get the present form. In particular, the Poisson bracket \eqref{Poisson1}-\eqref{Poisson2} on ${\mathcal M}_{dR}^{(p)}(u,d_p,...,d_2)$ is induced from the $GL_n$-reduction of the Poisson structure on the framed moduli space (extended coadjoint orbit), by the decouple lemma (Lemma 2.4 in \cite{Boalch2}, or Proposition 14(2) in \cite{Boalch4}).     
\end{rmk}
As a consequence, we have
\begin{cor}\label{twoPoisson}
The Poisson bracket on ${\mathcal M}_{B}^{(p)}=(U_-\times U_+)^p\times\mathfrak h$, induced from the ${\rm GL}_n$ reduction of the quasi-Hamiltonian two form $\frac{1}{2\pi\mathi}\omega$, is given by the matrix coordinate expressions \eqref{bracket1}-\eqref{6Pidentity} (provided the cyclic relabelling as in Remark \ref{opposite} is taken into account). 
\end{cor}

\subsection*{Acknowledgements}
\noindent It is a pleasure to thank Anton Alekseev for his valuable advice and incisive
comments. Several key points he raised during our discussions were essential
to the completion of this paper. The author is supported by the National Key Research and Development Program of China (No. 2021YFA1002000).

\Addresses
\end{document}